\documentclass[11pt]{amsart}
\pdfoutput=1
\usepackage{amsmath,amsfonts,latexsym,amsthm,amssymb,graphicx,multirow}
\usepackage[all]{xy}
\DeclareFontFamily{OT1}{rsfs}{}
\DeclareFontShape{OT1}{rsfs}{n}{it}{<-> rsfs10}{}
\DeclareMathAlphabet{\mathscr}{OT1}{rsfs}{n}{it}

\usepackage{ifpdf}
\ifpdf
   \pdfoutput=1
   \pdftrue
  \usepackage[pdftex]{hyperref}
\else
   \pdffalse
   \usepackage{cite}
\fi

\newtheorem{theorem}{Theorem}[section]

\newtheorem{prop}[theorem]{Proposition}
\newtheorem{claim}[theorem]{Claim}

{\theoremstyle{definition} }
{\theoremstyle{remark} \newtheorem{remark}[theorem]{Remark}
}

\newcommand{\Cbb}{{\mathbb{C}}}
\newcommand{\Pbb}{{\mathbb{P}}}
\newcommand{\cE}{{\mathscr E}}
\newcommand{\cL}{{\mathscr L}}
\newcommand{\cO}{{\mathscr O}}
\newcommand{\csm}{{c_{\text{SM}}}}
\newcommand{\oD}{{\overline D}}
\newcommand{\uD}{{\underline D}}
\newcommand{\uO}{{\underline O}}

\begin{document}

\begin{titlepage}
\begin{center}
\baselineskip=16pt
{\huge
New Orientifold Weak Coupling Limits \\
in F-theory\\
}
\vspace{2 cm}
\vfill
{\Large  Paolo Aluffi$^\clubsuit$  and Mboyo Esole$^\spadesuit$
  } \\
\vspace{.5 cm}

$^\clubsuit$
Mathematics Department, 
Florida State University,
Tallahassee FL 32306, U.S.A.\\

$^\spadesuit$Jefferson Physical Laboratory, Harvard University, Cambridge, MA 02138, U.S.A.\\
\end{center}

\vspace{2cm}
\begin{center}
{\bf Abstract}
\vspace{.5 cm}
\end{center}

{We present new explicit constructions of weak coupling limits of
F-theory generalizing Sen's construction to elliptic fibrations which
are not necessary given in a Weierstrass form. These new constructions
allow for an elegant derivation of several brane configurations that
do not occur within the original framework of Sen's limit, or which
would require complicated geometric tuning or break supersymmetry.
Our approach is streamlined by first deriving a simple geometric
interpretation of Sen's weak coupling limit.  This leads to a natural
way of organizing all such limits in terms of transitions from
semistable to unstable singular fibers.  These constructions provide a
new playground for model builders as they enlarge the number of
supersymmetric configurations that can be constructed in F-theory.  We
present several explicit examples for $E_8$, $E_7$ and $E_6$ elliptic
fibrations.}
 
 \vspace{1 cm}

 \begin{quote}
\begin{flushright}
{\it  He thought he saw a Platypus\\

Descending from a train:\\

He looked again, and found it was\\

A split D7-Brane.\\

`So sad, alas, so sad!' he said,\\

`That all has been in vain!'\footnote{From Matilde Marcolli,
{\em The Mad String Theorist's Song.}}
}
\end{flushright}
\end{quote}
  \vfill

\begin{flushleft}
$^\clubsuit${\tt aluffi@math.fsu.edu}\\ 
$^\spadesuit${\tt esole@physics.harvard.edu}\\
\end{flushleft}
\end{titlepage}

\tableofcontents


\newpage
\section{Introduction}\label{intro}

\bigskip

Type IIB string theory is an interesting framework for model building
in view of the possibility to stabilize moduli.  F-theory was introduced 
by Cumrun Vafa \cite{Vafa:1996xn} and provides an elegant
geometric formulation of non-perturbative IIB string theory in the
presence of 7-branes and D3 branes
\cite{Vafa:1996xn,Morrison:1996na,Morrison:1996pp}.  An F-theory
compactification requires an elliptic fibered space $Y\rightarrow B$
over the compact space $B$ on which the dual type IIB theory is
compactified.  The elliptic fiber is a two-torus whose complex
structure is interpreted as the type IIB axion-dilaton field. The
modular transformations of the elliptic fiber are a geometric
realization of type IIB S-duality transformations.  Around a singular
fiber, the monodromy of the axion-dilaton field characterizes the
non-perturbative $(p,q)$-seven-branes
\cite{Schwarz:1995dk, Schwarz:1995jq, Douglas:1996du,  
Johansen:1996am,   Gaberdiel:1997ud}.  An F-theory approach to 
type IIB model building brings more constraints coming from its
non-perturbative origin.  The type and number of branes, the gauge
group and the D3 tadpole are all determined by the geometry of the
elliptic fibration and particularly by the type of singular fibers.
The location of the 7-branes is given by the {\em discriminant locus}:
the locus of all singular fibers in the base.  Each irreducible
component of the discriminant locus corresponds to a $(p,q)$-brane. In
the case of F-theory compactified on elliptically fibered four-folds,
the D3 tadpole is given by the Euler characteristic of the four-fold
divided by~$24$ \cite{Sethi:1996es}.  For a review on F-theory type
IIB model building we refer to \cite{Denef:2008wq,Blumenhagen:2008zz}.
 
Several realistic features of particle phenomenology seem to emerge
naturally in the context of local models in F-theory where gravity is
decoupled \cite{Beasley:2008dc,Beasley:2008kw,Heckman:2009mn}.  For
reviews on local model building in F-theory, we refer to
\cite{Heckman:2009mn,Heckman:2008rb,Donagi:2008ca,Bourjaily:2009vf}.
The main shortcoming of local model building is that it does not keep
track of all the global consistency conditions required by the full
string theory. For example, the tadpole cancellation conditions are
usually not addressed.  Recently, there have been several attempts to
reconcile the results of local model building with all the global
constraints coming from F-theory on compact manifolds
\cite{Blumenhagen:2008zz,Andreas:2009uf,Blumenhagen:2009up,
Donagi:2009ra, Marsano:2009ym,Marsano:2009gv}.  In this context,
defining the weak coupling limit of F-theory is a crucial step. 
So far, Sen's weak coupling limit of F-theory was the only option
available~\cite{Sen:1997gv}.

The goal of this paper is to develop new ways of taking the weak
coupling limit of F-theory, generalizing Sen's technique to other
families of elliptic fibrations.  The results obtained here can also
be read `backwards' as providing new F-theory lifts for type IIB
orientifold compactifications. Interestingly, the new weak coupling
limits presented here avoid several inherent complications of Sen's
limit of a Weierstrass model.  In particular, we can avoid the
singularities of the D7 divisors that are characteristic of Sen's weak
coupling limit \cite{AlEs,Collinucci:2008pf,Braun:2008ua}.  We also
propose F-theory constructions that satisfy some important topological
relations required by the consistency of the duality between type IIB
and F-theory \cite{AlEs, Collinucci:2008pf}. This is done in
situations where the same relation would inevitably be violated or
would require turning on fluxes that would break supersymmetry in the
context of Sen's limit of a Weierstrass model. For example, we present 
F-theory engineering of brane-image-brane configurations that
do not break supersymmetry through the violation of D-term
constraints in contrast to a typical brane-image-brane configuration
resulting from the usual Sen's limit as explained in \S4.4 of
\cite{Collinucci:2008pf}.

The structure of this paper is as follows.  In the rest of the
introduction, we review the physics of Sen's weak coupling limit of
F-theory and we introduce the main ideas behind the constructions
presented in this paper. In particular, we discuss the geometrization
of Sen's limit and motivate the introduction of new families of
elliptic fibrations. As an illustration, we discuss the physics of
some explicit examples constructed in this paper that would suffer
from several difficulties in the usual Sen's limit.  In section
\ref{Sinfib}, we define the new families of elliptic fibrations we will 
consider in the paper and study systematically their singular fibers. In 
section \ref{Wcl} we derive the new weak coupling limits for the families
introduced in section \ref{Sinfib}. In section \ref{Tadrel} we prove a
compact generalization of Sethi-Vafa-Witten formula to compute the
Euler characteristic of an elliptic fibration in terms of its base,
and we analyze the topological relation induced by the matching
between the D3 tadpole in F-theory and in type IIB in the absence of
fluxes. Inspired by the physics, we can prove more general
relations, valid without any restriction on the dimension of the
elliptic fibration and without imposing the Calabi-Yau condition.  In
section \ref{morelimits} we provide a list of weak coupling limits for
$E_8$, $E_7$ and $E_6$ elliptic fibrations. Finally we conclude with a
discussion in section \ref{Final.Conclusion}.

\subsection{Sen's weak coupling limit}\label{Senswcl}
In the weak coupling limit of F-theory, one obtains a type IIB theory with 
with O7 orientifold planes in order to cancel the D7 tadpole.  
In the weak coupling limit of F-theory, Ashoke Sen
has provided the only systematic way to derive a type IIB orientifold
   limit of F-theory in the familiar case of elliptic fibrations given by
a Weierstrass model \cite{Sen:1997gv}.  For reasons that we will
explain later, we refer to the Weierstrass model as an
$E_8$-fibration. Its defining equation is
\begin{equation}\nonumber
E_8: y^2=x^3+ F x z^6 + G z^6,
\end{equation}
where $F$ and $G$ are respectively sections of ${\mathscr L}^4$ and
${\mathscr L}^6$, where ${\mathscr L}$ is a line bundle over the base
$B$. In the case of a Calabi-Yau elliptic fibration, we have
$c_1(B)=c_1(\mathscr{L})$. A Weierstrass model admits a discriminant
$\Delta=4 F^3+ 27 G^2$ whose vanishing locus determines the location 
of  singular fibers. There are essentially two types of singular
fibers: the general element of the discriminant locus $\Delta=0$ is a  
nodal curve; when $F=G=0$, the nodal curve specializes to a cuspidal
curve. Sen's limit is given by the expression
\begin{equation}\nonumber
F= -3 h^2 + C \eta,\quad  G=-2 h^3+ Ch\eta+ C^2 \chi. 
\end{equation}
At leading order, the discriminant and the $j$-invariant are 
\begin{equation*}
\Delta\sim h^2 (\eta^2+ 12 h\chi), \quad j\sim \frac{h^4 }
{C^2(\eta^2+ 12 h\chi)}.
\end{equation*}
One can construct a Calabi-Yau 3-fold which is a double cover of the base:
$$ 
X: \xi^2=h,
$$
and the involution is given by $\sigma:X\rightarrow X:\xi\mapsto -\xi$.
The orientifold involution is then  $\Omega \sigma (-)^{F_L}$, where 
$\Omega$ is the world-sheet orientation reversal and $F_L$ is the 
space-time fermion number from the left-moving sector of the world-sheet.
The orientifold plane in $X$ is $O:\xi=0$ while in the base it was given by
$\underline{O}:h=0$.  The D7 locus in the base was
  $\underline{D}:\eta^2+12 h \chi=0$ while in~$X$ it is given by
$D:\eta^2+12 \xi^2 \chi=0$.

The physics of Sen's limit has been explored recently in several
papers  \cite{AlEs,Braun:2008ua,Collinucci:2008pf, Brunner:2008bi,Collinucci:2008zs,Collinucci:2009uh}.  
We would like to emphasize the
following properties \cite{Collinucci:2008pf,AlEs,Braun:2008ua}:
\begin{itemize}
\item An orientifold seven plane is a perturbative object. It splits into 
two non-pertur\-bative $(p,q)$-branes as the coupling becomes
stronger \cite{Sen:1997gv}.
\item 
The locus of the D7 brane $D:\eta^2+12 \xi^2 \chi=0$ is singular along
the locus $\eta=\xi=0$ of double points and the singularity worsens on
the pinch locus $\eta=\xi=\chi=0$. Near each of these pinch points, we
can use $\eta,\chi$ and $\xi$ as local coordinate and the surface
wrapped by the D7 brane looks like a Whitney umbrella $D:\eta^2+12
\xi^2 \chi=0$. We will refer to such a singular D7-brane as a {\em Whitney 
D7-brane}. The presence of the singularities forces one to carefully define 
the Euler characteristic used to compute the induced D3 charge 
\cite{AlEs,Collinucci:2008pf}.  
\item
A brane-image-brane pair is just a specialization of a Whitney D7-brane 
$u^2-v^2 w=0$ when $w$ is a perfect square; if $12\chi=\psi^2$, then $D$ 
splits as $D=D_++D_-$ with $D_\pm :\eta\pm
\xi \psi=0$.  In the base, the equation of a Whitney D7-brane is
$\underline{D}: \eta^2+ 12 h\chi=0$, which specializes to
$\underline{D}:\eta^2+ h\psi^2=0$ for a brane-image-brane pair.
Interestingly, in the base, a brane-image-brane pair has the structure
of a Whitney umbrella ($u^2-v^2 w=0$ in $\Cbb^3$) whereas a Whitney
brane has the structure of a cone ($u^2-v w=0$ in $\Cbb^3$).

\item 
The D7 branes appearing in limit always have a double intersection
(or more generally an intersection with even multiplicity) with an  
orientifold seven plane. This double intersection property is a direct 
consequence of the Whitney umbrella geometry in Sen's limit 
\cite{Collinucci:2008pf,Braun:2008ua}. Indeed, as $\xi=0$, we get 
$\eta^2=0$.  More generally, it can be seen as a consequence of 
Dirac quantization in the case of O7$^-$ plane \cite{Collinucci:2008pf}.  
\item 
D3 branes being S-duality invariant, the D3 tadpole does not depend on
the string coupling. Therefore, one would expect to find a matching
between the F-theory and the type IIB prediction of the D3 tadpole. In
F-theory, the D3 tadpole is given by the Euler characteristic of the
elliptic fibration. In type IIB, the D3 tadpole admits contributions
coming from the induced D3 charge of seven branes and orientifold
planes. In the absence of fluxes, this leads to a {\em F-theory/Type
IIB tadpole matching condition} \cite{AlEs,Collinucci:2008pf}:
\begin{equation}\nonumber
2\chi(Y)\overset ?=  4\chi(O)+\chi(D).
\end{equation}
As stated, this condition is not verified in Sen's limit, and in other
similar situations.  In cases where the Euler characteristic of the
four-fold overshoots\footnote{\label{pos}In these considerations we
implicitly assume that $\dim B=3$ and $c_1(B)^3>0$, for simplicity.}
the contribution from orientifold planes and D7 branes, one could
compensate by turning on world-volume fluxes on the compact surfaces
wrapped by the seven branes \cite{Collinucci:2008pf}.  In the case of
Sen's limit, if the D7 were wrapping a smooth locus one would find an
induced D3 brane in type IIB that would on the contrary overshoot the
F-theory prediction. This mismatch could not be corrected by turning
on world volume fluxes, since they would only contribute positively to
the induced D3 tadpole and therefore worsen the discrepancy.  However,
in Sen's limit the D7 brane wraps a singular hypersurface, and the
correct notion of Euler characteristic for such a locus is open to
interpretation. As shown in \cite{AlEs}, the tadpole matching
condition is satisfied in Sen's limit for an appropriate Euler
characteristic $\chi_o$, where the singularities provide a
contribution taking care of the discrepancy.  This Euler
characteristic $\chi_o(D)$ was introduced in 
\cite{AlEs,Collinucci:2008pf}. It satisfies all the physics tests as 
shown in \cite{Collinucci:2008pf}. Mathematically it can be defined as 
the Euler characteristic of a normalization of the singular surface $D$,
corrected by the contribution from the pinch locus \cite{AlEs}. This
can be understood as a generalization of the stringy Euler
characteristic \cite{AlEs}.
\item 
In the context of Sen's weak coupling limit of a Weierstrass model,
there is in general only one D7 brane located on an irreducible
locus. 
There are also smooth specializations, such as the brane-image-brane pairs;
these configurations naturally satisfy the double intersection property of 
7-branes and orientifold since the brane and its image brane always meet 
on the orientifold. However, without fluxes, they always leads to a type IIB
D3 induced tadpole which is lower (for a base $B$ of dimension~$3$,
and $c_1(B)^3>0$) than the one expected from F-theory. Therefore, they
always require the presence of world volume fluxes.  Since the world
volume fluxes also contribute to the D5 brane charges but are odd
under the orientifold involution, they gives different central charges
for the brane on $D_+$ and its image-brane on $D_-$. This leads to a
breaking of supersymmetry since the D-term constraints are not
respected \cite{Collinucci:2008pf}.
\end{itemize}

\subsection{The need for new weak coupling limits}\label{needfornew}
The previous review of the physics of Sen's weak coupling limit of a
Weierstrass model (an $E_8$-elliptic fibration) immediately raises the
question of its uniqueness, for the following reasons. Although it is
well appreciated that F-theory adds severe constraints to type IIB
theory, Sen's limit reduces very drastically the number of allowed
configurations.  It also implies that some very simple configurations
in type IIB (like the brane-image-brane configuration) are generally
non-supersymmetric.  The presence of the Whitney umbrella singularity
makes it more difficult to geometrically engineer specific
configurations of intersecting branes in the context of Sen's limit.
We consider all these limitations as an open invitation to find new
weak coupling limits.  The non-uniqueness of Sen's weak coupling limit
can already be appreciated within the original framework of a
Weierstrass model.  Consider for example the limit
\begin{equation}\nonumber
F=-3 h^2+ C\eta, \quad G=-2 h^3 +C(h\eta+\chi).
\end{equation}
It leads to 
\begin{equation}\nonumber
\Delta\sim  C h^3 \chi, \quad j \sim \frac{h^3}{C\chi}.
\end{equation}
This corresponds to a bound state of an orientifold $O:\xi=0$ with a
brane-image-brane on top of it ($D_\pm :\pm \xi=0$) together with a
brane on a smooth locus $D: \chi=0$. Since $D_\pm $ have the same
Euler characteristic as the orientifold plane, the F-theory-type- IIB
D3 tadpole matching gives in the absence of fluxes:
\begin{equation}\nonumber
2\chi(Y)\overset{?}{=}6\chi(O)+\chi(D).
\end{equation}
This relation does not hold.  We also note that $D$ does
 not have a double intersection with the orientifold plane. Each
time it intersects the orientifold plane, it also intersects the two
D7 branes located on $O$, in total we have an odd number of
D7 contributions.

The study of more sophisticated configurations within the Weierstrass
model quickly leads to complicated expressions that seem completely
unmotivated and arbitrary. This seems to indicate that we might win
some intuition and simplicity by considering new weak coupling limits
that do not arise from smooth Weierstrass models. In the
context of Sen's weak coupling limit of F-theory with Weierstrass
models, new methods with direct applications to model building have
been developed recently by Andres
Collinucci\cite{Collinucci:2008zs,Collinucci:2009uh}.  The same ideas
can be applied to the new weak coupling limits of F-theory we
construct in this paper.

\subsection{Moving away from Weierstrass models}

A smooth elliptic curve is a curve of genus one (with a marked point).
While curves of genus
zero are all isomorphic to a two sphere (a $\mathbb{CP}^1$), curves of
genus one are classified by a complex parameter which controls their
complex structure. This is expressed by the $j$-invariant of the
elliptic curve. Two curves of genus one are isomorphic if and only if
they have the same $j$-invariant.  The study of elliptic curves and
elliptic fibrations is often simplified by the following two
properties:
\begin{itemize}
\item Every smooth elliptic curve is isomorphic to an elliptic curve
written in Weierstrass form.
\item Every elliptic fibration is birationally equivalent to a
Weierstrass model.
\end{itemize}
When we consider an elliptic fibration which is not in Weierstrass
form, it can be useful to know exactly the Weierstrass form of a
birationally equivalent elliptic fibration. This is because there are
important properties of the elliptic fibration that could be computed
more simply in certain birationally equivalent elliptic fibrations,
and particularly in a birationally equivalent Weierstrass model.  
Two important examples are the location of the singular fibers and the
$j$-invariant of an elliptic fiber, which may be computed by using the 
Weierstrass model. Because of this, it is often believed that it is in fact 
enough to consider only Weierstrass models in F-theory. 
However, this can only be done at the price of a desingularization
process that may be difficult to handle. The Weierstrass model {\em per 
se\/} does not preserve all physically relevant quantities of a fibration:
for example, the type of singular fibers of a fibration is only recoverable
from the Weierstrass model through a local analysis of the equation
and of its discriminant (by Tate's algorithm, see \cite{Bershadsky:1996nh}),
not directly from the model itself.

In the case where the elliptic fibration is a four-fold, the topology
of the singular fibers plays a central role in F-theory since it
determines the Euler characteristic of the elliptic fibration and
therefore the F-theory D3 tadpole.  The Euler characteristic of an
elliptic fibration gets contributions only from singular fibers. A
birational transformation can change the type of singular fibers and
therefore change the Euler characteristic of the elliptic
fibration\footnote{While smooth birationally equivalent Calabi-Yau
varieties have the same Euler characteristic, \cite{MR1714818}, and
even commensurable Chern classes, \cite{math.AG/0401167,MR2280127}, we
note that even in the Calabi-Yau case the Weierstrass model of a
smooth elliptic fibration may be singular; the Euler characteristic
should be expected to change in this case.}.  For compactifications on
four-folds, this leads to a different prediction for the D3 tadpole
and therefore it changes the physics.

Once we consider elliptic fibrations which are not in Weierstrass
form, their Euler characteristic can be lower\footnote{See  
footnote~\ref{pos}.}  than the Euler characteristic of a smooth Weierstrass 
model defined on the same basis, with respect to the same line bundle.  
In particular, it is possible to accommodate
brane-image-brane configurations that do not require turning on fluxes
to satisfy the F-theory-type IIB tadpole matching condition. It
follows that such brane-image-brane configurations would not break
supersymmetry in contrast to those obtained as special configurations
of Sen's limit in a Weierstrass model. This illustrates the fact that
restricting ourselves to the Weierstrass model reduces artificially
the number of consistent compactifications allowed in F-theory.  As an
illustration, we will show that for the fibrations that we consider in
this paper there are in general only six possible configurations of
branes wrapping smooth divisors in a Calabi-Yau and satisfying the
F-theory-type-IIB tadpole matching condition. Only one of these six
configurations is in a Weierstrass model, and this configuration is
not yet realized as a weak coupling limit of F-theory. Moreover, the
allowed configurations reduce to three if we do not insist on the
Calabi-Yau condition, and none of these three configurations is
realized in a Weierstrass model. This is reviewed further in the
introduction and proved in \S\ref{Class.smooth}.

\subsection{Geometrization of Sen's weak coupling limit}

In our investigation of weak coupling limits, a geometric picture
quickly emerges.  It is based on the classification of semistable and
unstable singular fibers of an elliptic fibration.  In the case of
curves of genus one, semistable and unstable curves are simply defined
by the two possible behaviors of the $j$-invariant: it can be infinite
or undetermined of type $\frac{0}{0}$. When it is infinite, we have a
(singular) {\em semistable fiber} and when it is undetermined, it
corresponds to an {\em unstable fiber}.

Since the $j$-invariant is a birational invariant, it can be computed
without any loss of generality in a Weierstrass model birationally
equivalent to the elliptic fibration we consider. The singular fibers
are located on the discriminant $\Delta=4 F^3+ 27 G^2$. The smooth
fibers are such that $\Delta\neq 0$.  Since the singular fibers are
characterized by $\Delta=0$ and the j-invariant is
\begin{equation}\nonumber
j\sim \frac{F^3}{\Delta},
\end{equation}
we see that the semistable singular fibers are characterized by
$\Delta=0, F,G\neq 0$ while the singular unstable fibers are given by 
$F=G=0$.  It follows that an unstable fiber is a cuspidal curve, and
the set of unstable fibers ($F=G=0$) is the closure of the cuspidal
locus. Its defining characteristic is that this is the locus of fibers
having an undefined $j$-invariant of type $\frac{0}{0}$.  They consist
of cuspidal curves and their specializations. The general semistable
singular fiber is a nodal curve, but it can specialize to other types
of curves that are not a specialization of a cusp.

\begin{table}
\begin{center}
\begin{tabular}{||c|c|c||}
\hline
\hline
$j$-invariant & (elliptic) curve & condition \\
\hline
\hline
$j< \infty$ & smooth & $\Delta\neq 0$\\
\hline
$j=\infty$ & singular and  semistable & $\Delta=0, F, G\neq 0$\\
\hline
$j=\frac{0}{0}$ &singular and  unstable & $F=G=0$ \\
\hline 
\end{tabular}
\end{center}
\caption{Types of fibers of elliptic fibrations.}
\end{table}

In terms of the complex structure of the elliptic curve, the
$j$-invariant is expressed as
\begin{equation}\nonumber
j(q)=744+ \frac{1}{q}+\sum_n d_n q^n, \quad  q=\exp(2\pi i \tau), 
\quad d_n \in \mathbb{N}.
\end{equation}
Here $\tau=a +i e^{\varphi}$ determines the complex structure of the
elliptic curve. In F-theory, it is the axion-dilaton field. In
particular, its imaginary part, namely $e^\varphi=\frac{1}{g_s}$ is
the inverse of the string coupling $g_s$. The module of $q=\exp(2\pi i
\tau)$ is $|q|=g_s^{2\pi}$.  It follows that when we go to the weak
coupling limit ($g_s\rightarrow 0$), the $j$-invariant is dominated by
its pole $\frac{1}{q}$ and goes to infinity.  However, as we will see
in \S\ref{stratecomm}, as we reach the orientifold plane the
$j$-invariant is undetermined, of the form $\frac{0}{0}$.  It follows
that {\em the weak coupling limit singles out semistable singular
fibers away from the orientifold and specializes to an unstable fiber
on the orientifold locus.}  A weak coupling limit analogous to Sen's
always defines a transition from a semistable singular fiber (away
from the orientifold) to an unstable singular fiber (on the
orientifold plane).

\begin{center}
\begin{tabular}{llll}
{
\begin{tabular}{l}
\\
Sen's limit:\\
\\
\end{tabular}
}
&
{
\begin{tabular}{||c||}
\hline 
semistable fiber\\
\hline
 ($j=\infty$)\\
away from the orientifold\\
\hline
\end{tabular}
}
&

{
\begin{tabular}{c}
\\
\huge{ ${ \longrightarrow}$}\\
\end{tabular}
}

&
{
\begin{tabular}{||c||}
\hline 
unstable fiber\\
\hline
 ($j=\frac{0}{0}$)\\
on  the orientifold\\
\hline

\end{tabular}
}
\end{tabular}
\end{center}

\subsubsection{Geometric classifcation of weak coupling limits }

The observation that a weak coupling limit is characterized by a
 transition from a semistable singular fiber to an unstable one is the
 basis for our strategy to construct weak coupling limit of a given
 elliptic fibration:
\begin{itemize}
\item We choose a family of elliptic fibrations. 
\item We classify all singular fibers of the elliptic fibration. 
\item We determine those that are semistable and those that are unstable.  
\item Each specialization ``semistable $\rightarrow $ unstable" leads
to a family potentially yielding a weak coupling limit. 
\end{itemize}

\subsection{Consistency check}\label{conscheck}

In our investigation of orientifold weak coupling limits of
$F$-theory, we will always check if the following properties are
realized:

\begin{itemize}
\item 
Existence of a smooth Calabi-Yau double cover of the base. This will
be done by restricting the orientifold divisor $h=0$ to be a section
of an appropriate line bundle over the base \cite{Sen:1997gv}.
However, we will work
with elliptic fibrations of arbitrary dimension and we will not
restrict ourself to Calabi-Yau spaces. We will denote by $O$ the
orientifold plane and by $D_j$ the D7 branes\footnote{To avoid an
awkward notation, we will refer to these as {\em D7\/} branes even
when working over a base of arbitrary dimension.}. When we refer to
their location in the base rather than in the double cover, we
underline ($\underline{O}, \underline{D}_j$) their notation.
\item Vanishing of the D7 tadpole: the sum of class of the orientifold
planes is 8 times the sum of all the D7 branes when computed in the
double cover:
\begin{equation}\nonumber
\label{D7tadpole}
8[O]=\sum_i [D_i].
\end{equation}
\item 
F-IIB matching of the D3 tadpole in the absence of fluxes.  As a
non-trivial check of the consistency of the new weak coupling limits,
we study the matching of the F-theory and type IIB D3 tadpole in the
absence of fluxes. This leads to a non-trivial topological relation
among the Euler characteristics of the varieties involved in the
construction \cite{AlEs,Collinucci:2008pf}:
\begin{equation}\nonumber
2\chi(Y)= 4 \chi(O)+\sum_j \chi(D_i).
\end{equation}
In particular, in the presence of brane-image-branes, this condition
will ensure that there is no need to add brane fluxes and therefore we
avoid the usual breaking of supersymmetry by a violation of the D-term
constraints. When a D7 brane wraps a singular variety, we use the
Euler characteristic $\chi_o$, that is the Euler characteristic of its
normalization corrected by a contribution from the pinch locus.
\item 
D7 branes always intersect the orientifold plane with even multiplicity 
in order to avoid a violation of the Dirac quantization \cite{Collinucci:2008pf}.
We will not impose this on the construction, but we will observe that it
always occurs.
\end{itemize}

\subsection{Topological results inspired by tadpole relations }
An interesting aspect of the analysis of the weak coupling limit is
that the Calabi-Yau condition does not play any special role, nor the
dimension of the base of the elliptic fibration.  Moreover,
topological relations like the F-theory-type-IIB D3 tadpole matching
conditions are much more general than what we can expect from its
F-theory origin. Indeed, the relation obtained for Euler
characteristics can be lifted to a relation valid at the level of
Chern classes. For an elliptic fibration $\varphi:Y\rightarrow B$ and
a double cover of the base $\rho: X\rightarrow B$, we have
\begin{equation}\nonumber
2 \varphi_* c(Y)= 4 \rho_* c(O)+\sum_j \rho_* c(D_i),
\end{equation}
In this relation, the push-forwards bring the classes to the homology
of $B$.  This relation was already proven for Sen's weak coupling
limit in \cite{AlEs}.  In the case $D_i$ is a Whitney umbrella,
$\rho_* c(D_i)$ is understood as the Chern class of the normalization
of $D$ corrected by the Chern class of the pinch locus properly
pushed-forward to the homology of $B$.
 
The more refined relations involving Chern classes are actually
easier to prove than the relations at the level of the Euler
characteristic, which end up being direct consequences of the Chern
classes relations.  In this context the Sethi-Vafa-Witten relation
that gives the Euler characteristic of the elliptic fibration in terms
of the intersection theory of the base follows from a simple Chern
class relation.  For an elliptic fibration $\varphi: Y\rightarrow B$
we get
\begin{equation}\nonumber
\varphi_* c(Y)=(a+1) c(Z),
\end{equation}
where $Z$ is a divisor of a suitable class, closely related to the
fibration, and $a$ is the number of distinguished sections of the
fibration (see Remark~\ref{SVWrem}).  This formula is proved in
Theorem~\ref{SVWup} for $E_8$, $E_7$ and $E_6$ elliptic fibrations,
for which we have respectively $a=1,2,3$. The Sethi-Vafa-Witten
relation of the Euler characteristic is retrieved by simply evaluating
the top Chern class of the previous equation.

\subsection{$E_8, E_7, E_6$ families of elliptic fibrations}
To be specific, we will study in this paper the weak coupling limits
of $E_8$, $E_7$ and $E_6$ elliptic fibrations.  Each of these three
families can be defined over an arbitrary smooth base $B$ endowed with
a line bundle~$\cL$.  These are the most common families studied in
the literature and generalize the traditional Legendre, Jacobi and
Hesse family of elliptic fibrations.
A large class of such Calabi-Yau fourfolds was obtained and their
topological numbers computed in \cite{Klemm:1996ts} in the context of
Calabi-Yau hypersurfaces of weighted projective spaces.  They also
appear in the context of type IIB with Calabi-Yau three-folds
\cite{Klemm:1995tj,Andreas:1999ty} and in F-theory with discrete
torsion and reduced monodromy \cite{Berglund:1998va}.  In our
analysis, we do not restrict the dimension of the base nor do we
impose the Calabi-Yau condition.

The $E_8$ elliptic fibration is the usual Weierstrass model written as
a sextic in a weighted projective $\mathbb{P}^2_{1,2,3}$ bundle; the  
$E_7$ and $E_6$ fibrations are written respectively as a quartic in a
$\mathbb{P}^2_{1,1,2}$ bundle and a cubic in a $\mathbb{P}^2_{1,1,1}$
bundle.  They inherit their names from the theory of Del~Pezzo
surfaces \cite{Klemm:1996hh}.  In a strict sense, $E_8$, resp.~$E_7$,
$E_6$ singular fibers (types~$\text{IV}^*$, resp.~$\text{III}^*$,
$\text{II}^*$ in Kodaira's terminology) do not appear in the
corresponding fibrations; they appear in specializations of these
families, after a standard resolution process.  

\begin{remark}\label{fibcontext}
In this paper we are mostly interested in studying
the actual fibers of the maps realizing our fibrations, before any
resolution process: indeed, most of the families we consider will have
a smooth total space, and in any case we find that the strict notion
of fiber that we use suffices for our general strategy producing weak
coupling limits.
See~\S\ref{Sinfib}, and in particular Tables~\ref{table.E7SingFib}
and~\ref{table.E6SingFib}, for a list of the singular fibers that do
appear in the main fibrations studied here.
\end{remark}

Families of elliptic curves have been studied since at least as far
back as the 18th century.  Among
the most famous ones we have the Legendre, Jacobi and Hesse
families. They admits the following affine defining equations (see
chapter 4 of \cite{MR2024529}):
\begin{align*}
&\text{Legendre family}:& y^2 - x(x-1)(x-\lambda ) &=0,\\
&\text{Jacobi family}:& y^2 - x^4 -2 \kappa x^2  -1 &=0,\\
&\text{Hesse family}:&  y^3 +x^3+ 1 -3  \mu x y &= 0.
\end{align*}
These three families are toy models for the elliptic fibrations we
will consider in this paper. They easily illustrate the fact that
elliptic fibrations can have very different singular fibers.  The
Legendre family admits two singular fibers, one at $\lambda=0$ and
another at $\lambda=1$. They are both nodal curves. The Jacobi family
admits two singular fibers at $\kappa=\pm 1$, each corresponding to
the union of two conics intersecting transversally at two different
points.  Finally, the Hesse family has one singular fiber at $\mu=1$
and corresponding to 3 lines intersecting each other at one point and
forming a triangle.

The study of $E_8$, $E_7$ and $E_6$ Del Pezzo surfaces shows that they
admit elliptic fibrations with fibers of type
\cite{Klemm:1996hh,Klemm:1996ts}:
\begin{align*}
& E_8\ \text{fiber}: & y^2 +x^3 + z^6 + \lambda x y z&=0,\\
& E_7\ \text{fiber}:& y^2 + x^4 + z^4 -2\kappa x y z &=0,\\
& E_6\  \text{fiber}:&  y^3 +x^3+ z^3 -3 \mu x y z&= 0.
\end{align*}
The previous equations are respectively written as a sextic in
$\mathbb{P}^2_{1,2,3}$, a quartic in $\mathbb{P}^2_{1,1,2}$ and a
cubic in $\mathbb{P}^2_{1,1,1}$ where the indices refer to the weights
of $[z,x,y]$.

More generally, for us an $E_8$, $E_7$, or $E_6$ elliptic fibration is
a smooth elliptic fibration written respectively as as a sextic in a
$\mathbb{P}^2_{1,2,3}$-bundle, a quartic in a
$\mathbb{P}^2_{1,1,2}$-bundle, or a cubic in a
$\mathbb{P}^2_{1,1,1}$-bundle, and starting respectively as $y^2+x^3$,
$y^2-x^4$, $y^3+x^3$. They can always be put in the form:
\begin{align*}
& E_8\ \text{fibration} :&\quad y^2 +x^3 +  f  x z^4 + g z^6 &=0,\\
& E_7\ \text{fibration} :&\quad y^2 - x^4 - e x^2 z^2 - f x z^3 - g z^4  &=0,\\
& E_6\ \text{fibration}:&\quad  y^3 +x^3-  d x y z -  e x z^2 - f y z^2 
- g z^3 &= 0,
\end{align*}
where in each cases, in order to balance the equation, the
coefficients are sections of appropriate power of the line bundle
${\mathscr L}$ over the base $B$.  We will formalize this more
carefully in~\S\ref{Sinfib}.  The Legendre, Jacobi and Hesse families
are respectively specializations of the $E_8$, $E_7$ and $E_6$
elliptic fibrations.   
As reviewed in  \S\ref{Senswcl}, an $E_8$ elliptic fibration has essentially two types 
of singular fibers.  The general element of the discriminant locus $27 f^2+4 g^3=0$ 
corresponds  to a nodal curve, which specializes to a cuspidal curve when $f=g=0$. 
In Kodaira notation they are respectively singular fibers of type $I_1$ and $II$.   
$E_7$ and $E_6$ elliptic fibrations have a much richer spectrum of singular fibers. 
They are determined in \S\ref{Sinfib}  and summarized in tables \ref{table.E7SingFib} 
and \ref{table.E6SingFib}.  

\begin{center}
\begin{table}[hbt]
\begin{tabular}{|c|l|c|c|}
\hline 
\hline
Nane & \multicolumn{1}{|c|}{Conditions}  & 
$\begin{matrix}
\text{Type}\\
\text{ of singular fiber}\\
\text{(Kodaira notation)}\end{matrix}$ &  \\
\hline
\hline
$\Delta_{112}$ &  
$ 4e^3f^2+27f^4-16e^4g+128e^2g^2-144egf^2-256g^3=0
$
 & $I_1$ & 
$\begin{matrix}
 \includegraphics[scale=.2]{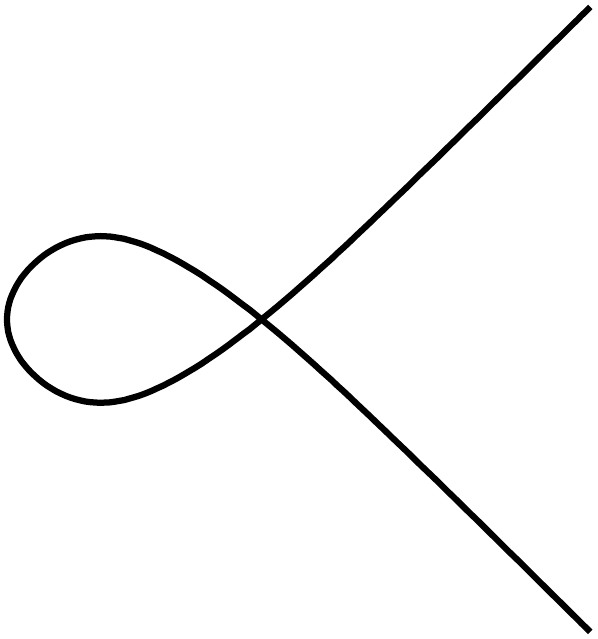}
 \end{matrix}$
 \\
\hline
$\Delta_{13}$ &$ e\ne 0 ,\quad 8e^3+27 f^2 = 0, \quad e^2+12 g=0$ & $II$  &
$\begin{matrix}
 \includegraphics[scale=.2]{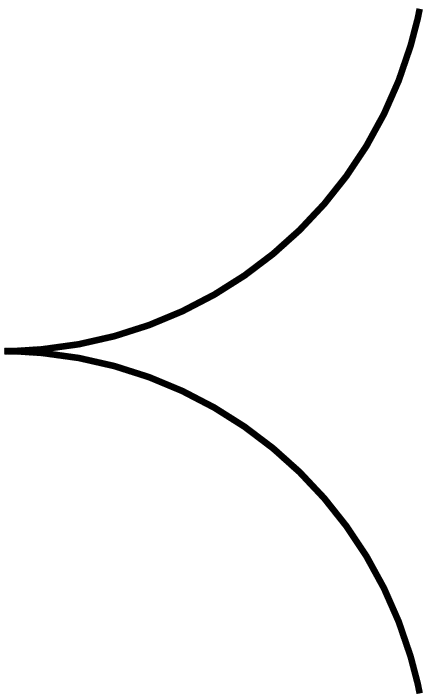}
 \end{matrix}$
  \\
\hline
$\Delta_{22}$ &$g\ne 0 ,\quad f=0 ,\quad e^2=4g$ & $I_2$  & 
$\begin{matrix}
\includegraphics[scale=.2]{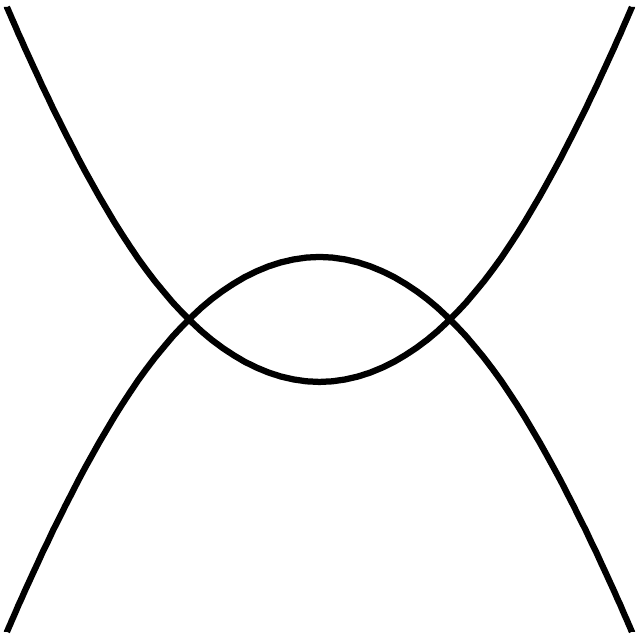}
\end{matrix}$
 \\
\hline
$\Delta_{4}$ &$ e=f=g=0 $ & $III$ & 
$\begin{matrix}
\includegraphics[scale=.2]{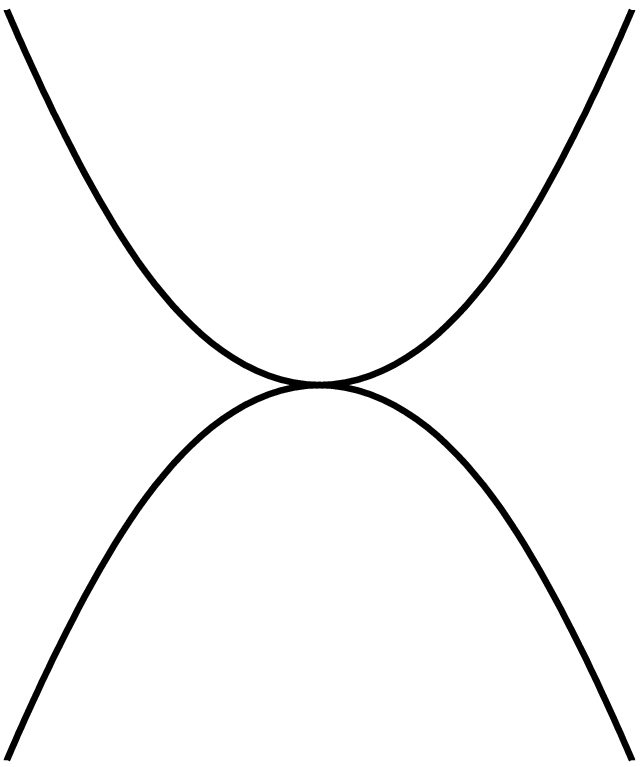}
\end{matrix}$
 \\
\hline
\end{tabular}
\caption{Singular fibers of $E_7$ elliptic fibrations. 
\label{table.E7SingFib}}
\end{table}
\end{center}

\begin{center}
\begin{table}[hbt]

\begin{tabular}{||l|l|c|c||}
\hline 
\hline
Name &\multicolumn{1}{|c|}{Conditions} & 
$\begin{matrix}
\text{Type of singular fiber}\\
\text{(Kodaira notation)}\end{matrix}$ &  \\
\hline
\hline
$\Delta_N$ & $\Delta=0$
 & $I_1$ &
$\begin{matrix}
\includegraphics[scale=.2]{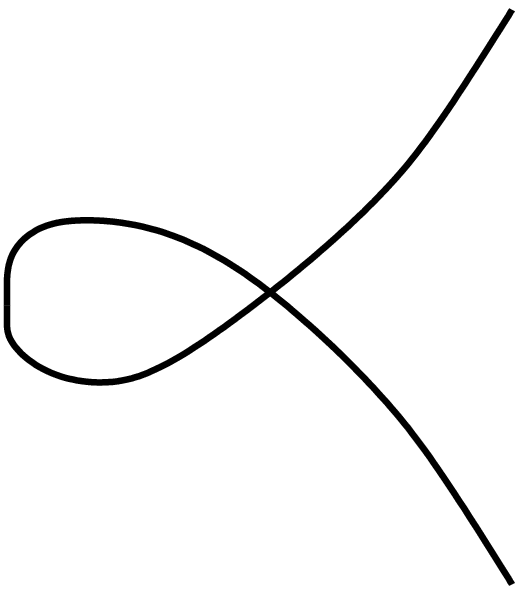}
\end{matrix}$
 \\
\hline
$\Delta_C$ &$
F=G=0$ & $II$ &  
$\begin{matrix}
\includegraphics[scale=.2]{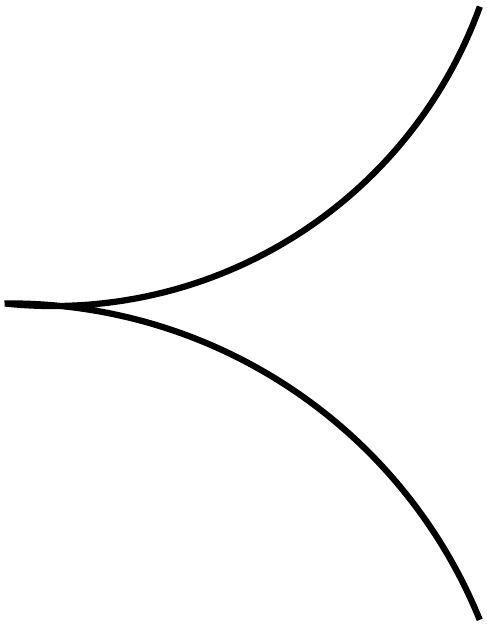}
\end{matrix}$
\\
\hline
$\Delta_Q $ &$ f =\rho e,\quad g =\frac{1}{27}d(9\rho^2 e-d^2)$  & $I_2$  &  
$\begin{matrix}
\includegraphics[scale=.2]{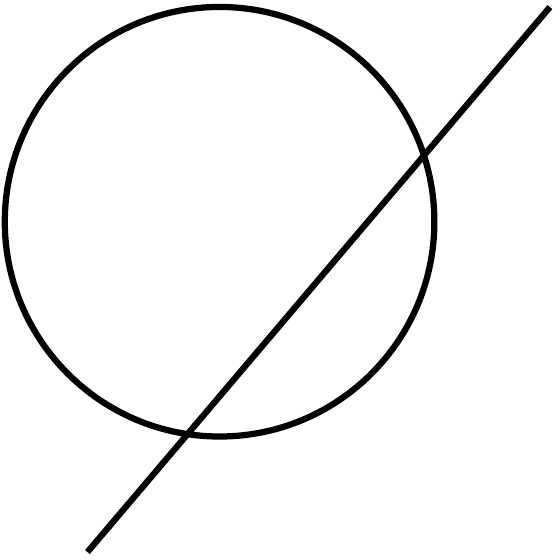}
\end{matrix}$
\\
\hline
$\Delta_P$ &$
f =\rho e, \quad 
e=\frac{1}{4}\rho d^2 ,\quad 
g =\frac{5}{108} d^3
$ & $III$ & 
$\begin{matrix}
\includegraphics[scale=.2]{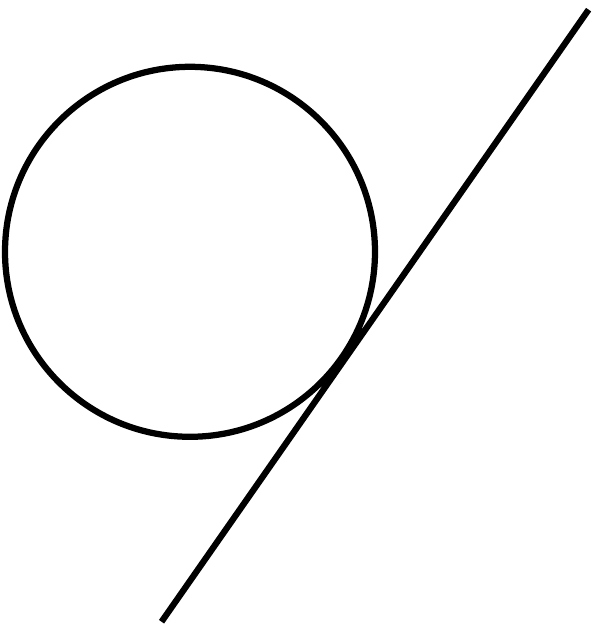}
\end{matrix}$
\\
\hline
$\Delta_T $ &$
e =f=0, \quad 
g =-\frac{1}{27}d^3
$
 & $I_3$ & 
 $\begin{matrix}
 \includegraphics[scale=.2]{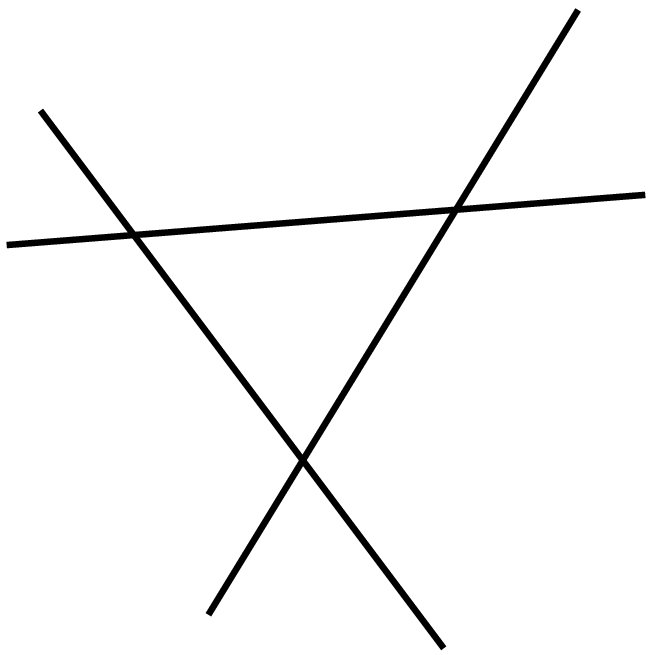}
\end{matrix}$
 \\
\hline
$\Delta_S$ &$\aligned   d=e=f=g=0 \endaligned $ & $IV$ & 
$\begin{matrix}
\includegraphics[scale=.2]{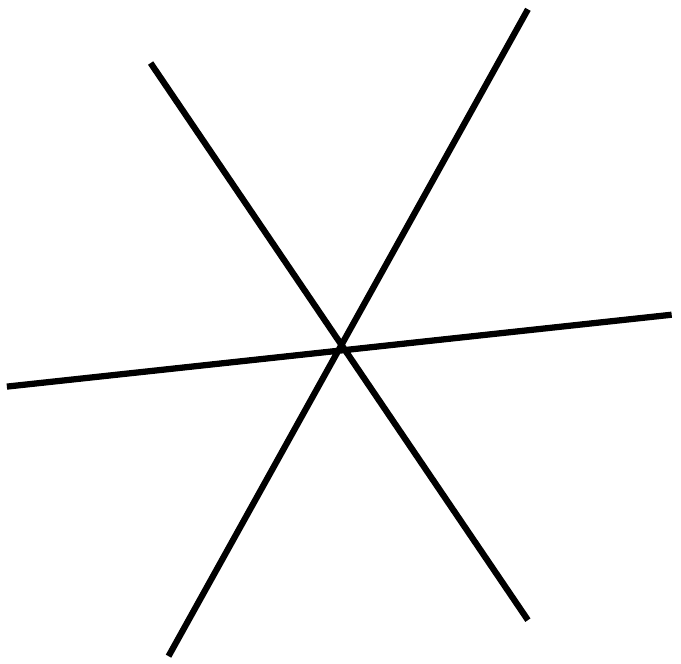}
\end{matrix}$
\\
\hline
\hline
\end{tabular}
\caption{Singular fibers of $E_6$ elliptic fibrations. 
Here $\rho^3=1$, 
$F=-3ef+\frac 92 gd-\frac 1{48} d^4$,
$G=-e^3-f^3+\frac{27}4 g^2+\frac 34 efd^2-\frac 58 gd^3-
\frac 1{864} d^6$, and $\Delta=4 F^3+27 G^2$.  \label{table.E6SingFib}
}
\end{table}
\end{center}

\subsection{Examples of new weak coupling limits}\label{IntroExamples}
We will present here some interesting configurations that would be
non-generic and not supersymmetric in the usual Sen's limit with a
smooth Weierestrass model, but which occur naturally in $E_7$
and $E_6$ fibrations.

\subsubsection{A  supersymmetric brane-image-brane 
configuration\label{IntroE71}}
In the $E_7$ case, we have 4 types of singular fibers.  In view of the
defining equation of an $E_7$ fibration, it is easy to see that the
singular fibers correspond to choices of $e$, $f$, $g$ such that the
quartic polynonial $Q(x)=x^4+ e x^2 z^2 + f x z^3 + g z^4$ admits
multiple roots.  All cases can be expressed in terms of partitions of
$4$. We will see in \S\ref{Sinfib} that there are in total 4 types of
singular fibers among which two are semistable (an ordinary nodal
curve ($\Delta_{112}$), and the transverse intersection of two conics
($\Delta_{22}$)) and two are unstable (a cusp ($\Delta_{13}$) and two
conics tangent at a point (${\Delta_4}$)).  We define a weak coupling
limit by a transition $\Delta_{22}\rightarrow \Delta_{4}$:
\begin{center}
 \includegraphics[scale=0.3]{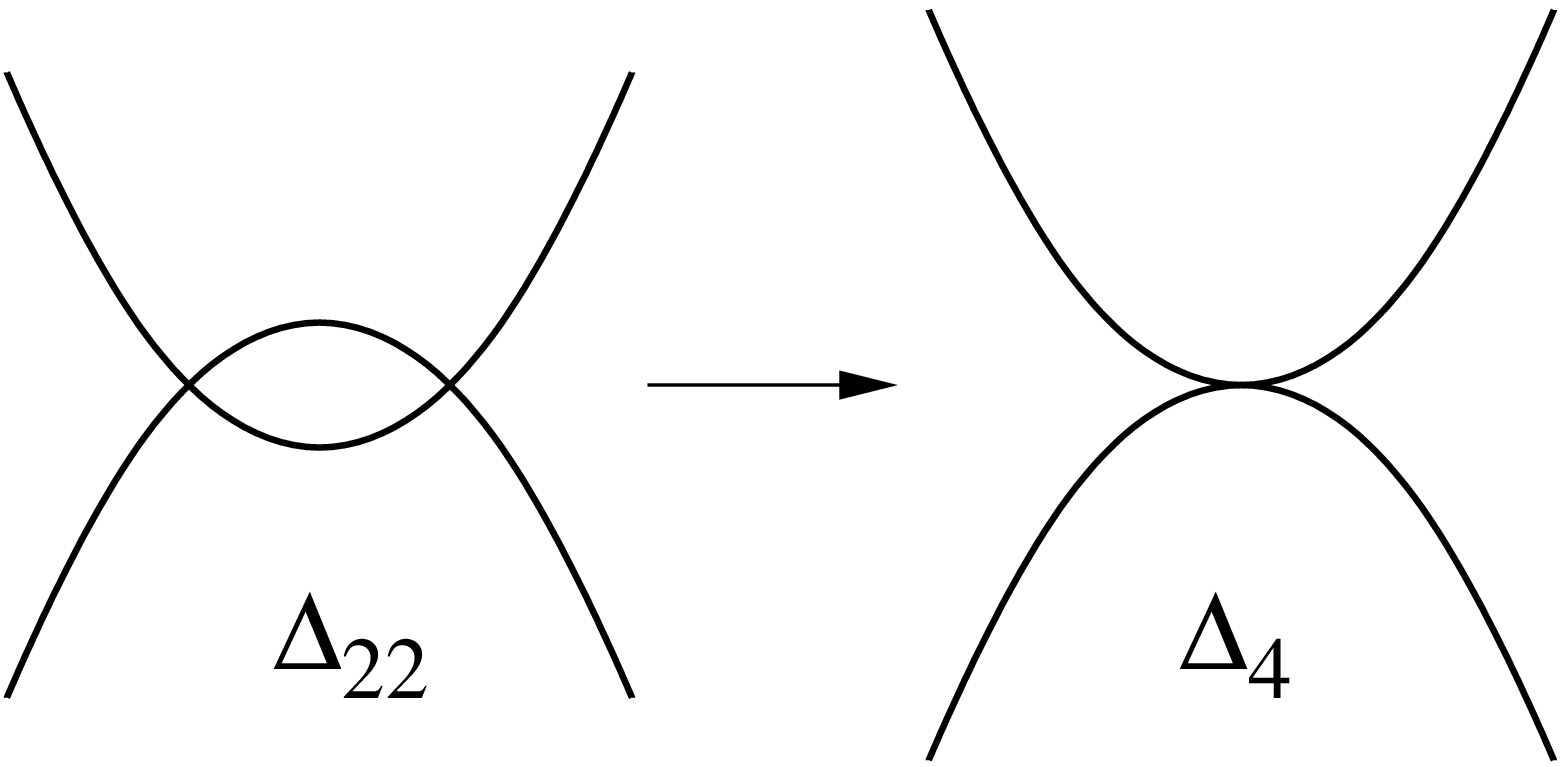}
\end{center}
$\Delta_{22}$ is given by $f=0$, $e^2=4g$ and $g\neq 0$ whereas
$\Delta_4$ is defined by $e=f=g=0$.  We get the following limit
\begin{equation*}
\begin{cases}
 e=-2 h \\
 f=C \phi\\
 g=h^2+ C\gamma
\end{cases}
\end{equation*}
We see indeed that at $C=0$ we are on $\Delta_{22}$ and when we
specialize to $C=h=0$ we are on $\Delta_4$.  We have at leading order
in $C$:
\begin{equation*}
\Delta\sim C^2 h^2 (\gamma^2-\phi^2 h), \quad j\sim
\frac{1}{C^2}\frac{h^4}{\gamma^2-\phi^2 h}.
\end{equation*}
Clearly when $C$ goes to zero away from $h\neq 0$ we are at the weak
coupling limit $j=\infty$.  But when we specialize to $C=h=0$, we get
an undetermination $j=\frac{0}{0}$.

\begin{figure}
\includegraphics[scale=0.4]{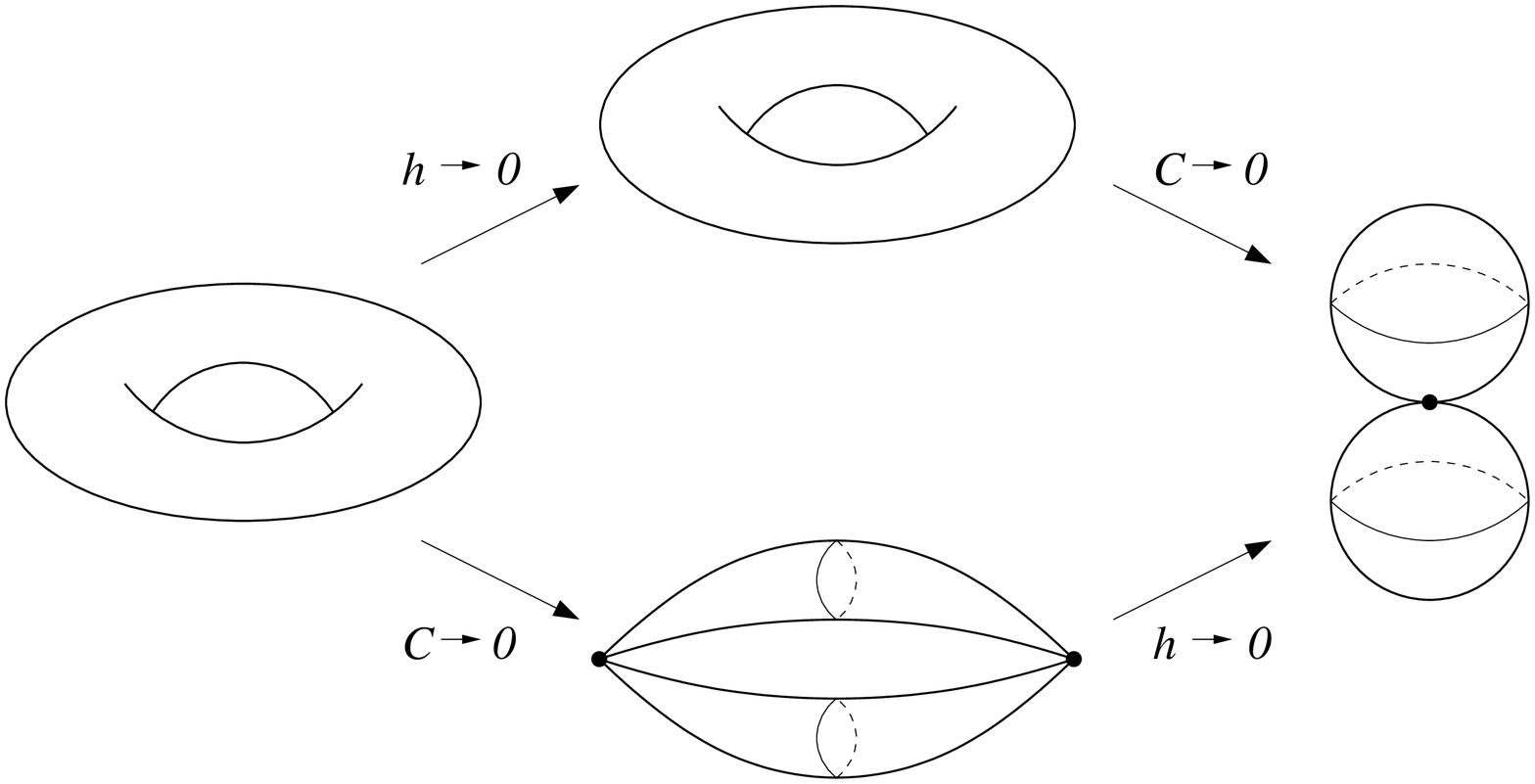}
\caption{{\bf Degeneration of the elliptic fiber as we take the weak
coupling limit $C\rightarrow 0$.}
This illustrates the case of the limit presented in \S\ref{IntroE71}
and \S\ref{IntroE6Lim}.  The limit $h\rightarrow 0$ is a
specialization to the orientifold locus $h=0$. As we take the weak
coupling limit ($C\rightarrow 0$), the elliptic curve reduces to the
union of two nonsingular rational curves intersecting transversally at
two distinct points. Topologically, this is two spheres meeting at two
distinct points. If we specialize to $h=0$ as we take $C\rightarrow
0$, the two rational curves become tangent to each other at a point.
In order to have a better understanding of the orientifold
configuration, we can get to $C=h=0$ by taking a different road.  If
we first specialize to $h=0$ before taking the weak coupling limit,
the elliptic fiber does not denegerate.  However, once we take the
weak coupling limit $C\rightarrow 0$ after taking $h=0$, all the
elliptic fibers flow to $e=f=g=0$ which corresponds to two  
2-spheres tangent at a point.
\label{E7WKL}} 
\end{figure}

The interpretation of this D-brane configuration is as follows.  In
the base, we have the orientifold plane $\underline{O}: h=0$ and a D7
brane-image-brane $\underline{D}:\gamma^2-\phi^2 h=0$.  In order to
appreciate that $\underline{D}$ corresponds indeed to a
brane-image-brane as seen in the base, we can move to the double cover
$X:\xi^2=h$, where the brane-image-brane will split into two
components. We get the following brane spectrum:
\begin{equation}\nonumber
O: \xi=0, \quad D=D_++D_-, 
\quad D_\pm :
\gamma\pm \phi \xi=0.
\end{equation}
We see that all the loci appearing here are smooth.
The F-theory type IIB tadpole matching without flux  gives 
\begin{equation}\nonumber
2\chi(Y)=4\chi(O)+\chi(D_+)+\chi(D_-).
\end{equation}
Since $D_+$ and $D_-$ are isomorphic, they have
the same Euler characteristic and therefore the tadpole matching
relation simplifies to
\begin{equation}\nonumber
\chi(Y)=2\chi(O)+\chi(D_+).
\end{equation}
Using the Sethi-Vafa-Witten relation for an $E_7$ fibration and the
adjunction formula to compute the Euler characteristic of $O$ and
$D_+$, it is easy to show that the tadpole matching relation holds
without any need for fluxes (cf.~\S\ref{E7224}).  Therefore the
brane-image-brane configuration does not break supersymmetry in
contrast to the special configuration of a brane-image-brane in a
Weierstrass model (see \S4.4 of \cite{Collinucci:2008pf}).

This tadpole relation holds at the level of Chern classes and for any 
dimension of the base and without the Calabi-Yau condition:
\begin{equation}\nonumber
\varphi_* c(Y)=2\rho_*c(O)+\rho_* c(D_+).
\end{equation}

\subsubsection{A  supersymmetric D7 Whitney - D7-brane-image-brane 
configuration}\label{IntroE7Lim2}

We can consider another limit defined by a transition
$\Delta_{112}\rightarrow \Delta_{13}$.  The details can be seen in
\S\ref{E711213}:
\begin{center}
 \includegraphics[scale=0.3]{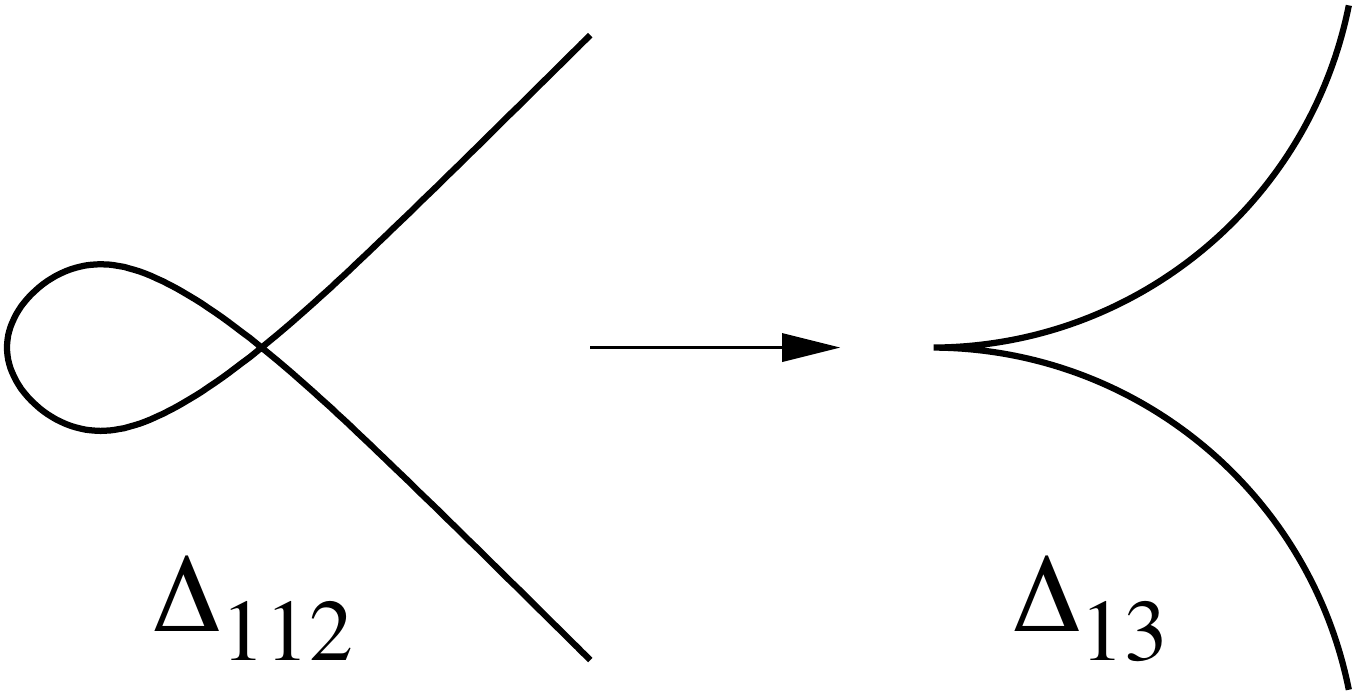}
\end{center}
The weak coupling limit is 
\begin{equation}\nonumber
\begin{cases}
e= - 6 k^2+ h\\
f=-2k (4 k^2-h)+ 2 C \phi\\
f=-k^2(3k^2-h)+ 2 C k \phi+ C^2 \gamma
\end{cases}
\end{equation}
We get at leading order in $C$
\begin{equation}\nonumber
\Delta\sim C^2 h^2 (h-4 k^2)(\phi^2- h \gamma), \quad j\sim\frac{1}{C^2}
\frac{h^4}{(h-4 k^2)(\phi^2-h \gamma)}.
\end{equation}
The D-brane spectrum is given by an orientifold plane
$\underline{O}:h=0$, a D7 brane-image-brane $\underline{D}_1: h-4 k^2$
and a singular D7 brane $\underline{D}_2:\phi^2-h \gamma$.  As we move
to the double cover $X:\xi^2=h$, we get
\begin{equation}\nonumber
O: \xi=0, \quad D_{1\pm}:\xi \pm 2 k=0, \quad 
D_2: \phi^2-\xi^2 \gamma=0.
\end{equation}
We see that $O$ and $D_{1\pm}$ are all smooth whereas $D_{2}$ has the
same singularities as in Sen's limit of a Weierstrass model.  The
D7-brane $D_{1+}$ and its orientifold image are in the same class as
$O$. Moreover, when $k$ vanishes, the brane-image-brane $D_\pm$ is on
top of the orientifold plane and $j\sim h^3$: $D_{1+}=D_{1-}=O$ as
$k\to 0$.

The F-theory/Type IIB  tadpole matching relation gives
\begin{equation}\nonumber
2\chi(Y)=4\chi(O)+\chi(D_{1+})+
\chi(D_{1-})+\chi_o(D_2).
\end{equation}
We can can simplify this relation by using $\chi(O)=
\chi(D_{1+})=\chi(D_{1-})$:
\begin{equation}\nonumber
2\chi(Y)=6\chi(O)+\chi_o(D_2).
\end{equation}
This relation does hold.  So we do not need fluxes here, and hence the
brane-image-brane configuration does not break supersymmetry.  In the
F-IIB tadpole matching relation, we have used the Euler characteristic
$\chi_o$ for $D_2$ since it has the singularities of a Whitney
umbrella. The fact that the relation holds (verified in
\S\ref{E711213}) is an additional non-trivial confirmation that this
is the right Euler characteristic to use in orientifold theories.

\subsubsection{Two  D7 brane-image-brane pairs in $E_6$}\label{IntroE6Lim}
An $E_6$ fibration admits a rich list of possible transitions from
semistable to unstable singular fibers.  Here, we consider the
specialization $\Delta_Q \rightarrow\Delta_P$ (with notation as in
\S\ref{SingE6}):
\begin{center}
 \includegraphics[scale=1]{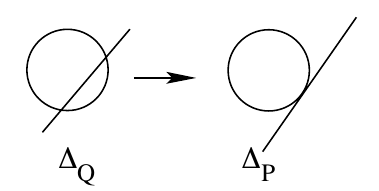}
\end{center}
The semistable fiber $\Delta_Q$ is a specialization of a nodal
curve. It corresponds to the transversal intersection of a conic with
a line.  The unstable curve $\Delta_P$ is a specialization of
$\Delta_Q$ where the line and the conic become tangent to each other.
The construction of this limit is given in details in section
\S\ref{E6lim}. The final result reads:
\begin{equation*}
\begin{cases}
d=6k \\
e = 9k^2+3h \\
f = 9k^2+3h+C\phi \\
g = 2k(5k^2+3h)+C(\gamma+k\phi)
\end{cases} \quad.
\end{equation*}
At leading order in $C$, we have 
\begin{equation*}
j\sim \frac 1{C^2}\cdot \frac{h^4}{(h+3k^2)(\gamma^2-h\phi^2)}.
\end{equation*}
We clearly reach the weak coupling limit when $C$ goes to zero. 
The limiting discriminant is given by 
\begin{equation*}
\Delta_h = h^2 (h+3k^2)(\gamma^2-h\phi^2).
\end{equation*}
Once we move to the double cover $X:\xi^2=h$, this corresponds to an
orientifold plane on $\xi=0$. A D7 brane-image-brane on
$D_{1\pm}:\gamma\pm \xi \phi=0$ and another one on $D_{2\pm} :\xi\pm
i\sqrt{3} k=0$.  All these branes are wrapping smooth divisors. The
brane-image-branes do not break supersymmetry since we can satisfy the
F-theory-type-IIB tadpole relation without any need for fluxes:
\begin{equation}\nonumber
2\chi(Y)=4\chi(O)+\chi(D_{1+})+\chi(D_{1-})+\chi(D_{2+})+\chi(D_{2-}).
\end{equation}
Since $[D_{1+}]=[D_{1-}]$ and $[D_{2+}]=[D_{2-}]=[O]$, we obtain the
simplified relation
\begin{equation}\nonumber
\chi(Y)=3\chi(O)+\chi(D_{1+}),
\end{equation}
verified in \S\ref{E6QP}.

\subsection{Singularities and the weak coupling limit of F-theory}
In the original $E_8$ limit considered by Sen, there is generally a
unique D7 brane, wrapping a singular divisor with the singularities of
a Whitney umbrella.  We have presented here two $E_7$ limits, one
admitting solely smooth divisors and another mixing smooth divisors
and divisors with Whitney umbrella singularities, and one $E_6$ limit
consisting of smooth divisors.  It is natural to ask for a
classification of all configurations of smooth divisors that can
appear generally for $E_8$, $E_7$, and $E_6$ fibrations.  The natural
constraint to consider is that the configuration should satisfy the
F-theory-type-IIB tadpole matching condition: this would ensure that
it can be supersymmetric even in the presence of brane-image-brane
pairs.

Surprisingly, this question has a very clean answer. Configurations
satisfying the tadpole matching condition on unrestricted bases form
an extremely short list: there are only three configurations of smooth
divisors satisfying this constraint (Theorem~\ref{nogothm}). Two of
these are precisely the $E_7$ and $E_6$ configuration we have realized
by weak coupling limits; the third one is not realized as such, and
cannot be realized in $E_6$, $E_7$, or $E_8$ fibrations.  In
particular, there are no configurations of smooth divisors satisfying
the tadpole matching conditions in $E_8$ fibrations. This indicates
that the appearance of the singular brane in Sen's limit cannot be
avoided.

There are three additional configurations of smooth divisors
satisfying the weaker requirement imposed by the tadpole matching
condition in the Calabi-Yau case, for a base of dimension~$3$: one in
$E_8$, one in $E_7$ and one in $E_6$. However, these extra three
configurations have not yet been realized through weak coupling
limits.  See \S\ref{Class.smooth} for details.


\section{Families of elliptic fibrations and their singular 
fibers}\label{Sinfib}
 
In this section, we define the elliptic fibrations of type $E_8$, $E_7$ and  
$E_6$ and study systematically their singular 
fibers (see Remark~\ref{fibcontext}).
We also introduce the notion of semistable and unstable fibers that will
play a crucial role in explaining the geometric interpretation of the
weak coupling limit.

We adopt the following terminology, after \cite{Klemm:1996ts}. 
Let $B$ be a nonsingular compact complex variety, endowed with a line
bundle $\cL$. We let $\cE=\cO\oplus \cL_1\oplus \cL_2$, with
$\cL_1=\cL_2=\cL$, and we consider the projectivization $\Pbb(\cE)$,
with tautological line bundle $\cO(1)$.  An {\em $E_6$ elliptic
fibration\/} over $B$ is a nonsingular hypersurface $Y$
of~$\Pbb(\cE)$, with equation
\begin{equation}
\label{E6eq}
x^3 + y^3= d\, xyz + e\, xz^2 + f\, yz^2 + g\, z^3\quad.
\end{equation}
Here, $z$, $x$, and $y$ are sections of $\cO(1)$, $\cO(1)\otimes \cL$,
and $\cO(1)\otimes \cL$, respectively, whose vanishing gives equations
for the embeddings $\Pbb(\cL_1\oplus \cL_2)\subseteq \Pbb(\cE)$,
$\Pbb(\cO\oplus \cL_2)\subseteq \Pbb(\cE)$, $\Pbb(\cO\oplus
\cL_1)\subseteq \Pbb(\cE)$, respectively.  Note that these three loci
have no common intersection; over each $p\in B$, the sections $x$,
$y$, $z$ serve as projective coordinates in the fiber of $\Pbb(\cE)$
over $p$. The powers of $\cL$ appearing in $\cE$ are chosen so that
\eqref{E6eq} is balanced; the same requirement then determines the
weights of $d,e,f,g$, so that the terms in \eqref{E6eq} are sections
of various line bundles as follows:
\begin{equation*}
\begin{tabular}{|c||c|}
\hline
$z$ & $\cO(1)$ \\
\hline
$x$ & $\cO(1)\otimes \cL$ \\
\hline
$y$ & $\cO(1)\otimes \cL$ \\
\hline
\end{tabular}
\quad,\quad
\begin{tabular}{|c||c|}
\hline
$d$ & $\cL$ \\
\hline
$e,f$ & $\cL^2$ \\
\hline
$g$ & $\cL^3$ \\
\hline
\end{tabular}\quad.
\end{equation*}
With these positions, $Y$ is the zero-set of a section of
$\cO(3)\otimes \cL^3$ in $\Pbb(\cE)$.  The bundle projection induces a
map $\varphi: Y \to B$; for every $p\in B$, the fiber
$\varphi^{-1}(p)$ is the cubic curve in the fiber $\Pbb^2$ of
$\Pbb(\cE)$ over $p$, defined by the equation \eqref{E6eq} in the
coordinates $x$, $y$, $z$.  The nonsingularity of $Y$ is guaranteed if
the sections $d,e,f,g$ are sufficiently general.

Similarly, an {\em $E_7$ elliptic fibration\/} over $B$ is a
nonsingular hypersurface $Y$ with equation
\begin{equation}
\label{E7eq}
y^2= x^4+e x^2 z^2+f x z^3 + g z^4
\end{equation}
in a bundle of weighted projective planes $\Pbb_{112}(\cO \oplus \cL
\oplus \cL^2)$ over $B$. Here, the terms appearing in \eqref{E7eq} are
sections of bundles as follows;
\begin{equation*}
\begin{tabular}{|c||c|}
\hline
$z$ & $\cO(1)$ \\
\hline
$x$ & $\cO(1)\otimes \cL$ \\
\hline
$y$ & $\cO(2)\otimes \cL^2$ \\
\hline
\end{tabular}
\quad,\quad
\begin{tabular}{|c||c|}
\hline
$e$ & $\cL^2$ \\
\hline
$f$ & $\cL^3$ \\
\hline
$g$ & $\cL^4$ \\
\hline
\end{tabular}\quad.
\end{equation*}
We note that while $\cO(2)$ is a line bundle on $\Pbb_{112}(\cO \oplus
\cL \oplus \cL^2)$ (it is the restriction of the tautological line
bundle via the natural embedding into $\Pbb(\cO\oplus \cL \oplus \cL^2
\oplus \cL^2)$), $\cO(1)$ only exists as a formal square root of
$\cO(2)$, due to the singularities of the weighted projective plane
$\Pbb_{112}$. However, $Y$ is disjoint from these singularities, so
this plays no role in our analysis of the geometry of $Y$, or in
intersection-theoretic computations.

The hypothesis of nonsingularity implies that general fibers of $E_6$
or $E_7$ fibrations, given by equations of type \eqref{E6eq},
\eqref{E7eq}, resp.~in $\Pbb^2$, $\Pbb_{112}$, resp., are elliptic
curves.  The subtleties of the geometry of these fibrations is encoded
in the {\em singular\/} fibers, whose type differs in the two types of
fibrations.

{\em $E_8$ fibrations\/} may also be defined similarly, by an equation
\begin{equation}
\label{E8eq}
y^2 = x^3 + f\, x z^4 + g\, z^6
\end{equation}
in a bundle of weighted projective planes 
$\Pbb_{123}(\cO\oplus \cL^2\oplus \cL^3) \to B$; 
$x$, $y$, etc.~are sections of line bundles as follows:
\begin{equation*}
\begin{tabular}{|c||c|}
\hline
$z$ & $\cO(1)$ \\
\hline
$x$ & $\cO(2)\otimes \cL^2$ \\
\hline
$y$ & $\cO(3)\otimes \cL^3$ \\
\hline
\end{tabular}
\quad,\quad
\begin{tabular}{|c||c|}
\hline
$f$ & $\cL^4$ \\
\hline
$g$ & $\cL^6$ \\
\hline
\end{tabular}\quad.
\end{equation*}
These fibrations are studied in \cite{AlEs}\footnote{In \cite{AlEs} we
place these fibrations in an ordinary projective bundle
$\Pbb(\cO\oplus \cL^3 \oplus \cL^2)$. This does not affect the
geometry of the fibration.}. The $E_8$ case is general in the sense
that every elliptic curve can be put in Weierstrass normal form; this
can be used to reduce the equations \eqref{E6eq}, \eqref{E7eq} to the
standard $E_8$ type, and it is useful in computing e.g.~the $j$
function and the discriminant locus of a fibration. However, this
operation affects the singular fibers: in the $E_8$ case, the singular
fibers may only be nodal (one singular point, with two distinct
tangent lines) or cuspidal (one singular point, a `double' tangent
line). See Figure~1 for representative pictures; the arrow denotes
that nodal curves can specialize to cuspidal curves in families.
\begin{figure}
\includegraphics[scale=0.4]{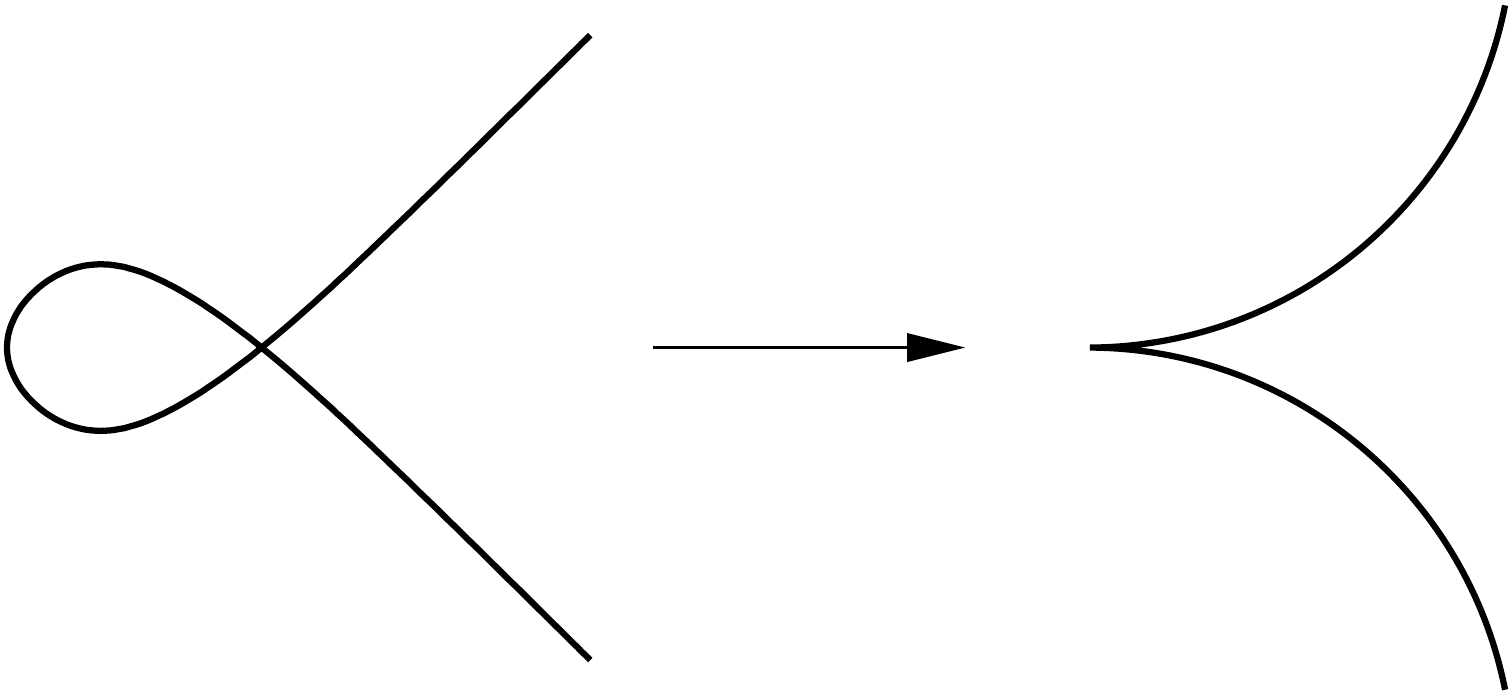}
\caption{Singularities of $E_8$ fibrations\label{E8pic}} 
\end{figure}
The discriminant has equation $4f^3+27g^2$, and the cuspidal locus has
equations $f=g=0$. The study of the corresponding stratification of
the discriminant is an ingredient in~\cite{AlEs}.

Fibrations of type $E_6$ and $E_7$ as defined above have a richer
catalog of possible singular fibers, and this information is key to
the weak coupling limits studied in this paper.  In the rest of this
section we classify singular fibers of $E_6$ and $E_7$ fibrations. 
The result is the summarized in tables  \ref{table.E7SingFib} and 
\ref{table.E6SingFib}.  

\subsection{Singularities of $E_7$ fibrations}\label{SingsE7}
Fibers of $E_7$ fibrations are obtained by specializing \eqref{E7eq},
that is, as curves with equation
\begin{equation}
\label{E7fiber}
y^2=x^4+ex^2z^2+fxz^3+gz^4
\end{equation} 
for fixed $e$, $f$, $g$. The Weierstrass model for this curve is
\begin{equation}
\label{WE7}
y^2+x^3+Fx+G =y^2+x^3+\left(-4g-\frac 13 e^2\right)x
+\left(-f^2+\frac 83 e g-\frac 2{27} e^3\right)=0\quad.
\end{equation}
Singularities of \eqref{E7fiber} are determined by the common
vanishing of the three partial derivatives in $x$, $y$, $z$; it
follows that they are necessarily contained in $y=0$, and in fact they
coincide with the multiple solutions of the restriction of the curve
to the line $y=0$, which has equation
\begin{equation}
\label{E7restr}
x^4-ex^2z^2-fxz^3-gz^4=0\quad.
\end{equation} 
This is a binary quartic form, which in general vanishes at four
distinct points in the projective line $\Pbb^1_{(x:z)}$. It acquires
multiplicities when its discriminant vanishes:
\begin{equation}
\label{E7disc}
\Delta:\quad 4e^3f^2+27f^4-16e^4g+128e^2g^2-144egf^2-256g^3=0\quad.
\end{equation}
This agrees with the discriminant $4F^3+27G^2=0$ of \eqref{WE7}, and
defines a hypersurface $\Delta$ of the base $B$.

To determine the different type of singular fibers, we note that the
multiplicities of the zeros of a quartic form \eqref{E7restr} can
occur in one of 5 different ways, determined by the five partitions of
the number 4. Thus there are four different types of forms
\eqref{E7restr} vanishing with multiplicities, and up to a change of
the coordinates $x$, $z$ they can be given as follows (listed together
with the corresponding partitions):
\begin{equation*}
\left\{
\begin{tabular}{ccc}
$1+1+2$ & : & $(x^2-z^2)x^2$ \\
$1+3$     & : & $(x+3z)(x-z)^3$ \\
$2+2$     & : & $(x^2-z^2)^2$ \\
$4$         & : & $x^4$
\end{tabular}\right.\quad.
\end{equation*}
Since a linear change of coordinates does not affect the singularity
type, this shows that there are four types of degenerate fibers in an
$E_7$ elliptic fibration. Listed in terms of a representative
equation, they are (cf.~Figure~\ref{E7pic}):
\begin{figure}
\includegraphics[scale=0.3]{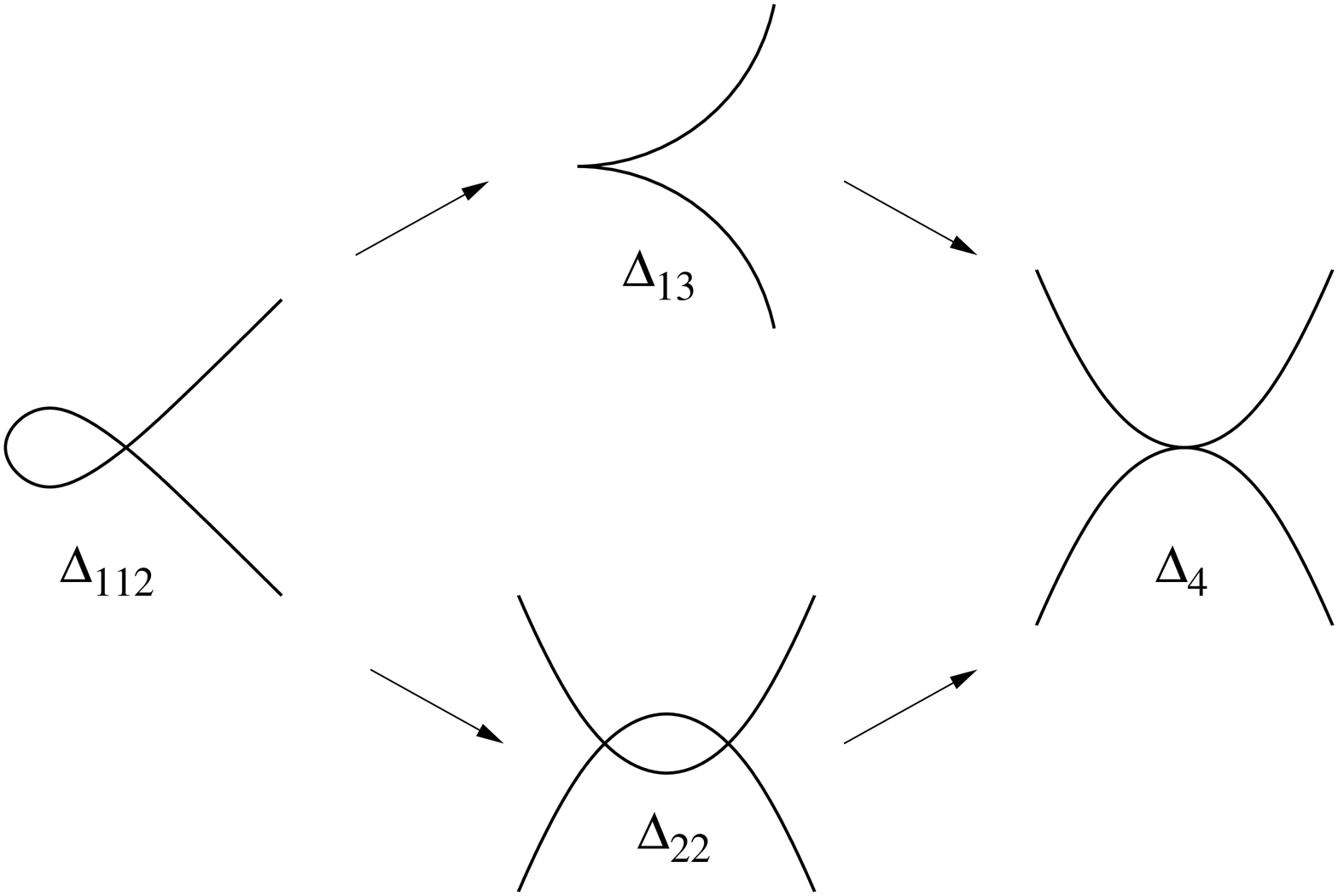}
\caption{Singularities of $E_7$ fibrations\label{E7pic}} 
\end{figure}
\begin{itemize}
\item $y^2=(x^2-z^2)x^2$,
a curve with a single singular point, a node;
\item $y^2=(x+3z)(x-z)^3$,
a curve with a single singular point, a cusp;
\item $y^2=(x^2-z^2)^2$,
the union of two nonsingular rational curves meeting transversally
at two points;
\item $y^2=x^4$, the union of two nonsingular rational curves, meeting
at one point and tangent to each other at that point.
\end{itemize}

These different types of singularities determine a stratification of
the discriminant~$\Delta$ into loci $\Delta_{112}$, $\Delta_{13}$,
$\Delta_{22}$, $\Delta_{4}$ (indexed according to the corresponding
partition), which is easy to determine explicitly. Indeed, it is easy
to check that the quartic form \eqref{E7restr} can be a fourth power
only if $e=f=g=0$:
\begin{equation}
\Delta_4:\begin{cases}
g=0 \\ f=0 \\ e=0
\end{cases}\quad;
\end{equation} 
it is of type $2+2$ if and only if it is 
\begin{equation}
(x^2-\alpha z^2)^2=x^4-2\alpha x^2 z^2+\alpha^2 z^4
\end{equation}
for some $\alpha\ne 0$, that is:
\begin{equation*}
\label{Delta22eq}
\Delta_{22}:\begin{cases}
g\ne 0 \\ f=0 \\ e^2=4g
\end{cases}\quad;
\end{equation*}
it is of type $3+1$ if and only if it is
\begin{equation*}
(x+3\alpha z)(x-\alpha z)^3=x^4-6\alpha^2 x^2 z^2+8\alpha^3 x z^3
-3\alpha^4 z^4
\end{equation*}
for some $\alpha\ne 0$:
\begin{equation}
\label{E713eq}
\Delta_{13}:\begin{cases} 
e\ne 0 \\
8e^3+27 f^2 = 0 \\
e^2+12 g=0
\end{cases}\quad;
\end{equation}
and it is of type $1+1+2$ if it satisfies \eqref{E7disc} but it does
not satisfy any of the conditions listed above:
$\Delta_{112}=\Delta\smallsetminus (\Delta_{22}\cup \Delta_{13}\cup
\Delta_4)$.

We can also note that $\Delta_{22}$ and $\Delta_{13}$ have 
codimension~$1$ in $\Delta$, while $\Delta_4$ has codimension~2. 
The intersection of the closures of $\Delta_{22}$ and $\Delta_{13}$ is 
$\Delta_4$; this intersection is not transversal.

Finally, this stratification may be compared with the stratification of
the discriminant of the Weierstrass model of the fibration. Stability 
considerations imply that the singular curves of the Weierstrass model
will be cuspidal over $\Delta_4$ and $\Delta_{13}$, and nodal over
$\Delta_{22}$ and $\Delta_{112}$; this can also be verified by an 
explicit computation.

\subsection{Singularities of $E_6$ fibrations\label{SingE6}}
The situation for $E_6$ fibrations is richer. In general, the possible
singularities of a cubic plane curves are easily classified: an {\em
irreducible\/} singular cubic is either nodal or cuspidal, and in fact
isomorphic via a linear transformation to either $y^2z=x^3+x^2z$ or
$y^2z=x^3$; representative drawings for these two singularities are
given in Figure~\ref{E8pic}. A {\em reducible\/} cubic must be the
union of a nonsingular conic and a line, or of three lines, and it is
easy to list the possible types of intersections and multiplicities
that can occur.  Figure~\ref{singcubs} summarizes the situation,
indicating
\begin{figure}
\includegraphics[scale=1]{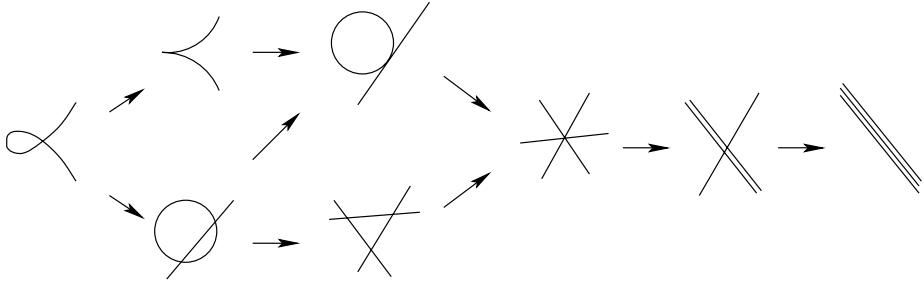}
\caption{Singularities of cubic curves\label{singcubs}} 
\end{figure}
specializations by arrows: for example, cubics decomposing as
unions of a conic and a tangent line are in the closure of the set
of cuspidal cubics, but `triangles' are not. 

We can use the following approach in classifying the singularities
that can appear by specializing the equation \eqref{E6eq} of an $E_6$
elliptic fibration: start from each type of singular cubic, and
determine the relation on the coefficients $d,e,f,g$ obtained by
imposing that \eqref{E6eq} is of the prescribed type.  Qualitatively,
the end-result of this analysis is represented graphically in
Figure~\ref{E6pic}. The discriminant of \eqref{E6eq} is stratified by
loci $\Delta_N$, $\Delta_C$, $\Delta_Q$, $\Delta_P$, $\Delta_T$,
$\Delta_S$.
\begin{figure}
\includegraphics[scale=1]{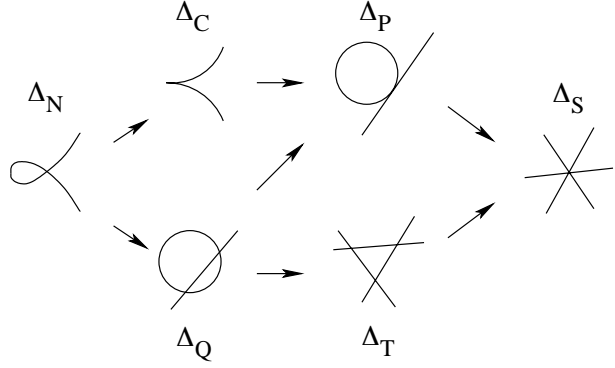}
\caption{Singular fibers of $E_6$ fibrations\label{E6pic}} 
\end{figure}
Each picture corresponds to a singular fiber, obtained over a point of
the named locus; as above, arrows denote specialization. We also let
$N=\overline{\Delta_N}, C=\overline{\Delta_C}$, etc.; thus,
$\Delta_N=N\smallsetminus (C\cup Q)$, etc. Equations for the six loci
$N,\dots,S$ in $B$ will be given below. The codimension of $N$ in $B$
is $1$; of $C$ and $Q$,~$2$; of $P$ and $T$, $3$; and $S$ has
codimension~$4$; for example, $S=\emptyset$ if $\dim B=3$ (and the
sections are sufficiently general).  The loci $Q$ and $P$ each consist
of three components. The three components of $Q$ meet precisely along
the locus $T$.  The three components of~$P$ likewise meet precisely
along $S$.

A more thorough analysis suggests that (if the sections $d$, $e$, $f$,
$g$ are sufficiently general) the discriminant $N$ is singular, with
multiplicity $2$ along $C$ and $Q$, multiplicity~$3$ along $P$
and~$T$, and multiplicity $4$ along $S$ if $\dim B\ge 4$.

The Weierstrass form for \eqref{E6eq} is{\small
\begin{multline}
y^2+x^3+Fx+G \\
=y^2+x^3+\left(-3ef+\frac 92 dg-\frac 1{48} d^4\right)x+
\left(-e^3-f^3+\frac{27}4 g^2+\frac 34 efd^2-\frac 58 gd^3-
\frac 1{864} d^6\right)=0;
\end{multline}}
its discriminant is $4F^3+27G^2$, giving the equation of $N$:{\small
\begin{equation*}
\boxed{
N:\,\, 4\left(-3ef+\frac 92 gd-\frac 1{48} d^4\right)^3
+27 \left(-e^3-f^3+\frac{27}4 g^2+\frac 34 efd^2-\frac 58 gd^3-\frac 1{864} d^6
\right)^2=0}
\end{equation*}}
Over a general point of this hypersurface $N$, \eqref{E6eq} defines an
irreducible curve with a node; to verify this, it suffices to produce
one such curve, and one may be obtained by setting $e=f=g=0$, $d=1$:
\begin{equation*}
x^3+y^3=xyz
\end{equation*}
is an irreducible curve with a single node at $(z:x:y)=(1:0:0)$.

The $j$-invariant $\sim\frac{F^3}{4F^3+27G^2}$ of a cubic curve is
defined (or $\infty$) if and only if the curve is not in the closure
of the cuspidal locus; it follows that the the closure $C$ of the
locus parametrizing cuspidal curves has equations
$F=G=0$. Alternately, one can impose that the tangent cone at a
singularity consists of a double line; this can be done by a Hessian
computation, and leads to the equation $F=0$; together with the
discriminant $4F^3+27G^2$, this again shows that $C$ has equations
$F=G=0$:
\begin{equation*}
\boxed{C:\quad
\left\{
\aligned 
-3ef+\frac 92 gd-\frac 1{48} d^4 &=0\\
-e^3-f^3+\frac{27}4 g^2+\frac 34 efd^2-\frac 58 gd^3-\frac 1{864} d^6 &=0
\endaligned
\right.}
\end{equation*}
To see that the general point of $C$ indeed corresponds to an
irreducible curve with an ordinary cusp, it suffices to produce one
curve of this type and satisfying the equations for~$C$; one such
curve is
\begin{equation*}
x^3+y^3=3xz^2-2z^3\quad.
\end{equation*}

To proceed further, we describe explicitly the other singular curves
represented in Figure~\ref{singcubs}, and find conditions on the
coefficients $d,e,f,g$ obtained by imposing that \eqref{E6eq} matches
the equations we obtain.

The general equation of the union of a line and a conic is:
\begin{equation*}
(x+\rho y+\alpha z)(x^2+\gamma x y+\delta y^2+\epsilon xz+\phi yz
+\beta z^2)=0\quad.
\end{equation*}
We may assume that the coefficients of $x$, $x^2$ in the two factors
are~$1$, since the product will have to match \eqref{E6eq}. For the
match to occur, it is necessary that the coefficients of $x^2y$,
$x^2z$, $xy^2$, $y^2z$ all vanish.  The first two constraints give
\begin{equation*}
\gamma=-\rho\quad,\quad \epsilon=-\alpha\quad;
\end{equation*}
the third then forces
\begin{equation*}
\delta=\rho^2\quad;
\end{equation*}
and the fourth translates into
\begin{equation*}
\rho(\phi+\alpha\rho)=0\quad.
\end{equation*}
At this point the coefficient of $y^3$ is $\rho^3$; matching \eqref{E6eq} 
forces $\rho^3=1$, and in particular $\rho\ne 0$, yielding
\begin{equation*}
\phi=-\alpha\rho\quad.
\end{equation*}
The conclusion is that these curves are in the form represented by
\eqref{E6eq} if and only if they are in fact
\begin{equation}
\label{lineandc}
(x+\rho y+\alpha z)(x^2-\rho x y-\alpha xz+\rho^2 y^2-\alpha \rho y z
+\beta z^2)=0\quad,
\end{equation}
with $\rho$ a cubic root of~$1$:
\begin{equation*}
\rho = 1\quad\text{or}\quad 
\frac{-1\pm i\sqrt 3}2\quad.
\end{equation*}
Expanding \eqref{lineandc} and matching with \eqref{E6eq} we find that
\begin{equation*}
\left\{
\aligned
d &=3\rho\alpha\\
e &=\alpha^2-\beta \\
f &=\rho(\alpha^2-\beta) \\
g &=-\alpha \beta
\endaligned
\right.\quad,
\end{equation*}
and eliminating $\alpha$, $\beta$ ($\rho$ is already determined up to a
finite choice) gives the equations
\begin{equation}
\label{Qeq}
\boxed{Q:\quad
\left\{
\aligned
f &=\rho e\\
g &=\frac{d(9\rho^2 e-d^2)}{27}
\endaligned
\right.\quad,\quad \rho^3=1}
\quad.
\end{equation}
The three allowed values for $\rho$ give three (clearly irreducible)
components of the locus~$Q$ over which the fibers of \eqref{E6eq}
consist of the union of a line and a conic. For example, taking
$\rho=1$ and $e=f=-1$, $d=g=0$, gives
\begin{equation*}
x^3+ y^3+xz^2+ yz^2=(x+y)(x^2-xy+y^2-z^2)=0\quad,
\end{equation*}
the union of a line and a nonsingular conic meeting transversally.
Similar examples may be given on the other two components, showing
that the general fiber over each of the three components is of the
same type.

Also note that from the above discussion it follows that, on $Q$, we
must have
\begin{equation*}
\alpha = \frac{\rho^2 d}{3} ,\quad
\beta =\frac{\rho d^2}9-e ,\quad
\gamma =-\rho,\quad
\delta =\rho^2 ,\quad
\epsilon = -\frac {\rho^2 d}3 ,\quad
\phi = -\frac{d}{3} ,\quad
(\rho^3 =1) .
\end{equation*}
The fiber over a general point of $Q$ consists of the union of the
line $x=-\rho y-\alpha z$ and the conic $x^2-\rho x y-\alpha xz+\rho^2
y^2-\alpha\rho yz +\beta z^2=0$; intersecting these two components
gives the polynomial in $x$, $z$:
\begin{equation*}
3\rho^2 y^2+3\alpha \rho yz+ (2\alpha^2+\beta)z^2\quad;
\end{equation*}
the line intersect the conic in a double point if this polynomial has
a double root, that is, when its discriminant vanishes. This is true
if and only if
\begin{equation*}
5\alpha^2 + 4\beta = 0\quad.
\end{equation*}
As shown above, on $Q$ this is equivalent to
\begin{equation*}
5\left(\frac{\rho^2 d}{3}\right)^2+4\left(\frac{\rho d^2}9-e\right)=0\quad,
\end{equation*}
that is,
\begin{equation*}
4e=\rho d^2\quad.
\end{equation*}
Putting this equation together with the equations for $Q$ gives
equations for $P$:
\begin{equation}
\label{Peq}
\boxed{P:\quad
\left\{
\aligned
f &=\rho e=\frac{\rho^2 d^2}4 \\
g &=\frac{5 d^3}{108}
\endaligned
\right.\quad,\quad \rho^3=1}
\end{equation}
It is easy to check that $P=C\cap Q$; this is not a transversal
intersection.  A typical point of~$P$ may be obtained by setting
$\rho=1$ and $d=-6$; solving for $e,f,g$ gives the curve
\begin{equation*}
x^3+y^3+6xyz-9xz^2-9yz^2+10z^3
=(x+y-2z)(x^2-xy+2xz+y^2+2yz-5z^2)=0
\end{equation*}
It is easy to check that the conic is nonsingular, and the line is
tangent to it at the point $(z:x:y)=(1:1:1)$.

The other way in which a curve consisting of a line union a
transversal conic may degenerate is as a triangle of nonconcurrent
lines; this happens precisely when the conic becomes singular. The
conic has equation
\begin{equation*}
x^2-\rho x y-\alpha xz+\rho^2 y^2-\alpha\rho yz+\beta z^2=0\quad,
\end{equation*}
and imposing that the three partials have a common root gives the
condition $\beta=\alpha^2$.  As seen above, on $Q$ this is equivalent
to $e=0$. Putting together with the equations for $Q$ gives the
equations for the locus $T$ parametrizing triangles:
\begin{equation*}
\boxed{T:\quad
\left\{
\aligned
e &=f=0\\
g &=-\frac{d^3}{27}
\endaligned
\right.}
\quad.
\end{equation*}
A typical fiber over a point of $T$ is obtained by taking $d=3$, which
gives
\begin{equation*}
x^3+y^3+z^3=3xyz\quad.
\end{equation*}
This factors (over $\Cbb$) as a product of three nonconcurrent lines.

The triangle may degenerate by having the vertices come together.  The
resulting locus~$S$ of `stars' may be viewed as the intersection of
$P$ and $T$, and hence has equations
\begin{equation*}
\boxed{S:\quad
d=e=f=g=0
}\quad.
\end{equation*}
In other words, the only singular fiber of this type is the curve
\begin{equation*}
x^3 + y^3 = 0\quad.
\end{equation*}

Since the star is isolated, the fibers cannot degenerate
further. Thus, the last two (nonreduced) types of singular cubics
shown in Figure~\ref{singcubs} cannot occur in the family given by
\eqref{E6eq}.  This completes the classification of singular fibers in
a fibration with equation~\eqref{E6eq}.

Finally, one can examine the behavior of fibers of the corresponding
Weierstrass model over the strata of the discriminant. Since $F$ and
$G$ vanish along $C$, these fibers are cuspidal curves over $\Delta_C$
and its specializations $\Delta_P$ and $\Delta_S$. Since $F\ne 0$
along the other loci, the fibers of the Weierstrass model are nodal
over $\Delta_N$, $\Delta_Q$, and $\Delta_T$.

\subsection{Semistable and unstable fibers}
In each of the cases considered in this section, we will refer to the
fibers over the closure of the cuspidal loci as {\em unstable,\/} and
to the other types as {\em semistable;\/} see Figures~\ref{unstable}
and~\ref{sstable}.
\begin{figure}
\includegraphics[scale=1]{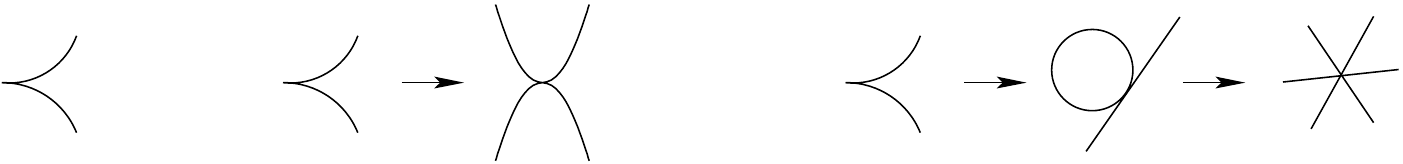}
\caption{Unstable fibers of $E_8$, $E_7$, $E_6$ fibrations\label{unstable}} 
\end{figure}
\begin{figure}
\includegraphics[scale=1]{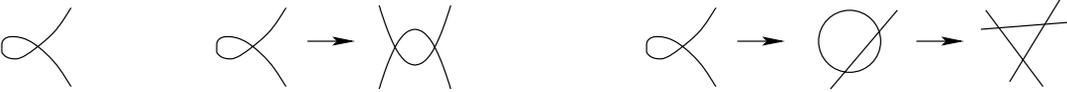}
\caption{Semistable fibers of $E_8$, $E_7$, $E_6$ fibrations\label{sstable}} 
\end{figure}

This terminology is borrowed from algebraic geometry. The
distinguishing features of the curves in Figure~\ref{unstable} is that
the $j$ invariant is {\em not\/} defined for these curves. These are
the curves for which the coefficient $F$ in the corresponding
Weierstrass model vanishes; since the discriminant $\Delta=4F^3+27G^2$
vanishes for all singular curves, the quotient $F^3/\Delta$ (which
equals $j$, up to a multiple) is undefined for unstable curves.  By
contrast, $F\ne 0$ for semistable curves; their $j$ invariant is
unambiguously defined to be $\infty$.

This distinction will be useful in \S\ref{Wcl}.


\section{Weak coupling limits}\label{Wcl}

The analysis carried out in \S\ref{Sinfib} provides us with a key
element of information in the construction of `weak coupling limits'
in the style of Sen's celebrated work (\cite{Sen:1997gv}). Our
strategy is to focus on the specialization of a type of singular fiber
to another in the situations encountered in \S\ref{Sinfib}, in a way
that we describe in general terms in~\S\ref{strate}. We precede this
description with a revisitation of Sen's $E_8$ limit, in~\S\ref{Sens};
we follow it in~\S\ref{E7lim} and~\S\ref{E6lim} with applications of
the strategy to the $E_7$ and $E_6$ environments, obtaining new weak
coupling limits.

In \S\ref{stratecomm} we clarify the reasons underlying
the choices inherent to our procedure.

\subsection{Sen's limit, revisited}\label{Sens}
Here we review the limit for $E_8$ fibrations considered by Ashoke Sen
in~\cite{Sen:1997gv}, in a way which emphasizes the general features
of the strategy we will apply to $E_7$ and $E_6$ in the rest of the
section.

As recalled in \S\ref{Sinfib}, $E_8$ fibrations may be defined by 
equation~\eqref{E8eq}:
\begin{equation*}
y^2 = x^3 + f\, x z^4 + g\, z^6
\end{equation*}
in a bundle of weighted projective planes $\Pbb_{123}(\cO\oplus
\cL^2\oplus \cL^3) \to B$.  Note that in the $E_8$ case we have only
one type of semistable fiber and one type of unstable fiber: nodal and
cuspidal curves, respectively. In terms of the parameters $f$ and $g$,
the fiber is a nodal curve when
\begin{equation}
\label{E8ss}
4f^3+27g^2=0\quad,\quad
f\ne 0
\end{equation}
while the unstable locus is
\begin{equation*}
f=g=0\quad.
\end{equation*}
We view this situation in a parameter space with coordinates $f$, $g$.
In this space, \eqref{E8ss} represents a cuspidal curve, which admits
the standard parametrization
\begin{equation*}
h\mapsto (f,g) = (-3h^2,-2h^3)\quad.
\end{equation*}
This motivates us to define a degenerate fibration $Y_h$ by setting
\begin{equation}
\label{E8Yh}
\begin{cases}
f = -3 h^2 \\
g = -2 h^3
\end{cases} \quad,
\end{equation}
or in divisor form
\begin{equation*}
Y_h:\quad
y^2 = x^3 - 3h^2\, x z^4 - 2h^3\, z^6 = (x+h z^2)^2\, (x-2h z^2)\quad.
\end{equation*}
Note that since $f$, $g$ are resp.~sections of $\cL^4$, $\cL^6$, this
can be achieved with $h$ a section of $\cL^2$; this will be one of the
requirements that we will pose in~\S\ref{strate}. Also note that
$f=g=0$ precisely if $h=0$; thus, the fibers of $Y_h$ over points of
$\uO$ (the hypersurface $h=0$ in $B$) are unstable (i.e., cuspidal),
and the fibers over points of $B\smallsetminus \uO$ are semistable
(i.e., nodal); this will also be a general feature of our
construction.

Next, we perturb $Y_h$ to get a general family $Y_h(C)$. We perturb
the assignment for~$f$ in \eqref{E8Yh} with a general choice of a
section $\eta$:
\begin{equation*}
f = -3h^2 + C\eta\quad,
\end{equation*}
where $\eta$ is a general section of $\cL^4$. Setting likewise
$g=-2h^3 +C \gamma$, one sees that $j\propto h^3$ at leading order in
$C$ unless $\gamma=h \eta$.  The case $j\propto h^3$ is interesting
(see \S\ref{needfornew}): it corresponds to a bound state of the
orientifold with a brane-image-brane on top of it.  We will consider
the general case of $j\sim h^{4-n}$, $n=0,1,2,3,4$ in
\S\ref{stratecomm}.  However, we will in general insist that $j\propto
h^4$, corresponding to a pure O7-orientifold.  This prompts us to set
$g=-2h^3 +C h \eta$ at first order in $C$, and adding a further term
to achieve generality we arrive at the description of $Y_h(C)$:
\begin{equation*}
\begin{cases}
f = -3 h^2 + C \eta\\
g = -2 h^3 + C h \eta + C^2 \chi
\end{cases} \quad,
\end{equation*}
with $h$, $\eta$, $\chi$ resp.~general sections of $\cL^2$, $\cL^4$,
$\cL^6$, resp. This is precisely the limit considered by Sen. With
these positions,
\begin{equation*}
j\sim \frac 1{C^2}\cdot \frac{h^4}{\eta^2+12 h \chi}
\end{equation*}
at leading order in $C$. Studying the limiting discriminant
$h^2(\eta^2+12 h \chi)$ and its pull-back to the double cover $X$ over
$B$ ramified along $\uO$ leads to interesting Euler characteristics
and Chern class identities, explored in \cite{AlEs}.

\subsection{The strategy}\label{strate}
We now abstract the key points of the procedure leading to Sen's limit
in \S\ref{Sens}. In order to produce weak coupling limits for the
elliptic fibrations considered in \S\ref{Sinfib} we propose the
following procedure:
\begin{itemize}
\item choose a specialization of a semistable singular fiber $S_1$
to an unstable one~$S_2$ among the types listed in~\S\ref{Sinfib};
\item realize the specialization by means of a section $h$ of $\cL^2$,
producing a degenerate family $Y_h$ for which the the fiber over the 
hypersurface $\uO$ defined by $h=0$ is of type $S_2$, and the general
fiber (over $B\smallsetminus \uO$) is of type $S_1$;
\item perturbe $Y_h$ to a general family $Y_h(C)$ depending on a
scalar parameter $C$, for which the $j$ invariant behaves at leading
order in $C$ as
\begin{equation*}
j\sim \frac 1{C^r}\cdot 
\frac{h^4}{\uD}
\end{equation*}
for some $r>0$.  Here, $\uD$ will be a component (or a collection of
components) of the limiting discriminant $\Delta_h$ of $Y_h(C)$ as
$C\to 0$ (that is, the leading term of the discriminant $\Delta_h(C)$
of $Y_h(C)$).  The other component of this limiting discriminant is
supported on $\uO$.
\end{itemize}

For example, $j\sim h^4$ ensures the right charge for a pure O7-brane.
In practice, satisfying the two basic requirements that $h$ be a
section of $\cL^2$ and that $j\sim h^4$ usually leads to nontrivial
constraints, and can be achieved only for certain choices of the
specialization $S_1\to S_2$.  Whenever we are successful in carrying
out this program, we can consider the double cover $X$ of $B$ ramified
along $\uO$. The case of greatest physical interest is the case in
which both $X$ and $Y$ are Calabi-Yau varieties, which corresponds to
$c_1(\cL)=c_1(TB)$.  In this case we will identify the ramification
locus $O$ of $X\to B$ as the support of an orientifold, and the
preimage $D$ of $\uD$ in $X$ as carrying D7 branes. Every such
situation leads then to the prediction of a tadpole relation between
the Euler characteristics of these loci. More details for these
constructions are given in~\S\ref{Tadrel}.
\smallskip

A few comments motivating more explicitly the specific requirements
listed here will be given in~\S\ref{stratecomm}.

\subsection{Two $E_7$ weak coupling limits}\label{E7lim}
We now apply the strategy reviewed in~\S\ref{strate} to $E_7$
elliptic fibrations, with equation \eqref{E7eq}:
\begin{equation*}
y^2= x^4+e x^2 z^2+f x z^3 + g z^4\quad.
\end{equation*}
The choice of a specialization from a semistable to an
unstable fiber is restricted to the list given in Figure~\ref{E7spec}.
\begin{figure}
\includegraphics[scale=0.25]{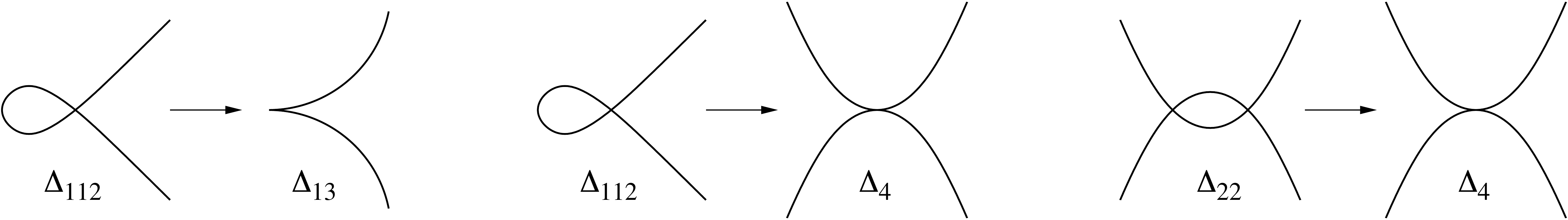}
\caption{Allowed specializations for $E_7$ fibrations\label{E7spec}} 
\end{figure}
We first choose to specialize $\Delta_{22}\to \Delta_4$. The closure of 
$\Delta_{22}$ has equations (cf.~\eqref{Delta22eq}):
\begin{equation*}
e^2 = 4g\quad,\quad f=0\quad.
\end{equation*}
Viewing $e$, $f$, $g$ as coordinates in a parameter space, these
equations describe a nonsingular conic curve, parametrized by
\begin{equation*}
h\mapsto (e,f,g) = (-2h,0,h^2)\quad.
\end{equation*}
This parametrization leads to the following description of the
degenerate fibration $Y_h$:
\begin{equation}
\label{E7Yh}
\begin{cases}
e = -2h \\
f = 0 \\
g = h^2
\end{cases} \quad,
\end{equation}
in divisor form
\begin{equation*}
Y_h:\quad
y^2 = x^4-2h x^2 z^2+ h^2 z^4=(x^2-h z^2)^2\quad.
\end{equation*}
We note that this is done by choosing $h$ to be a general section of
$\cL^2$, thus satisfying one of the basic requirements from
\S\ref{strate}; and that $g=0$ if and only if $h=0$, so that the
specialization $\Delta_4$ is attained precisely over $h=0$, also as
prescribed in \S\ref{strate}.

We next perturb \eqref{E7Yh} in the simplest way needed to produce a
general family $Y_h(C)$:
\begin{equation*}
\begin{cases}
e = -2h \\
f = C\phi  \\
g = h^2+C\gamma
\end{cases} \quad.
\end{equation*}
With this choice,
\begin{equation*}
j\sim \frac 1{C^2}\cdot \frac{h^4}{\gamma^2-\phi^2 h}
\end{equation*}
at leading order: thus, all the requirements listed in \S\ref{strate}
are satisfied. The limiting discriminant has equation
\begin{equation}
\label{E7Delh1}
h^2 \left(\gamma^2-\phi^2 h\right)=0\quad:
\end{equation}
with notation as in \S\ref{strate}, $\uD$ is the hypersurface
$\gamma^2-\phi^2 h=0$.  In the double cover $X:\xi^2=h$, this
hypersurface splits into two smooth components $D_\pm$ that are mapped
to each other by the orientifold involution, so that this hypersurface
corresponds to a brane-image-brane pair.  See \S\ref{IntroE71} and
\S\ref{E7224} for a description of the corresponding configuration of
O7 and D7 in the double cover, and for a precise identity of Euler
characteristics related to this configuration.
\smallskip

A second possible choice of specialization is $\Delta_{112} \to
\Delta_{13}$, leading to a different limit. We follow the same
procedure: in the parameter space $(e,f,g)$, the closure
of~$\Delta_{13}$ is parametrized by
\begin{equation*}
k \mapsto (e,f,g)=(-6k^2,-8k^3,-3k^4)
\end{equation*}
(see \eqref{E713eq}). We can extend this to a parametrization of 
$\overline{\Delta_{112}}$:
\begin{equation*}
(k,h) \mapsto (e,f,g)=(-6k^2+h,-2k(4k^2-h),-k^2(3k^2-h))
\end{equation*}
yielding the degenerate fibration $Y_h$:
\begin{equation*}
\begin{cases}
e=-6k^2+h \\
f=-2k(4k^2-h)\\
g=-k^2(3k^2-h)
\end{cases}\quad.
\end{equation*}
Here, $k$, $h$ are general sections of $\cL$, $\cL^2$, respectively.
The fiber is of type $\Delta_{112}$ over $B\smallsetminus \uO$,
and generally of type $\Delta_{13}$ over $\uO$, degenerating 
further to type $\Delta_4$ over $h=k=0$.

As in Sen's limit, perturbing~$Y_h$ in the simplest way by adding
linear terms in $C$ to $f$ and $g$ leads to a situation for which
$j\propto h^3$, unless a simple constraint between these terms is
satisfied.  We satisfy this constraint and add a $C^2$ term to
preserve generality, obtaining the following family $Y_h(C)$:
\begin{equation*}
\begin{cases}
e=-6k^2+h \\
f=-2k(4k^2-h)+2C\phi \\
g=-k^2(3k^2-h)+2C k \phi+C^2 \gamma
\end{cases}\quad.
\end{equation*}
For this family,
\begin{equation*}
j\sim \frac 1{C^2}\cdot \frac{h^4}{(h-4k^2)(\phi^2-h\gamma)}
\end{equation*}
at leading order in $C$, with limiting discriminant
\begin{equation}
\label{E7Delh}
\Delta_h = h^2(h-4k^2)(\phi^2-h\gamma)\quad.
\end{equation}
The geometry of this situation has similarities with Sen's limit, and
also leads to an Euler characteristic/Chern class identity, in which
the singularities of the support of~$\Delta_h$ play a key role; see
\S\ref{IntroE7Lim2} and \S\ref{E711213}.  In the double cover
$X:\xi^2=h$, the brane spectrum will consist on an orientifold, a
smooth brane-image-brane pair resulting from the splitting of $h-4
k^2$, and a singular `Whitney brane' (see \S\ref{Senswcl}) corresponding 
to the component
$(\phi^2-h \gamma)$ of the discriminant. These singularities are
typical of Sen's limit.

\subsection{An $E_6$ weak coupling limit}\label{E6lim}
The situation for $E_6$ fibrations is in principle richer. There are
several allowed specializations, shown in Figure~\ref{E6spec};
\begin{figure}
\includegraphics[scale=1]{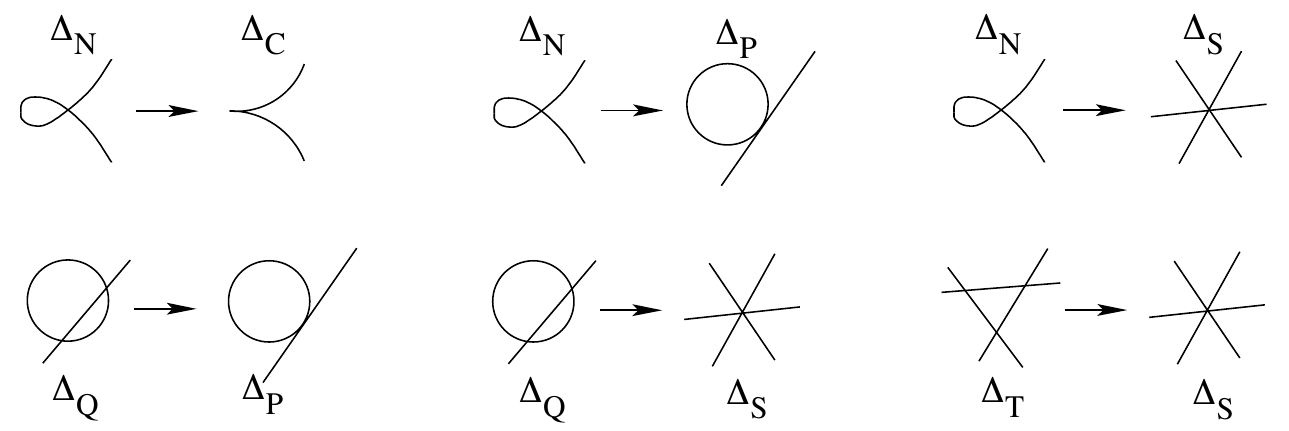}
\caption{Allowed specializations for $E_6$ fibrations\label{E6spec}} 
\end{figure}
we focus on the $\Delta_Q\to \Delta_P$ case. Recall that $Q$ and $P$
each consists of three components; according to \eqref{Qeq} (choosing
$\rho=1$), equations for one component of $Q$ are
\begin{equation}
\label{Qeqp}
f=e \quad,\quad
g =\frac{d(9 e-d^2)}{27}
\end{equation}
and equations for the corresponding component of $P$ are (from
\eqref{Peq})
\begin{equation}
\label{Peqp}
f=e=\frac{d^2}4 \quad,\quad
g =\frac{5d^3}{108}\quad.
\end{equation}
Again we take the parameters $d,e,f,g$ of the fibrations as
coordinates of a parameter space, and study the loci defined
by~\eqref{Qeqp} and~\eqref{Peqp} in this space.  We note that
\eqref{Peqp} can be parametrized as follows:
\begin{equation*}
k \mapsto (d,e,f,g) = (6k, 9k^2, 9k^2, 10k^2)\quad;
\end{equation*}
and this is easily extended to a parametrization of \eqref{Qeqp}:
\begin{equation*}
(k,h) \mapsto (d,e,f,g) = (6k, 9k^2+3h, 9k^2+3h, 2k(5k^2+3h))\quad.
\end{equation*}
Accordingly, we choose the degenerate fibration $Y_h$ to be given by
\begin{equation}
\label{E6Yh}
\begin{cases}
d=6k \\
e = 9k^2+3h \\
f = 9k^2+3h \\
g = 2k(5k^2+3h)
\end{cases} \quad.
\end{equation}
Here, $k$ is a general section of $\cL$, and $h$ is a general section
of $\cL^2$; once more, this fulfills one of the basic requirements of
\S\ref{strate}. The fibers over points on $\uO=\{h=0\}$ are of the
unstable type $\Delta_P$ (they degenerate further to type $\Delta_S$
where both $h$ and~$k$ vanish); over $B\smallsetminus \uO$, they are
of type $\Delta_Q$, as required.

Next, we perturb \eqref{E6Yh} in order to construct an associated
general family $Y_h(C)$. Since $d$ and $e$ may already be assumed to
be general (as $k$ and $h$ are general), we only need to perturb $f$
and $g$:
\begin{equation*}
\begin{cases}
d=6k \\
e = 9k^2+3h \\
f = 9k^2+3h+C\phi \\
g = 2k(5k^2+3h)+C\chi
\end{cases} \quad.
\end{equation*}
With these positions, the discriminant $\Delta_h(C)$ is
\[
\Delta_h(C)\sim C^2\, h^2(h+3k^2)((\chi-k \phi)^2-h\phi^2)
\]
at leading order in $C$. We set $\gamma=\chi-\phi k$, to simplify this
expression (without affecting geometric properties of the family), and
we obtain the following description for the general family $Y_h(C)$:
\begin{equation*}
\begin{cases}
d=6k \\
e = 9k^2+3h \\
f = 9k^2+3h+C\phi \\
g = 2k(5k^2+3h)+C(\gamma+k\phi)
\end{cases} \quad.
\end{equation*}
For this family,
\begin{equation*}
j\sim \frac 1{C^2}\cdot \frac{h^4}{(h+3k^2)(\gamma^2-h\phi^2)}
\end{equation*}
to leading order in $C$. Thus, all requirements set forth in
\S\ref{strate} are satisfied, with limiting discriminant
\begin{equation}
\label{E6Delh}
\Delta_h = h^2 (h+3k^2)(\gamma^2-h\phi^2)\quad.
\end{equation}
The D-brane spectrum was determined in \S\ref{IntroE6Lim} and consists
of an orientifold and two brane-image-brane pairs; see~\S\ref{E6QP}
for a discussion of the tadpole relation for this limit.

\subsection{Discussion}\label{stratecomm}
In this subsection we motivate the choices leading to the strategy
explained in \S\ref{strate}. These are principally the following:
\begin{itemize}
\item we specialize a {\em semistable\/} fiber $S_1$ to an {\em
unstable\/} fiber $S_2$;
\item the parameter $h$ controlling the degenerate fibration $Y_h$ is
a section of $\cL^2$; the fiber of this degenerate fibration is $S_1$
at all points of $B$ where $h$ does {\em not\/} vanish;
\item $j\propto h^4$ at leading order in $C$, up to a factor depending
on the other components of the limiting discriminant.
\end{itemize}

For $Y_h(C)$,
\begin{equation*}
j\propto \frac{F^3}{4 F^3+27 G^2}\quad,
\end{equation*}
where $F$ and $G$ are the parameters of the Weierstrass model of
$Y_h(C)$; $F$ is a section of $\cL^4$, and $G$ is a section of
$\cL^6$.  The term of lowest order in $C$ in the denominator defines
the limiting discriminant, a section of $\cL^{12}$. We factor it as
$h^a \uD$ for some $a\ge 0$ and a term $\uD$ which is not a multiple
of $h$.

The fact that the general degenerate fiber $S_1$ is semistable
guarantees that $F$ is not a multiple of $C$, so that generally $j\to
\infty$ as $C\to 0$:
\begin{equation*}
j\sim \frac 1{C^r}\cdot \frac{F^3|_{C=0}}{h^a \uD}
\end{equation*}
at leading order in $C$, for some $r>0$.  Since the fiber of $Y_h$ is
semistable where $h\ne 0$, it follows that $F|_{C=0}$ divides a power
of $h$, so in fact $F|_{C=0}$ equals a power of $h$ up to a constant
multiple.

In particular, a power of $h$ is a section of $\cL^4$; the most
natural way to achieve this is by assuming that $h$ is itself a
section of $\cL^c$ for some $c\,|\,4$: $c=1$, $2$, or $4$.  Now note
that since $C$ divides $4 F^3+27 G^2$, it follows that $G|_{C=0}$ is
also a power of $h$, and by the same token we see that $c\,|\,6$. Thus
necessarily $c=1$ or $2$.  As our basic geometric construction demands
that we take a {\em double\/} cover of $B$ ramified along the
hypersurface $\uO$ given by $h=0$, it is natural to assume that $c=2$:
this justifies our requirement that $h$ be a section of $\cL^2$.  Thus
\begin{equation*}
j\sim \frac 1{C^r}\cdot \frac{h^6}{h^a\uD}\quad.
\end{equation*}
Note that in particular $F|_{C=0}=0$ over points of the hypersurface
$\uO$ defined by $h=0$; thus, the special degenerate fiber $S_2$ is
necessarily unstable, as we prescribe.

Finally, we take into account the D7 tadpole relation
(\eqref{D7tadpole} in \S\ref{conscheck}): $\uD$ must be a section of
$\cL^8$ if $O$ corresponds to a pure orientifold, and as $h^a \uD$ is
a section of $\cL^{12}$ and $h$ is a section of $\cL$, this forces
$a=2$. Thus $j\propto h^4$ in the pure orientifold case, and this
justifies our last requirement.

It is in fact interesting to consider cases in which at $h=0$ we have
a bound state of one orientifold and $n$ brane-image-branes: one
example has already been illustrated in \S\ref{needfornew}, and
several will be listed in \S\ref{morelimits}.  In this case, $\uD$ is
a multiple of $h^n$, and $j\sim h^{4-n}$. The case $n=0$ gives $j\sim
h^4$ and corresponds to a pure orientifold.  The extreme case $n=4$
corresponds to $\Delta\sim h^6$, giving $4$ brane-image-branes on top
of the orientifold. This can be obtained as specialization of several limits 
listed in \S\ref{morelimits}.


\section{Tadpole relations}\label{Tadrel}
In this section we review intersection-theoretic aspects of the limits
found in \S\ref{Wcl}. As we will explain in \S\ref{tadmat}, tadpole
matching in type IIB and F-theory leads to identities relating the
Euler characteristic of the total space $Y$ of the fibration, and loci
in the double cover $X$ of $B$ ramifed along the hypersurface $\uO$
with equation $h=0$.  These loci are the preimages in $X$ of the
components of the limiting discriminant $\Delta_h$: the preimage $O$
of $\uO$, where the orientifold is located, and the preimage $D$ of
the other components $\uD$.

The identities predicted by the physics considerations of
\S\ref{tadmat} hold when $X$ and~$Y$ are Calabi-Yau varieties, of
dimension $3$ and~$4$, respectively. The results in this section will
show that these predictions have a much broader range of validity:
suitably interpreted, the identities will hold in arbitrary dimension,
and for varieties that are not necessarily Calabi-Yau.

We note here that the Calabi-Yau case corresponds to the case in which
the basic line bundle $\cL$ on $B$ is the anticanonical line bundle.
Here and in the following, we denote by $c_i(V)$ the $i$-th Chern
class $c_i(TV)$ of the tangent bundle of the variety $V$.

\begin{prop}\label{CYcond}
With $X$ and $Y$ as above, and with notation as in \S2 and \S3, $X$
and $Y$ are both Calabi-Yau varieties if $c_1(\cL)=c_1(B)$.
\end{prop}
In fact, $Y$ is a Calabi-Yau if $c_1(\cL)=c_1(B)$; $X$ is a Calabi-Yau
if $[\uO]=2c_1(B)$. Thus, in the Calabi-Yau case this gives a second
motivating reason to require $h$ to be a section of $\cL^2$, beside
the more general considerations of \S\ref{stratecomm}.

The proof of Proposition~\ref{CYcond} is standard \cite{Sen:1997gv},
so we leave it to the reader\footnote{As an illustration, consider the
$E_6$ case. Then $Y$ is defined by a divisor of class
$c_1(\cO(3)\otimes \cL^3)$ in the projectivization $\pi: \Pbb(\cE) \to
B$, where $\cE=\cO\oplus \cL\oplus \cL$.  By e.g., B.5.8 in
\cite{MR85k:14004},
\begin{equation*}
c_1(\Pbb(\cE))=\pi^*c_1(B)+\pi^* c_1(\cE)+3 c_1(\cO(1))
=\pi^* c_1(B)+2\pi^* c_1(\cL)+3 c_1(\cO(1))\quad;
\end{equation*}
by the adjunction formula, therefore, 
\begin{equation*}
c_1(Y)=(\varphi^*c_1(B)+\varphi^* c_1(\cE)+3 c_1(\cO(1)))-
(3 c_1(\cL)+3 c_1(\cO))
=\varphi^* (c_1(B)-c_1(\cL))\quad.
\end{equation*}
It follows that $Y$ is a Calabi-Yau if $c_1(\cL)=c_1(B)$. 

Concerning $X$, it is obtained as the double cover ramified over
$\uO$; it follows that $2c_1(X)$ agrees with the pull-back of
$2c_1(B)-[\uO]$, so $c_1(X)=0$ if $2c_1(B)=[\uO]$. As $h$ is a section
of $\cL^2$ (see \S\ref{stratecomm}), this shows that $X$ is a
Calabi-Yau when $c_1(\cL)=c_1(B)$.}.

While the case $c_1(\cL)=c_1(B)$ is therefore the most significant for
physical applications, we stress that the identities we will study in
this section will in fact hold without this additional hypothesis.

\subsection{F-theory-type-IIB tadpole matching condition}\label{tadmat}
F-theory is a strongly coupled version of type IIB string theory. It
realizes geometrically the $S$-duality group of type IIB.  One of the
strongest constraints in compactified theory is the cancellation of
the tadpole. This can be seen as a direct consequence of Gauss's law:
the total charge in a compact space has to vanish since the flux has
nowhere to go.

In the context of string theory with D-branes, the tadpole condition
requires the vanishing of all D-brane charges. This is complicated by
the fact that higher dimensional D-branes admit lower brane charges
due to Chern-Simons terms included to cancel anomalies. In the context
of type IIB with O3/O7 planes, this is especially important for D7
branes and O7 planes wrapping divisors of the Calabi-Yau
three-fold. They have an induced D3 charge given by the Euler
characteristic of the divisor they wrap. The D3 tadpole also has
contributions coming from world-volume fluxes located on the seven
branes.  If we ignore the world volume fluxes, the type IIB D3 tadpole
induced by the D7 branes and the O7 plane is:
\begin{equation}\nonumber
\text{Type IIB}:\quad N_{D3}=\frac{1}{2}\big(\sum_i  
\frac{4\chi(O_i)}{24}+\sum_j 
\frac{\chi(D_j)}{24}\big).
\end{equation}
In F-theory, the choice of the D7 branes is given by the geometry of
the Calabi-Yau four-fold. The D3 tadpole induced by the 7-branes is
given solely in terms of the Euler characteristic of the four-fold:
\begin{equation}\nonumber
\text{F-theory}:\quad N_{D3}=  \frac{\chi(Y)}{24}.
\end{equation}
Since the D3 branes are S-dual invariant, the D3-tadpole should be the
same in F-theory and in type IIB. In the absence of fluxes, this gives
a topological relation between the Euler characteristic of the
four-fold and the Euler characteristics of the orientifold planes and
the D7 branes. This is the {\em F-theory-type-IIB D3 tadpole matching
relation} \cite{Collinucci:2008pf, AlEs}:
\begin{equation}\nonumber
\text{F-theory-type IIB D3 matching}:\quad 2\chi(Y)=4\sum_i
\chi(O_i)+\sum_j \chi(D_j).
\end{equation}
This relation has been checked to hold in the case of Sen's weak
coupling of F-theory compactified on a Calabi-Yau four-fold given by a
Weierstrass model (elliptic fibration of type $E_8$), once the
singularities found in Sen's limit are taken into account as providing
a contribution to the corresponding Euler characteristic
\cite{Collinucci:2008pf, AlEs}.  An F-theory-type IIB tadpole mismatch
signals that one has to consider adding world-volume fluxes on the
D7-branes. This is reviewed in \cite{Collinucci:2008pf}.

Now that we have introduced several new weak coupling limits of
F-theory, it is natural to check the validity of this relation. Since
some of our limits lead to configurations involving only smooth
divisors, the F-theory tadpole matching condition is even more
constrained since there are no singularities providing contributions
to the Euler characteristic to correct a possible mismatch.

An interesting aspect of the tadpole matching condition in the
examples we have examined is that it
follows from a Chern class identity largely independent of the specific 
hypotheses leading to it:
 it has been
shown in the $E_8$-case to hold for elliptic fibrations of any
dimension and without imposing the Calabi-Yau condition, and it is
valid at the level of Chern classes.  We will see that all these
properties hold as well for the $E_6$ and $E_7$ examples obtained in
\S\ref{Wcl}.

The left-hand side of the F-theory-type-IIB D3 tadpole matching
relation involves the Euler characteristic of the elliptic
fibration. This can be expressed solely in terms of the geometry of
the base. A formula computing the Euler characteristic of an elliptic
fibered four-fold was first given by Sethi-Vafa-Witten
\cite{Sethi:1996es}.  It was then extended to $E_7$ and $E_6$
fibrations by Klemm, Lian, Roan and Yau \cite{Klemm:1996ts}.  An
equivalent formula for singular four-folds was obtained by Andreas and
Curio \cite{Andreas:1999ng}.  In the $E_8$ case, the relation valid
outside of the Calabi-Yau condition and at the level of Chern class
was derived in \cite{AlEs}.

In this section we will generalize Sethi-Vafa-Witten at the level of
the total Chern class for elliptic fibrations of type $E_6$, $E_7$,
$E_8$ of any dimension. The case of $E_8$ was already treated in
\cite{AlEs}.  We will then use the generalized Sethi-Vafa-Witten
formula to prove the F-theory-type-IIB tadpole matching condition for
the different limits we have obtained in the previous section. We will
see that it holds for the smooth configurations in $E_7$ and $E_6$,
and also for the $E_7$ limits involving Whitney umbrella
singularities.  For the latter we use the Euler characteristic
$\chi_o$ introduced in \cite{AlEs,Collinucci:2008pf} and its
corresponding Chern class.

The new elliptic fibrations we consider here allow for configurations
of smooth branes.  In \S\ref{Class.smooth} we provide a theorem which
classifies all possible types of smooth brane configurations
satisfying the tadpole matching condition in a universal sense, for
$E_8$, $E_7$, $E_6$ fibrations over a base of arbitrary dimension.  In
particular, we find that no such configurations can occur in the $E_8$
case, over an unrestricted base.  We observe that restricting to the
Calabi-Yau case in dimension~$3$ potentially allows for more types of
configurations of smooth branes satisfying the tadpole relation.

\subsection{Sethi-Vafa-Witten formulas}\label{SVWfors}
Let $Y\to B$ be an elliptic fibration of the kind considered in
\S\ref{Sinfib}.  Under the assumption that $Y$ is a Calabi-Yau
manifold, formulas for the Euler characteristic of $Y$ may be found in
\cite{Klemm:1996ts}, (3.11) and \S\S7~and~8; some of these formulas
were also given in \cite{Sethi:1996es}.
\begin{prop}[\cite{Sethi:1996es}, \cite{Klemm:1996ts}]\label{SVW}
Let $Y\to B$ be an elliptic fibration of type $E_6$, $E_7$, $E_8$, with
$Y$ a Calabi-Yau fourfold. Then
\begin{equation*}
\begin{cases}
E_6: & \chi(Y) =\int 12 c_1(B) c_2(B)+72 c_1(B)^3 \\
E_7: & \chi(Y) =\int 12 c_1(B) c_2(B)+144 c_1(B)^3 \\
E_8: & \chi(Y) =\int 12 c_1(B) c_2(B)+360 c_1(B)^3 
\end{cases}
\end{equation*}
\end{prop}
Here $c_i(B)$ denotes the Chern class $c_i(TB)$ of the tangent bundle,
and $\int$ denotes degree in~$B$.  These formulas suffice as one
ingredient in the verification of the relations arising from tadpole
matching in the limits considered in this paper, in the physically
more immediately significant case of Calabi-Yau fourfolds.

One point of our message is however that these relations, inspired by
the physically significant cases, have a much wider range of validity.
In order to realize this, we have to upgrade Proposition~\ref{SVW} to
smooth fibrations involving varieties of arbitrary dimension, and for which
$Y$ is not necessarily a Calabi-Yau. Also, we have to produce
identities at the level of total Chern classes; the Euler
characteristic is just the degree of the piece of dimension zero in
the total Chern class.

Surprisingly, the general formulas we obtain are in fact {\em
simpler\/} than the formulas recalled in Proposition~\ref{SVW}, in the
sense that they have a more conceptual geometric formulation, and they
do not directly involve random-looking coefficients. They extend
directly the observation (hidden in the formulas of
Proposition~\ref{SVW}) that the Euler characteristic of a fibration
$Y$ as in \S\ref{Sinfib} equals a fixed multiple of the Euler
characteristic of the hypersurface $g=0$ of $B$.  In this formulation,
the Calabi-Yau condition is irrelevant.

\begin{theorem}\label{SVWup}
Let $\varphi: Y\to B$ be an elliptic fibration of type $E_6$, $E_7$,
$E_8$, as in \S\ref{Sinfib}. Let $Z$ be a smooth hypersurface of $B$
in the same class as $g=0$ (with notation as in \eqref{E6eq},
\eqref{E7eq}, and~\eqref{E8eq}). Then
\begin{equation*}
\varphi_*(c(Y)) = m\cdot c(Z)
\end{equation*}
where $m=4,3,2$, resp.~for $Y$ of type $E_6$, $E_7$, $E_8$, resp.
\end{theorem}

A straightforward computation shows that Theorem~\ref{SVWup} implies
the formulas given in Proposition~\ref{SVW}; for this, use that $\int
\varphi_*(c(Y))=\int c(Y)=\chi(Y)$ and $\int c(Z)=\chi(Z)$, the
adjunction formula, and Proposition~\ref{CYcond}.

\begin{remark}\label{SVWrem}
The ingredients appearing in Theorem~\ref{SVWup} admit the following
compelling geometric description. The fibrations considered here have
$m-1$ distinguished sections, obtained by intersecting with $z=0$ the
loci specified by \eqref{E6eq}, \eqref{E7eq}, \eqref{E8eq}.  Over the
divisor $Z$ with equation $g=0$ there is in fact one more section,
corresponding to the point $(z:x:y) =(1:0:0)$ in each fiber. The
content of Theorem~\ref{SVWup} is that the complement of these $m$
sections of $Z$ in $Y$ contributes $0$ to the push-forward of
$c(Y)$. This is intriguing, because the fibers of the projection from
this complement to $Y$ do not all have vanishing Euler characteristic.
\end{remark}

\begin{proof}[Proof of Theorem~\ref{SVWup}]
This is again a relatively straightforward computation, using
adjunction and standard intersection-theoretic calculus.  We will
sketch the $E_6$ case, as an illustration. Recall that $Y$ is defined
by equation~\eqref{E6eq} in $\Pbb(\cO\oplus \cL \oplus \cL)$.  Denote
by $\pi$ the projection from this bundle to $B$; also, denote by $L$,
$H$, resp., the classes $c_1(\cL)$, $c_1(\cO(1))$, resp., and their
pull-backs.  Standard arguments give
\begin{equation*}
c(Y)=\frac{(1+H)(1+L+H)^2}{1+3H+3L}(3H+3L)\, \pi^*(c(B))
\end{equation*}
as a class in $\Pbb(\cO\oplus \cL\oplus \cL)$, and hence by the
projection formula
\begin{equation*}
\varphi_*(c(Y))=\pi_*\left(\frac{(1+H)(1+L+H)^2}{1+3H+3L}(3H+3L)
\right) c(B)\quad.
\end{equation*}
On the other hand, $g$ is a section of $\cL^3$ in the $E_6$ case; thus
\begin{equation*}
c(Z)=\frac{3L}{1+3L}\, c(B)\quad.
\end{equation*}
Thus, to prove the statement in the $E_6$ case it suffices to show that
\begin{equation}
\label{SVWid}
\pi_*\left(\frac{(1+H)(1+L+H)^2}{1+3H+3L}(3H+3L)\right)
=4\cdot \frac{3L}{1+3L}\quad.
\end{equation}
Again by the projection formula, performing this verification amounts
to computing the push forwards $\pi_*(H^i)$. For these, note that
\begin{equation*}
\pi_*(1+H+H^2+\cdots)=s(\cO\oplus \cL\oplus \cL)=
\frac 1{(1+L)^2}\cap [B]  =(1-2L+3L^2-4L^3+\cdots) \cap [B]:
\end{equation*}
the first equality holds by definition of Segre class (\S3.1 in
\cite{MR85k:14004}), and the second holds because the Segre class is
the inverse of the Chern class (\S3.2 in \cite{MR85k:14004}).
Matching terms of like dimension, we see that $\pi_*$ acts as follows:
\begin{equation*}
1 \mapsto 0;\quad
H \mapsto 0;\quad
H^2\mapsto 1;\quad
H^3\mapsto -2L;\quad
H^4\mapsto 3L^2;\quad
\text{etc.}
\end{equation*}
This reduces \eqref{SVWid} to an elementary verification.

The proofs in the $E_7$ and $E_8$ case follow the same
pattern\footnote{In \cite{AlEs} we gave a different argument for the
$E_8$ case, based on the study of the explicit contributions of
singular fibers.}.
\end{proof}

\subsection{Euler characteristic and Chern class identities}\label{Chernids}
We now move on to the verification of the tadpole relation for the
limits considered in \S\ref{Wcl}. The $E_8$ case of Sen's limit is
treated in detail in \cite{AlEs}, and we briefly review the result
here for the convenience of the reader. The second $E_7$ limit
encountered in \S\ref{E7lim} has similarities with Sen's from the
geometric point of view; we will show that the same notion of Euler
characteristic (and Chern class) used in \S\ref{Senstad} leads to a
correct identity in this case. For the other two limits obtained in
\S\ref{Wcl} the support of the branes is nonsingular, so the
verification of the identities is (even) more straightforward.

The identities involve the fibration $Y$ and hypersurfaces in the
double cover $X$ of $B$ ramified along $\uO$; we take $h=\xi^2$ as an
equation for $X$ in the total space of $\cL$; $\xi=0$ defines the zero
section in the total space, and the ramification locus $O$ on $X$.
The orientifold in the theory is supported on $O$. The various branes
are supported on the other components $D$ of the preimage of the
limiting discriminant~$\Delta_h$.

\subsubsection{Sen's $E_8$ limit.}\label{Senstad}
Pulling back the equation of $\Delta_h$ to $X$ gives
\begin{equation*}
\xi^4(\eta^2+12 \xi^2 \chi)\quad:
\end{equation*}
we find one $D7$-brane $D$ supported on the hypersurface $\eta^2+12
\xi^2\chi=0$.  This is a singular variety, with a Whitney umbrella
structure along the singular locus $\eta=\xi=0$, pinched along
$\chi=\eta=\xi=0$ (see Figure~\ref{Senspic}): thus, $D$ is a Whitney
D7 brane.
\begin{figure}
\includegraphics[scale=.6]{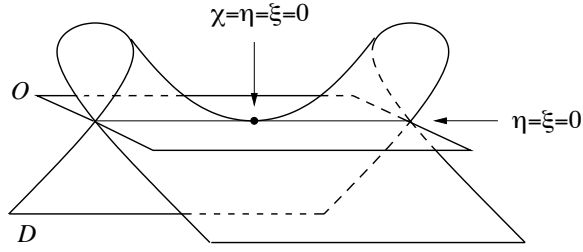}
\caption{Sen's $E_8$ limit\label{Senspic}} 
\end{figure}
In \cite{AlEs} we prove that 
\begin{equation*}
2\chi(Y)=4\chi(O)+\chi_o(D)\quad,
\end{equation*}
where $\chi_o(D)$ is defined as the Euler charactistic of the
normalization $\oD$ of $D$, corrected by a contribution of the pinch
locus $S$. More precisely (\cite{AlEs}, Corollary~4.7),
\begin{equation}
\label{tadrelE8chi}
2\chi(Y)=4\chi(O)+\chi(\oD)-\chi(S)\quad;
\end{equation}
in the case $\dim B=3$, the locus $S$ consists of isolated points, and
$\chi(S)$ is simply the number of such points. However,
\eqref{tadrelE8chi} holds regardless of the dimension of $B$. In fact,
\eqref{tadrelE8chi} is just the degree of the term of dimension~$0$ in
an identity involving the total Chern classes of these loci:
\begin{equation}
\label{tadrelE8}
2\varphi_*(c(Y))=4 \rho_*(c(O))+\pi_*(c(\oD))-i_*(c(S))\quad,
\end{equation}
where the push-forwards all bring the named classes to the homology of
$B$ (\cite{AlEs}, Theorem~4.6).

\subsubsection{The $\Delta_{112} \to \Delta_{13}$ $E_7$ limit}\label{E711213}
The limiting configuration for the second $E_7$ limit obtained in
\S\ref{E7lim} has points of similarity with Sen's limit. In this case,
$\Delta_h$ is defined by~\eqref{E7Delh},
\begin{equation*}
h^2(h-4k^2)(\phi^2-h\gamma)\quad,
\end{equation*}
and hence it pulls back to
\begin{equation*}
\xi^4(\xi-2k)(\xi+2k)(\phi^2-\xi^2\gamma)\quad.
\end{equation*}
Thus, we find the orientifold $O$, together with a brane $D_1$
splitting into a brane-image-brane, and a Whitney D7-brane $D_2$.  The
situation is analogous to that in Sen's limit, but two additional
smooth hypersurfaces $D_{1+}$, $D_{1-}$ are present, which are
interchanged by the symmetry $\xi\mapsto -\xi$. Note that these two
hypersurfaces are in the same class as $\xi$.

\begin{claim}\label{E7tadrelchi}
$2\chi(Y)=4\chi(O)+\chi(D_{1+})+\chi(D_{1-})+\chi_o(D_2)\quad.$
\end{claim}
Here, $\chi_o(D_2)$ is defined in precisely the same fashion as in
Sen's limit.  The Euler characteristics of $D_{1+}$ and $D_{1-}$ are
ordinary Euler characteristics, since these components are
smooth. They in fact equal $\chi(O)$.

Claim~\ref{E7tadrelchi} is obtained by taking the degree of the term
of dimension~$0$ in the following identity of Chern classes, which
holds in arbitrary dimension:
\begin{equation}
\label{E7tadrel}
2\varphi_*(c(Y))=\rho_*(4c(O)+c(D_+)+c(D_-))+\pi_*(c(\oD))-i_*(c(S))
\quad,
\end{equation}
where again $S$ denotes the pinch locus, and the push-forwards map
the classes to the homology of $B$.

To verify this identity, first note $O$, $D_+$, and $D_-$ are smooth
and have the same class in $X$; thus $4c(O)+c(D_+)+c(D_-)=6 c(O)$.
Since $O$ maps isomorphically to $\uO$ in $B$,
\begin{equation*}
\rho_*(4 c(O)+c(D_+)+c(D_-))=6\, c(\uO)=6\cdot \frac{2L}{1+2L}\cdot c(B)
\quad.
\end{equation*}
Concerning the other terms, the normalization $\oD$ of the Whitney
umbrella maps generically $2$-to-$1$ onto the corresponding component
$\uD'$ of $\Delta_h$, with equation
\begin{equation*}
\phi^2-h\gamma=0\quad,
\end{equation*}
and maps $1$-to-$1$ over the pinch locus $S$. It follows that
\begin{equation*}
\pi_*(c(\oD))-i_*(c(S))=2 \csm(\uD')-2 c(S)\quad,
\end{equation*}
where $\csm(\uD')$ is the Chern-Schwartz-MacPherson class of $\uD'$
(see \S4.2 of \cite{AlEs}). The class $\csm(\uD')$ is computed in
Lemma~4.4 of \cite{AlEs}; considering that $\uD'$ has class $6L$, and
the class of $S$ is $(2L)(3L)(4L)$, we get
\begin{equation*}
\csm(\uD')=\left(\frac{6L}{1+6L} - \frac 1{1+6L} \frac{(2L)(3L)(4L)}
{(1+2L)(1+3L)(1+4L)}\right)c(B)\quad,
\end{equation*}
and hence $\pi_*(c(\oD))-i_*(c(S))$ equals{\small
\begin{multline}
2\left(\frac{6L}{1+6L} - \frac 1{1+6L} \frac{(2L)(3L)(4L)}
{(1+2L)(1+3L)(1+4L)}\right)c(B)-2 \frac{(2L)(3L)(4L)}
{(1+2L)(1+3L)(1+4L)}c(B)\\
=\frac{12L}{(1+2L)(1+4L)}c(B)
\end{multline}}
Summarizing, the right-hand side of \eqref{E7tadrel} equals
\begin{equation}
\label{rhsE7}
6\cdot \frac{2L}{1+2L}\, c(B) + \frac{12L}{(1+2L)(1+4L)}\,c(B)
=\frac{24L}{1+4L}\,c(B)\quad.
\end{equation}
On the other hand, Theorem~\ref{SVWup} shows
\begin{equation*}
\varphi_*(c(Y))=3\cdot c(Z)\quad,
\end{equation*}
where $Z$ is defined by $g=0$, and hence has class $4L$. Thus, the
left-hand side of \eqref{E7tadrel} is
\begin{equation}
\label{lhsE7}
2\cdot 3\cdot \frac {4L}{1+4L}\,c(B)\quad.
\end{equation}
The expressions \eqref{rhsE7} and~\eqref{lhsE7} match, completing the
proof of \eqref{E7tadrel} and verifying Claim~\ref{E7tadrelchi}.

\subsubsection{The $\Delta_{22} \to \Delta_{4}$ $E_7$ limit}\label{E7224}
The first limit obtained in \S\ref{E7lim} has a simpler geometric
description.  The limiting discriminant is defined by~\eqref{E7Delh1}:
\begin{equation*}
h^2 \left(\gamma^2-\phi^2 h\right)
\end{equation*}
and pulls back on $X$ to
\begin{equation*}
\xi^4 \left(\gamma+\phi \xi\right)\left(\gamma-\phi \xi\right)\quad.
\end{equation*}
In this case, the whole brane contribution to the configuration is
split into an image-brane-image pair.  We let $D_+$, $D_-$ denote
these two components. All three components $O$, $D_+$, $D_-$ are
smooth, of class $L$, $4L$, $4L$ respectively.
\begin{figure}
\includegraphics[scale=.5]{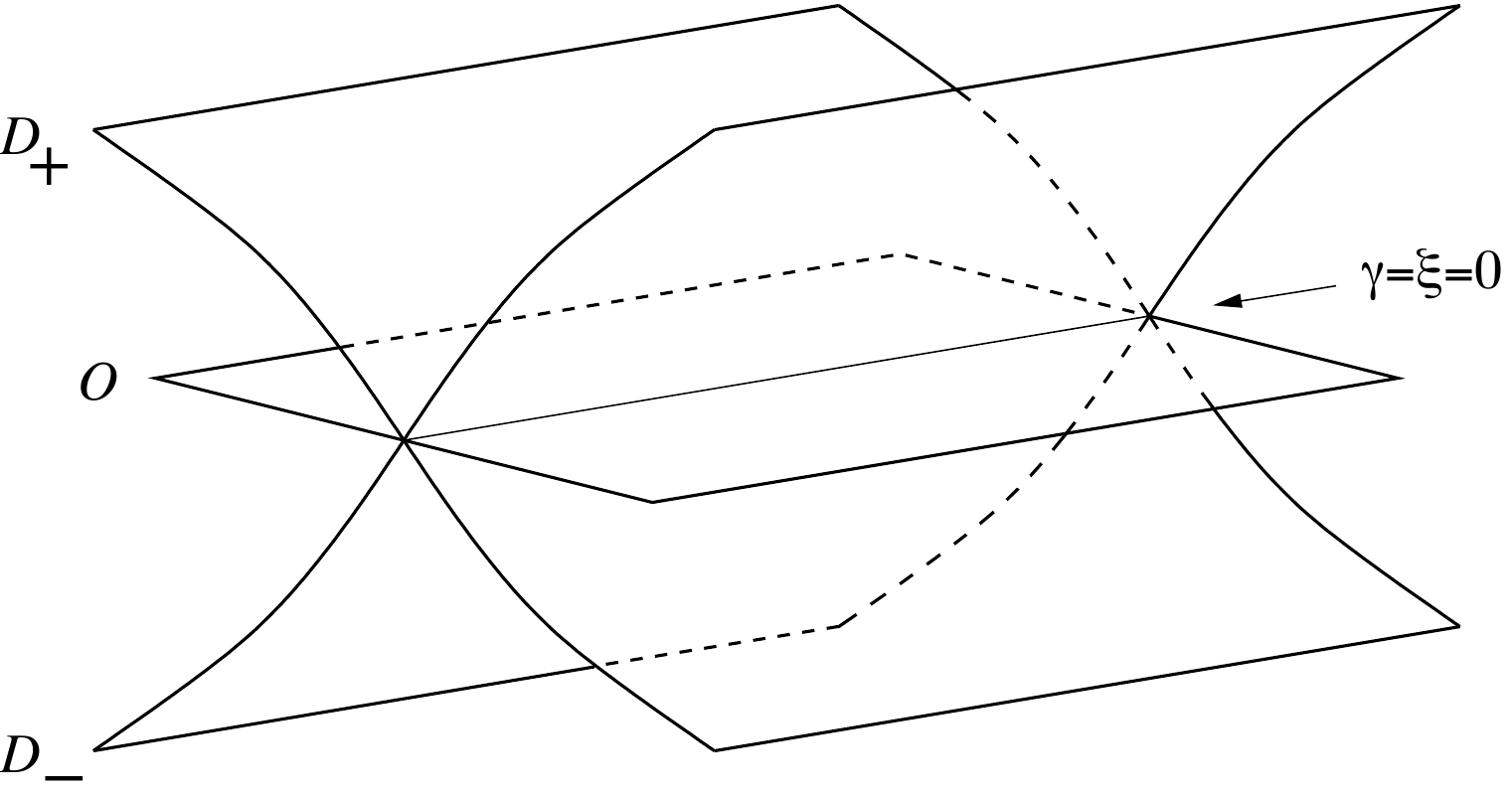}
\caption{$E_7$ limit, $\Delta_{22}\to \Delta_4$\label{E7lpic}} 
\end{figure}
The hypersurfaces $D_+$ and~$D_-$ are tangent to each other along
$\xi=\gamma=\phi=0$, and otherwise meet transversally.

\begin{claim}\label{E7tadrelchi1}
$2 \chi(Y)=4 \chi(O)+\chi(D_+)+\chi(D_-)$.
\end{claim}

Again, this identity is just a manifestation of an identity of Chern
classes, holding in any dimension and without any Calabi-Yau
restriction:
\begin{equation}
\label{E7tadrel1}
2\varphi_*(c(Y))=\rho_*\left( 4 c(O)+c(D_+)+c(D_-)\right)\quad.
\end{equation}

Verifying \eqref{E7tadrel1} is essentially straightforward. Since $O$
maps isomorphically to $\uO$,
\begin{equation*}
\rho_* c(O)=c(\uO)=\frac{2L}{1+2L} c(B)\quad.
\end{equation*}
Each of $D_+$ and $D_-$ is a smooth hypersurface in $X$, of class $4L$,
hence
\begin{equation*}
c(D_{\pm})=\frac{4L}{1+4L} c(X)=\frac{4L}{1+4L} \cdot \frac{1+L}{1+2L} 
\cdot \rho^*(c(TB)) \cap[X]
\quad.
\end{equation*}
By the projection formula:
\begin{multline}
\label{pfDpm}
\rho_* c(D_{\pm})=\frac{4L}{1+4L} c(X)=\frac{4L(1+L)}{(1+2L)(1+4L)}  
c(TB) \cap {(2[B])}\\
=\frac{8L(1+L)}{(1+2L)(1+4L)} c(B)
\quad.
\end{multline}
Therefore, the right-hand side of \eqref{E7tadrel1} equals
\begin{equation*}
\left(4\cdot \frac{2L}{1+2L} + 
2 \cdot \frac{8L(1+L)}{(1+2L)(1+4L)}\right) c(B)
=\frac{24 L}{1+4L} \, c(B)\quad.
\end{equation*}
This equals the left-hand side by Theorem~\ref{SVWup}
(cf.~\eqref{lhsE7}), completing the verification.

\subsubsection{The $\Delta_Q \to \Delta_P$ $E_6$ limit}\label{E6QP}
Our last example is the $E_6$ limit found in \S\ref{E6lim}. The
limiting discriminant is defined by~\eqref{E6Delh}:
\begin{equation*}
h^2 (h+3k^2)(\gamma^2-h\phi^2)\quad.
\end{equation*}
and pulls back to
\begin{equation*}
\xi^4 (\gamma+\xi \phi)(\gamma-\xi \phi)(\xi+i\sqrt 3k)(\xi-i\sqrt 3k)\quad.
\end{equation*}
This configuration consists of the orientifold $O$ and of {\em two\/}
brane-image-brane pairs: a first pair $D_{1\pm}$, with components of
class~$3L$, and a second pair $D_{2\pm}$, with components of
class~$L$. The second pair is transversal, while $D_{1+}$ and $D_{1-}$
are tangent along $\xi=\gamma=\phi =0$.

\begin{claim}\label{E6tadrelchi}
$2 \chi(Y)=4 \chi(O)+\chi(D_{1+})+\chi(D_{1-})+\chi(D_{2+})+\chi(D_{2-})\quad.$
\end{claim}

Once more, this identity of Euler characteristics is simply a
numerical avatar of a general identity of Chern classes, which holds
in arbitrary dimension and regardless of Calabi-Yau restrictions:
\begin{equation}
\label{E6tadrel}
2\varphi_*(c(Y))=\rho_*\left(4c(O)+c(D_{1+})+c(D_{1-})+c(D_{2+})+c(D_{2-})\right)\quad.
\end{equation}

The verification of \eqref{E6tadrel} follows the same guidelines as
the verification of \eqref{E7tadrel1} given above. Since $O$,
$D_{2\pm}$ are smooth hypersurfaces of the same class, mapping
isomorphically to $\uO$,
\begin{equation*}
\rho_*(4c(O)+c(D_{2+})+c(D_{2-}))=6 \rho_*(c(O))=6\cdot \frac {2L}{1+2L} c(B)
\quad.
\end{equation*}
The components $D_{1\pm}$ have class $3L$ in $X$, and their
contribution can be computed by the same method used for $D_\pm$ in
\S\ref{E7224}: this gives
\begin{equation*}
\rho_* c(D_{1\pm})=\frac{3L(1+L)}{(1+2L)(1+3L)}  c(TB)
\cap {(2[B])}
=\frac{6L(1+L)}{(1+2L)(1+3L)} c(B)
\quad.
\end{equation*}  
Therefore, the right-hand side of \eqref{E6tadrel} equals
\begin{equation}
\label{rhsE6}
\left( 6\cdot \frac {2L}{1+2L}+2\cdot \frac{6L(1+L)}{(1+2L)(1+3L)}\right) 
c(B) =\frac{24 L}{1+3L}\, c(B)\quad.
\end{equation}
By Theorem~\ref{SVWup}, the left-hand side of \eqref{E6tadrel} equals
\begin{equation*}
\varphi_*(c(Y))=4\cdot c(Z)\quad,
\end{equation*}
where $Z$ is given by $g=0$. In the $E_6$ case $g$ is a section of
$\cL^3$, so this gives
\begin{equation*}
2 \varphi_*(c(Y))=2\cdot 4\cdot \frac{3L}{1+3L}
\end{equation*}
with the same result as \eqref{rhsE6}, concluding the verification.

\subsection{A classification of configurations of smooth 
branes}\label{Class.smooth}

The method used in \S\ref{Chernids} to verify the tadpole relations
can be used to explore the set of possible orientifold/brane
configurations satisfying the tadpole relation, under the assumption
that the supports of all branes are {\em nonsingular.\/}

For this, we propose the Ansatz that the tadpole relation, while
originally obtained under specific hypotheses satisfied in the
motivating physical setting, reflects a {\em universal\/} identity,
holding regardless of the specific choice of $L$ and of any
restriction posed on the nonsingular base variety $B$. This is the
case in the examples examined in \S\ref{Chernids}. While the original
tadpole relation is a specific statement about the geometry of
Calabi-Yau fibrations, the universal tadpole relation is a formal
identity of Euler characteristics, computed as functions of a variable 
$L$ and of the Chern classes $c_1(B), c_2(B), \dots$ of the base.

The verifications carried out in \S\ref{Chernids} show that the
tadpole relation holds universally in the examples obtained in
\S\ref{Wcl}.  In the $\Delta_{22} \to \Delta_4$ $E_7$ case and in the
$\Delta_Q \to \Delta_P$ $E_6$ case, the branes are supported on {\em
nonsingular\/} components of the limiting discriminant. We can ask
whether there are other such `nonsingular' configurations satisfying
the universal tadpole relation.  In the other two situations
encountered in \S\ref{Chernids}, a {\em singular\/} brane appears. Is
this a necessity, or a pathology?

Informally, we find the following:
\begin{claim}\label{nogoclaim}
\begin{itemize}
\item There are {\em no\/} configurations of smooth branes satisfying
the universal tadpole relation for $E_8$ fibrations;
\item The only configuration of smooth branes satisfying the universal
tadpole relation for $E_7$ fibrations is the one arising from the
$\Delta_{22} \to \Delta_4$ specialization, examined in \S\ref{E7224};
\item The only configuration of smooth branes satisfying the universal
tadpole relation for $E_6$ fibrations is the one arising from the
$\Delta_{Q} \to \Delta_P$ specialization, examined in \S\ref{E6QP}.
\end{itemize}
\end{claim}

Thus, the appearance of the singular $D7$-brane in Sen's limit and in
the $E_7$ limit studied in \S\ref{E711213} is not accidental.

Note that the Calabi-Yau condition is not used in this statement. It
is interesting to study configurations that satisfy the tadpole
relation universally but subject to a Calabi-Yau hypothesis; this
narrows the scope of the tadpole relation, so it can potentially be
satisfied by more configurations. We will discuss this case at the end
of the section.

We formalize Claim~\ref{nogoclaim} as follows. We consider
configurations of (not necessarily distinct) hypersurfaces
$D_1,\dots,D_r$ of the double cover $X$ of $B$ ramified along $O$.  We
assume that $D_i$ has class $a_i L$, and that $\sum a_i =12$: this is
the case if the divisors appear as components of a limit discriminant.
Typically, the orientifold $\xi=0$ corresponds to $4$ components
$D_i$, with $a_i=1$, but we do not impose this as a necessary
condition.

{\em Assume all hypersurfaces $D_i$ are nonsingular.\/} A strong 
tadpole relation then takes the form
\begin{equation}
\label{preSTR}
2 \varphi_* c(Y) = \sum_{i=1}^r \rho_*(c(D_i))\quad.
\end{equation}
Again, typically $D_1=D_2=D_3=D_4$ would be the support of the
orientifold $O$, so this relation, after taking degrees, would amount
to
\begin{equation}
\label{STRchi}
2\chi(Y) = 4 \chi(O) + \sum_{i>4} \chi(D_i)\quad,
\end{equation}
the ordinary tadpole relation. Applying Theorem~\ref{SVWup},
\eqref{preSTR} may be rewritten
\begin{equation}
\label{STR}
2 m\, c(Z) = \sum_{i=1}^r \rho_*(c(D_i))\quad,
\end{equation}
where $Z$ is a nonsingular hypersurface in $B$, of class $a L$. We
have $(m,a)=(2,6)$, $(3,4)$, $(4,3)$ resp.~in the $E_8$, $E_7$, $E_6$
case, resp. We say that \eqref{STR} is `universally satisfied' if the
corresponding identity depending on $a$, $a_i$ is satisfied as a
formal identity of Chern classes, independently of the choice of $L$
  or $B$.

The precise version of Claim~\ref{nogoclaim} is as follows.

\begin{theorem}\label{nogothm}
Assume all hypersurfaces $D_i$ are nonsingular, and $D_i$ has class
$a_iL$, with $a_i>0$ integers and $\sum_i a_i=12$; assume $Z$ is also
nonsingular, and has class $aL$. Then \eqref{STR} is universally
satisfied only in the following three cases:
\[
\begin{tabular}{||c|c||}
\hline
$(m,a)$ & $(a_1,a_2,\dots)$ \\
\hline\hline
$(3,4)$ & $(1,1,1,1,4,4)$ \\
\hline
$(4,3)$ & $(1,1,1,1,1,1,3,3)$ \\
\hline
$(6,2)$ & $(1,\dots,1)$ \\
\hline
\end{tabular}
\]
In fact, the weaker requirement that \eqref{STRchi} is universally
satisfied for $\dim B=3$ suffices to draw the same conclusion.
\end{theorem}

The first and second case reproduce the configurations appearing in
\S\ref{E7224} and~\ref{E6QP}: the orientifold (of class $L$) appears
with multiplicity~$4$ in both cases; in the $E_7$ case, a
brane-image-brane appears with components of class $4L$; in the $E_6$
case, two brane-image-brane pairs appear, of class $L$ and $3L$,
respectively.

No case corresponding to $E_8$ fibrations appears in the list. We do
not have a geometric interpretation for the third case displayed in
Theorem~\ref{nogothm}.

To prove Theorem~\ref{nogothm}, rewrite \eqref{STR} in terms of $L$
and $c(B)$: by essentially the same arguments used in
\S\ref{Chernids}, this yields
\begin{equation}
\label{STRpf}
m\cdot \frac{aL}{1+aL}
=\left(\sum_{i=1}^r \frac{a_i L}{1+a_i L}\right) \frac{1+L}{1+2L}
\end{equation}
up to a common factor of $c(B)$, which may be discarded as the
identity is required to hold universally. For the same reason,
\eqref{STRpf} must hold as an identity of rational functions in the
indeterminate~$L$.  Theorem~\ref{nogothm} is proved by verifying that
\eqref{STRpf} is satisfied in this sense only in the three cases
listed in the statement.

In fact, requiring that $\sum_i a_i=12$, multiplying both sides of
\eqref{STRpf} by $c(B)$, and imposing a match in degree~$3$ leads to
precisely the same list of cases.  Thus, the weaker requirement that
the tadpole relation should be universally satisfied at the level of
Euler characteristics for $\dim B=3$ (in the form~\eqref{STRchi})
constrains the situation in precisely the same way.

Carrying out this last verification under the additional assumption
that $L=c_1(TB)$ explores the possible configurations in the more
physically significant Calabi-Yau case
(cf.~Proposition~\ref{CYcond}). This produces three more
possibilities:
\[
\begin{tabular}{||c|c||}
\hline
$(m,a)$ & $(a_1,a_2,\dots)$ \\
\hline\hline
$(2,6)$ & $(1,1,1,1,1,7)$ \\
\hline
$(3,4)$ & $(1,1,1,1,1,1,1,5)$ \\
\hline
$(4,3)$ & $(1,2,2,2,2,3)$ \\
\hline
\end{tabular}
\]
The first case is the $E_8$ configuration described in \cite{AlEs},
\S3.  We stress that these configurations satisfy the tadpole relation
only in the Calabi-Yau case; this is a weaker requirement than that
imposed in Theorem~\ref{nogothm}.

It would be interesting to provide a geometric realization of the  
third configuration listed in Theorem~\ref{nogothm}.


\section{A visit to the zoo}\label{morelimits}

The requirement that $j\propto h^4$ adopted in \S\ref{Wcl} corresponds
to a pure O7 orientifold plane with no D7-brane on top of it.  As we
have observed at the end of \S\ref{stratecomm}, this requirement may
be weakened to $j\propto h^{4-n}$ for $0\le n\le 4$, in order to
accommodate bound states of the orientifold plane with $n$ pairs of D7
brane-image-brane on top of it. These configurations also lead to
Chern class identities.  For example the $E_6$ limit presented in
\S\ref{IntroE6Lim} and constructed in detail in \S\ref{E6lim}
specializes to a configuration with $j\propto h^3$ when $k=0$. Indeed,
in that case the D7 brane-image-brane $D_{2\pm}: \xi\pm i\sqrt{3} k=0$
reduces to $D_{2\pm }:\xi=0$ and therefore they are wrapping the same
locus as the orientifold plane. This explains the behavior $j\propto
h^3$.

There are other mechanisms producing different limits, corresponding
to different ways to satisfy the generality hypothesis (often by
introducing terms with higher $C$-weight).  Such operations can often
be interpreted as specializations of other known limits, and change
the underlying geometry in a controlled fashion. Again, it is often
possible to provide a physical interpretation for the resulting
configurations.  However, limits arising in this way are often
necessarily non-supersymmetric. 
Here we simply list several families $Y_h(C)$, with the corresponding
leading order in~$j$. 
The discriminant can be recovered from the $j$-invariant. If the $j$-invariant 
is given as an irreducible fraction $j\sim \frac{h^{4-n}}{C^r \Delta'}$, the 
discriminant is recovered as $\Delta= h^{2+n}\Delta'$. 
Note that several limits in the list presented here do not satisfy the double 
intersection properties. 

\subsection{$E_8$ limits}

\[
\begin{tabular}{||c| l  |c||}
\hline
\hline
 & & \\
Type & \multicolumn{1}{|c|}{ Realization} & $j$-invariant \\
& & \\
\hline
\hline
\multirow{4}{*}{\includegraphics[scale=0.25]{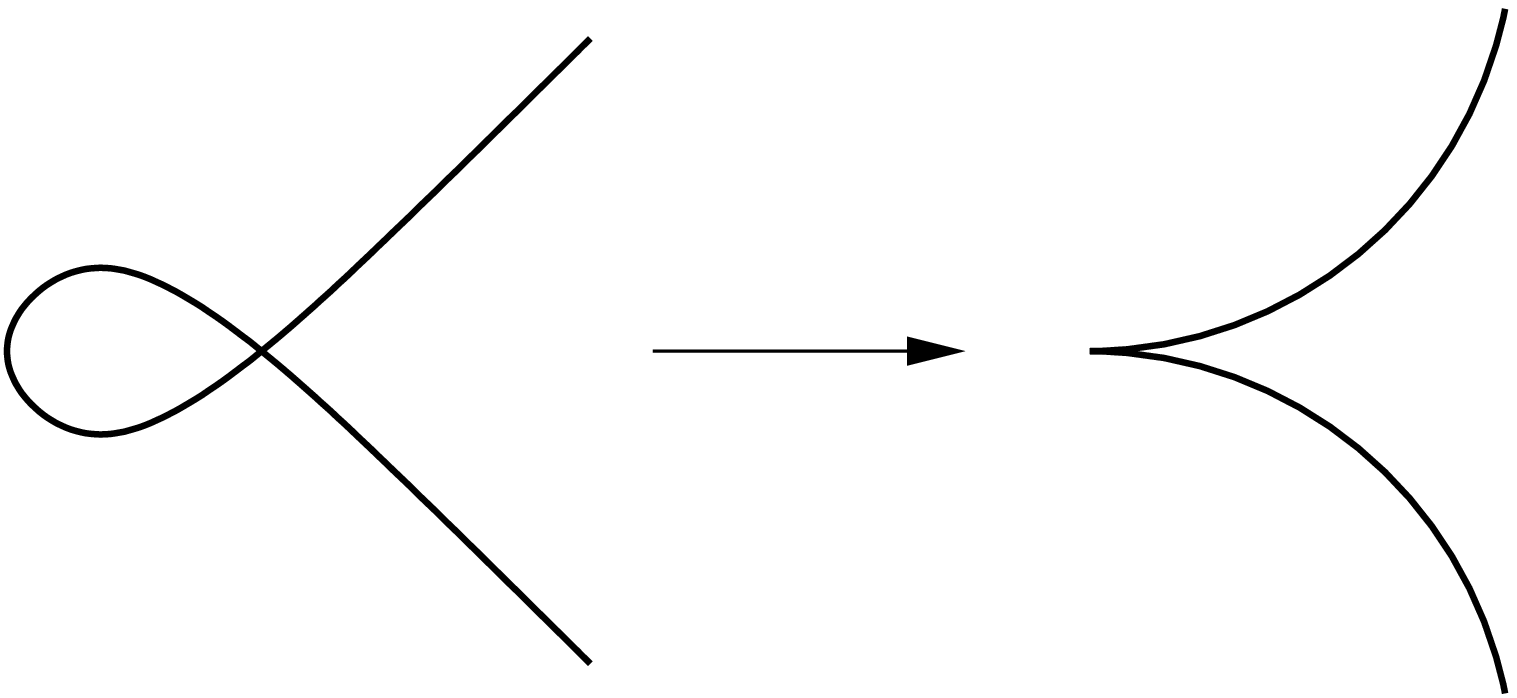}} 
& 
$\aligned
f &=-3 h^2+C \phi \\
g&=-2h^3 +C \gamma
\endaligned$
& $\displaystyle{\frac 1{C}\cdot \frac{h^3}{h\phi-\gamma}}$ \\
\cline{2-3}
& $\aligned
f &=-3 h^2 +C \phi \\
g &=-2h^3 +C h\phi +C^3 \gamma
\endaligned$
& $\displaystyle{\frac 1{C^2}\cdot \frac{h^4}{\phi^2}}$ \\
\hline
\end{tabular}
\]

\subsection{$E_7$ limits}

\[
\begin{tabular}{||c| l  |c||}
\hline
\hline
 & & \\
Type & \multicolumn{1}{|c|}{ Realization} & $j$-invariant \\
& & \\
\hline
\hline
\multirow{5}{*}{\includegraphics[scale=0.3]{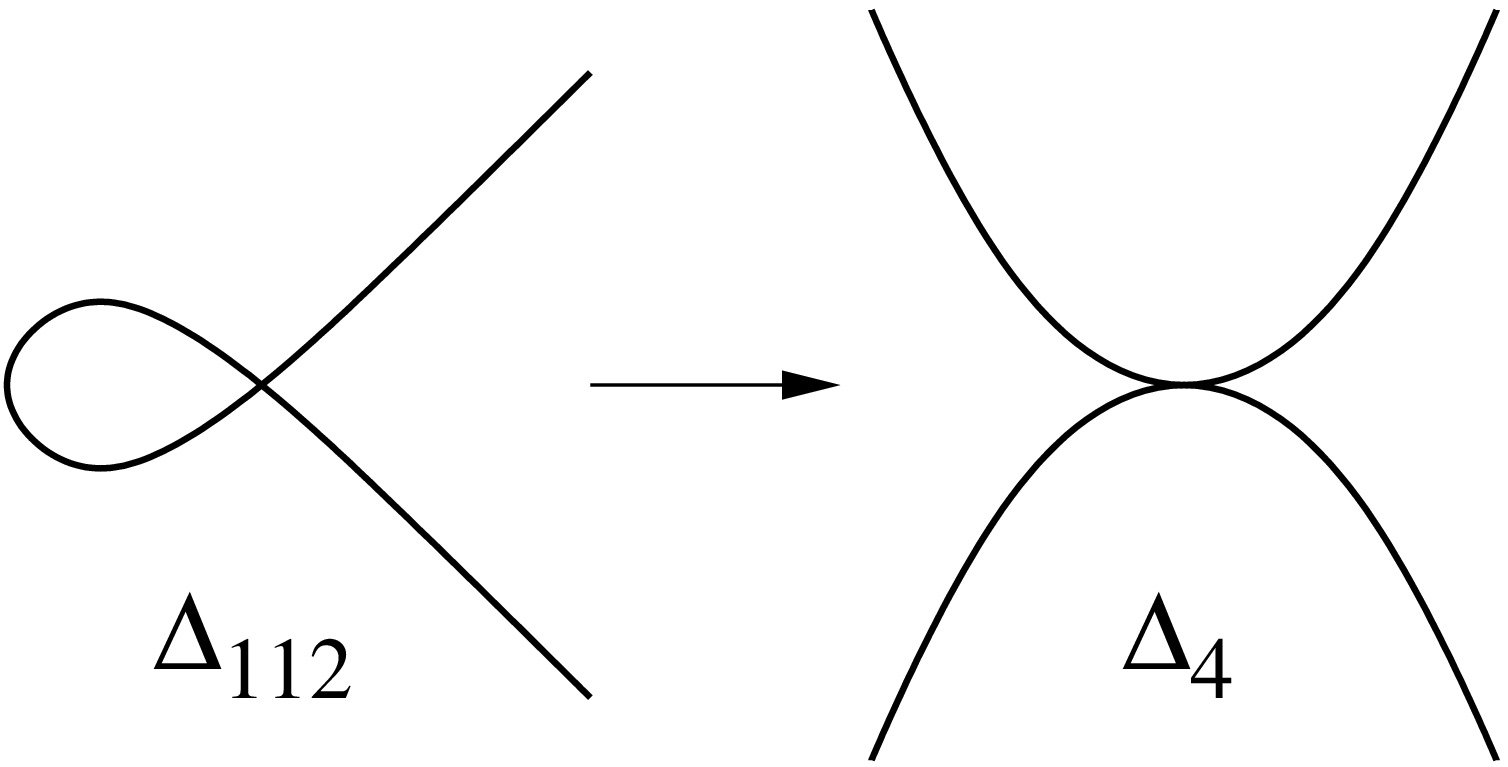}} 
& 
$\aligned 
e=&h \\
f= &C \phi \\
g= &C \gamma
\endaligned
$
& 
$\displaystyle{\frac 1{C}\cdot \frac{h^2}{\gamma}}$\\
\cline{2-3}
& 
$
\aligned
e&=h \\
f&= 2 C \phi \\
g&= C^2 \gamma
\endaligned$
& $\displaystyle{\frac 1{C^2}\cdot \frac{h^3}{h \gamma-\phi^2}}$\\
\hline
\hline
\multirow{9}{*}{ \includegraphics[scale=0.3]{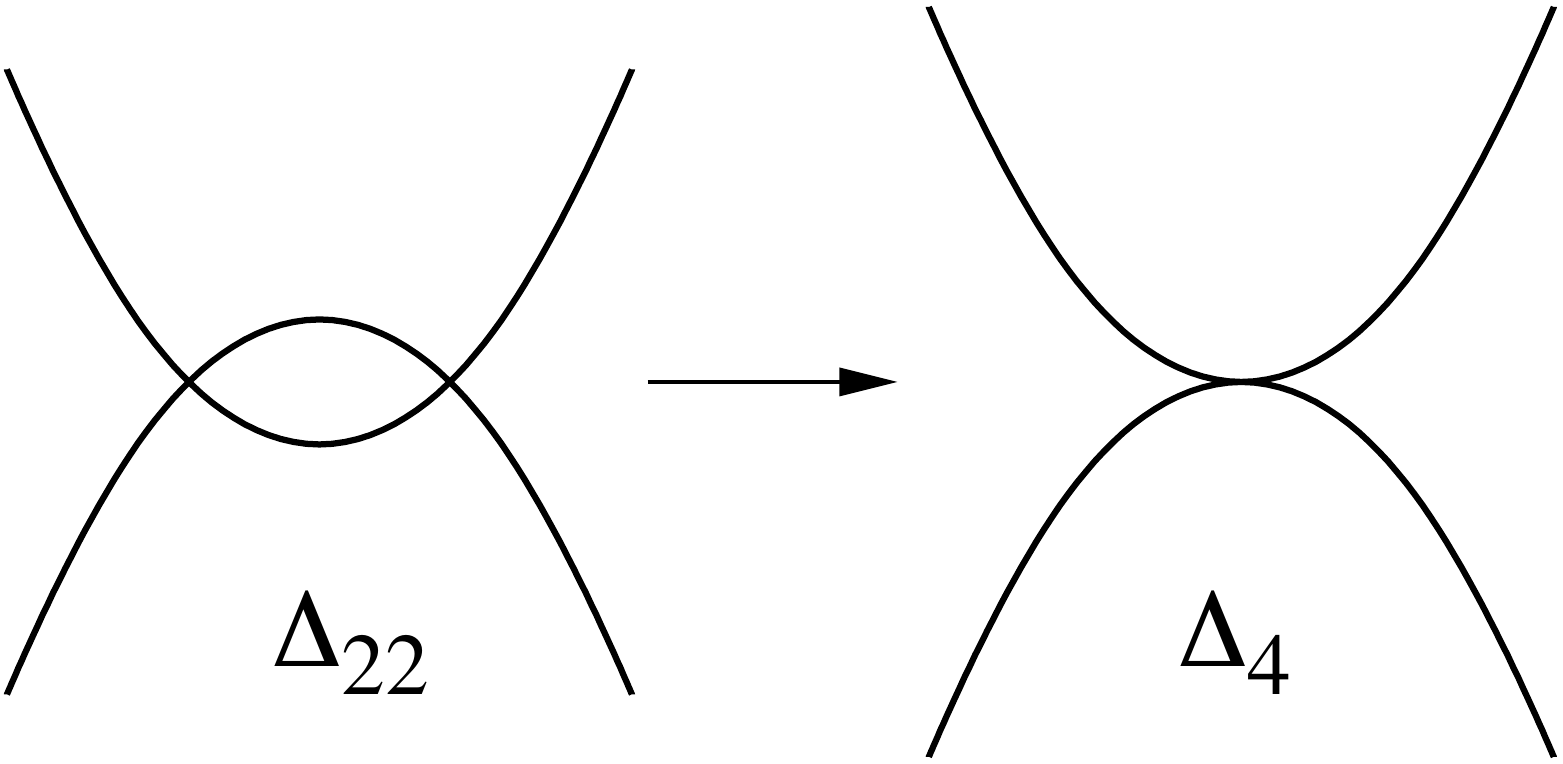}} & 
$
\aligned
e&=-2h \\
f&= Ch\eta+C^2 \phi \\
g&= h^2+Ch^2+C^2 \gamma
\endaligned$
& $\displaystyle{\frac 1{C^2}\cdot \frac{h}{(h-\eta^2)}}$\\
\cline{2-3}
& 
$
\aligned
e&=-2h \\
f&= Ch\eta+C^2 \phi \\
g&= h^2+C \gamma
\endaligned$
& $\displaystyle{\frac 1{C^2}\cdot \frac{h^4}{(\eta^2h^3-\gamma^2)}}$\\
\cline{2-3}
& 
$
\aligned
e&=2h \\
f&= C^2 \phi \\
g&= h^2+C \gamma
\endaligned$
& $\displaystyle{\frac 1{C^2}\cdot \frac{h^4}{\gamma^2}}$\\
\hline
\hline
\multirow{20}{*}{  \includegraphics[scale=0.35]{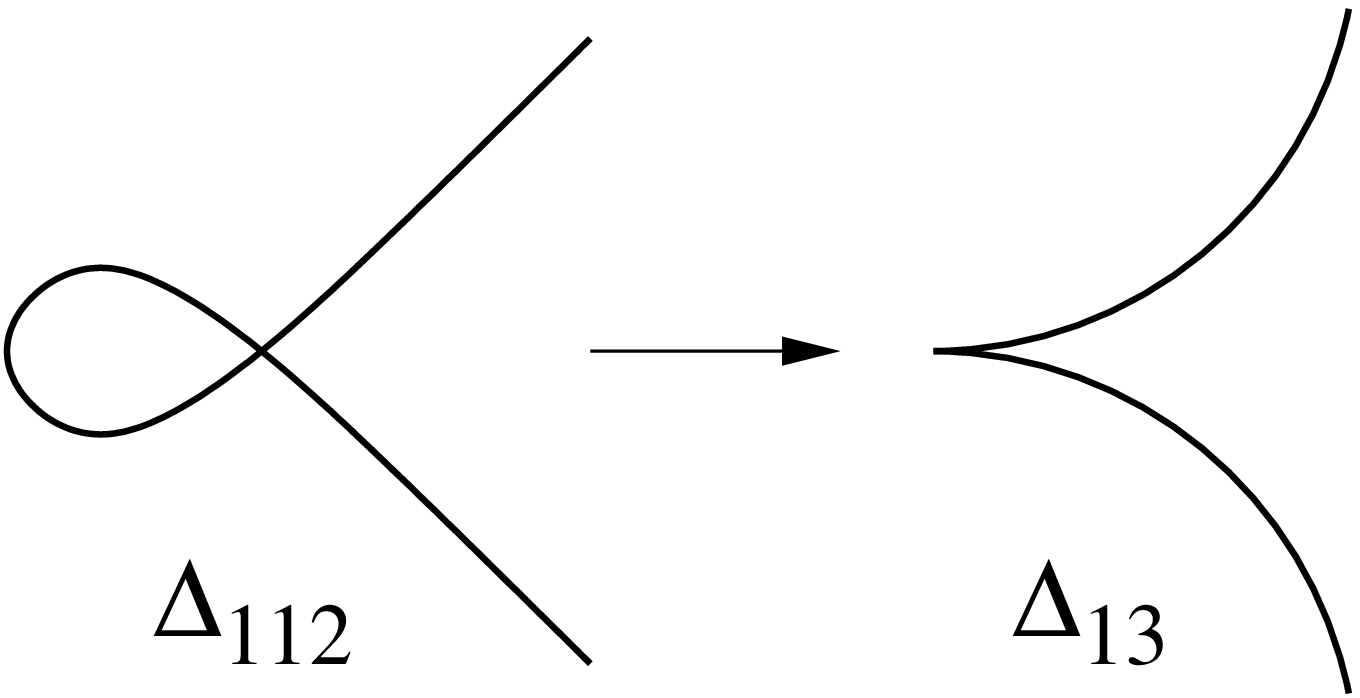}}
 &
  $
\aligned
e&=4 h-6k^2 \\
f&= 8k(h-k^2)+ C h \eta+C^2 \phi \\
g&= k^2(4h-3k^2)+Ch^2 +C^2\gamma
\endaligned$
& $\displaystyle{\frac 1{C}\cdot \frac{h^3}{(h-k^2)(h-k\eta)}}$\\
\cline{2-3}
& 
$
\aligned
e&=4 h-6k^2 \\
f&= 8k(h-k^2)+ C \phi \\
g&= k^2(4h-3k^2)+Ch^2 +C^2\gamma
\endaligned$
& $\displaystyle{\frac 1{C}\cdot \frac{h^3}{(h-k^2)(h^2-k\phi)}}$\\
\cline{2-3}
& 
$
\aligned
e&=4 h-6k^2 \\
f&= 8k(h-k^2)+ C h \eta+C^2 \phi \\
g&= k^2(4h-3k^2)+C\gamma
\endaligned$
& $\displaystyle{\frac 1{C}\cdot \frac{h^3}{(h-k^2)(\gamma-kh\eta)}}$\\
\cline{2-3}
& 
$
\aligned
e&=4 h-6k^2 \\
f&= 8k(h-k^2)+ C \phi \\
g&= k^2(4h-3k^2)+C \gamma
\endaligned$
& $\displaystyle{\frac 1{C}\cdot \frac{h^3}{(h-k^2)(\gamma-k\phi)}}$\\
\cline{2-3}
& 
$
\aligned
e&=4 h-6k^2 \\
f&= 8k(h-k^2)+ C \phi \\
g&= k^2(4h-3k^2)+C^3 \gamma
\endaligned$
& $\displaystyle{\frac 1{C}\cdot \frac{h^3}{k\phi(h-k^2)}}$\\
\cline{2-3}
 & 
$
\aligned
e&=4 h-6k^2 \\
f&= 8k(h-k^2)+ C \phi \\
g&= k^2(4h-3k^2)+Ck\phi+C^3 \gamma
\endaligned$
& 
$\displaystyle{\frac 1{C^2}\cdot \frac{h^4}{\phi^2(h-k^2)}}$\\
\hline
\end{tabular}
\]

\newpage
\subsection{$E_6$ limits}

\begin{center}
\[
\begin{tabular}{||c| l |c||}
\hline
\hline
 & & \\
Type & \multicolumn{1}{|c|}{ Realization} & $j$-invariant \\
& & \\
\hline
\hline
\multirow{12}{*}{  \includegraphics[scale=1.2]{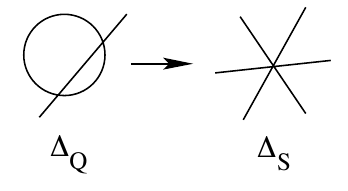}}
 & $
\aligned
d&=C\delta \\
e&=3h \\
f&= 3h+C \phi\\
g&= C\delta h+C^2 \gamma
\endaligned$
& $\displaystyle{\frac 1{C^2}\cdot \frac{h^2}{\phi^2}}$\\
\cline{2-3}
 & $
\aligned
d&=C\delta \\
e&=3h \\
f&= 3h+C \phi\\
g&= C(\delta+\eta) h+C^2 \gamma
\endaligned$
& $\displaystyle{\frac 1{C^2}\cdot \frac{h^2}{\phi^2-h\eta^2}}$\\
\cline{2-3}
 & $
\aligned
d&=C\delta \\
e&=3h \\
f&= 3h+C \phi\\
g&= C(h \delta + \gamma)
\endaligned$
& $\displaystyle{\frac 1{C^2}\cdot \frac{h^3}{\gamma^2-h \phi^2}}$\\
\hline
\hline
\multirow{16}{*}{  \includegraphics[scale=1.2]{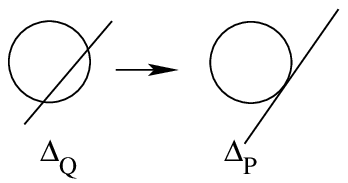}}
& $
\aligned
d&=6 k \\
e&=9k^2+3h \\
f&= 9k^2+3h +C\phi \\
g&= 2k(5k^2+3h)\\
&\   \  +C(h\eta+k\phi)+C^2\gamma
\endaligned$
& $\displaystyle{\frac 1{C^2}\cdot \frac{h^3}{(h\eta^2-\phi^2)(h +3k^2)}}$\\
\cline{2-3}
 & $
\aligned
d&=6 k \\
e&=9k^2+3h \\
f&= 9k^2+3h +C^2\phi \\
g&= 2k(5k^2+3h)+C\gamma
\endaligned$
& $\displaystyle{\frac 1{C^2}\cdot \frac{h^4}{\gamma^2 (h +3k^2)}}$\\
\cline{2-3}
& $
\aligned
d&=6 k \\
e&=9k^2+3h \\
f&= 9k^2+3h +C^2\phi \\
g&= 2k(5k^2+3h)\\
&\  \  +C(h\eta+k\phi)+C^2\gamma
\endaligned$
& $\displaystyle{\frac 1{C^2}\cdot \frac{h^4}{(h +3k^2)(h\eta+k\phi)^2}}$\\
\cline{2-3}
 & $
\aligned
d&=6 k \\
e&=9k^2+3h \\
f&= 9k^2+3h +C\phi \\
g&= 2k(5k^2+3h)+C^2\gamma
\endaligned$
& $\displaystyle{\frac 1{C^2}\cdot \frac{h^4}{\phi^2(h +3k^2)(h-k^2)}}$\\
\hline
\end{tabular}
\]
\end{center}

\bigskip


%
\newpage
\section{Conclusions and discussions}\label{Final.Conclusion}

In this paper, we have constructed new type IIB orientifold weak
coupling limits of F-theory.  This was done by considering elliptic
fibrations which are not in Weierstrass form.  The derivation of the
limits was streamlined by first introducing a geometric interpretation
of Sen's weak coupling limit in terms of a transition from semistable
singular fibers to unstable singular fibers.  This leads to a simple
algorithm to construct new limits from new families of elliptic
fibrations.  We would like to conclude with a discussion on different
issues raised by the existence of the new weak coupling limits and
with a look at future directions.

\subsubsection*{\bf New F-theory lift of type IIB and F-theory geometric 
engineering} 

The new weak coupling limits presented in this paper can
be used `in reverse' as different possible F-theory lifts for type IIB
orientifold compactifications.  This provides new tools for model
builders. The geometric intuition behind the construction of a given
weak coupling limit is a valuable asset for F-theory geometric
engineering.  Since each limit comes with a natural brane spectrum,
one can avoid intricate geometric tunings by starting with a family
that will more naturally lead to a brane spectrum close to the one we
would like to get in type IIB. For example, Sen's weak coupling limit
of a Weierstrass model favors type IIB configurations where branes
recombined into a unique irreducible and singular component, which we
refer to as a `Whitney D7-brane' since its singularities are
reminiscent of the Whitney umbrella $u^2-v^2 w=0$ in
$\mathbb{C}^3$. Naively, one might therefore conclude that F-theory
does not at all favor brane-image-brane pairs, in contrast to
assumptions generally made in the type IIB orientifold literature not
related to F-theory \cite{Jockers:2004yj}.  However, the existence of
new weak coupling limits shows that Sen's limit is not the unique weak
coupling limit of F-theory and other limits might allow supersymmetric
brane-image-brane configurations. This is exactly the case of elliptic
fibrations of type $E_6$ and $E_7$ since they naturally favors smooth
brane-image-brane pairs as illustrated in \S\ref{IntroExamples}.  The
D-brane deconstruction point of view \cite{Collinucci:2008pf}, which
is based on K-theory and is independent of taking a weak coupling
limit of F-theory, gives the general structure of D7-brane
singularities.  Brane-image-brane pairs and Whitney D7-branes are just
particular realizations of this general structure.

\subsubsection*{\bf  F-theory away from Weierstrass models}
Since any elliptic fibration is associated by means of a birational
transformation to a Weierstrass model, each of the new limits
presented here leads to a new weak coupling limit for a Weierstrass
model as well.  However, the resulting Weierstrass model is not
necessarily smooth; its Euler characteristic is different from the
Euler characteristic of the (smooth) fibration it is modeling, and
also different from the Euler characteristic of a smooth $E_8$
fibration on the same base and endowed with the same line bundle. In
particular, for F-theory compactified on a four-fold this implies that
all these models will in general have different D3 tadpoles. Elliptic
fibrations which are not in Weierstrass form have to be considered as
physically different, and interesting in their own right.  For
example, we show that supersymmetric brane-image-brane configurations
are allowed in F-theory for smooth $E_6$ and $E_7$ elliptic fibrations
whereas they are generally excluded for smooth Weierstrass models.

\subsubsection*{\bf F-theory-type-IIB D3 tadpole matching}
The duality between type IIB and F-theory requires that the D3 tadpole
is the same in both theory since D3 branes are invariant under
S-duality.  In the absence of fluxes, this leads to a simple relation
between the Euler characteristics of the four-fold and of the divisors
on which the branes and orientifold planes are wrapped.  As explained
in \cite{Collinucci:2008pf}, this relation can be used as an
indication for turning on world volume fluxes for certain
configuration of D-branes for which the F-theory-type-IIB D3 tadpole
matching does not hold in the absence of fluxes.  However, such fluxes
can induce a violation of the D-term constraints, for example for a
brane-image-brane pair except if the pair coincides with the
orientifold plane.  For this reason, brane-image-brane pairs in Sen's
limit of a Weierstrass model are generally non-supersymmetric
\cite{Collinucci:2008pf} since they require fluxes to satisfy the
F-theory-type IIB D3 tadpole relation. Here, we have constructed
limits leading to supersymmetric brane-image-brane
configurations. This was only possible by considering families for
which the F-theory-type-IIB D3 tadpole matching relation holds without
any need for fluxes, thanks to the fact that the Euler characteristic
of these elliptic fibrations (as computed in Proposition~\ref{SVW},
for $c_1(B)^3>0$) is lower than the one of a smooth Weierstrass model
constructed on the same base (in dimension~$3$, if $c_1(B)^3>0$). This
is a simple example on how the choice of an appropriate family of
elliptic fibrations is an important ingredient in the definition of a
F-theory lift of a type IIB orientifold compactification.

\subsubsection*{\bf   Singular branes and the Euler characteristic 
$\mathbf{\chi_o}$} 
Several of the limits we have constructed here admit only branes
wrapping smooth divisors in contrast to the case of Sen's limit of a
general Weierstrass model, which generally leads to singular divisors
of the Whitney umbrella type. We have also obtained a configuration
that mixes a Whitney brane and smooth brane-image-brane pairs.  We
have explicitly checked that when the divisor is singular, the Euler
chactacteristic $\chi_o$ introduced in \cite{Collinucci:2008pf, AlEs}
is the appropriate one to use since it ensures the matching of the D3
tadpole in type IIB and in F-theory.  We have presented a list of all
smooth configurations that would satisfy the F-theory-type-IIB D3
tadpole relation. Several elements of that list have been constructed
in this paper, but others have not, suggesting that there are still
interesting weak coupling limit to be discovered.

\subsubsection*{\bf  Double intersection property and D-brane motion}
The new limits presented here (in \S\ref{morelimits}) have some
intersection properties that do not match those of Sen's limit.  For
example, one would expect a 
 D-brane to intersect an orientifold with even multiplicity.  This is always
the case for Withney D7 branes, the typical example being Sen's limit
of a Weierstrass model. The double intersection property also holds
for brane-image-brane pairs.  The double intersection property is
supported by a probe argument which would imply that it should always
hold for O$7^-$ branes characterized by having a negative D7 charge
and orthogonal gauge group, when a stack of branes is on top of them  
\cite{Collinucci:2008pf}.  However, several of the limits listed in \S\ref{morelimits}
do not satisfy this property. Shall we disregard
such limits or do they correspond to some exotic type of branes and/or
orientifold planes?  Another aspect of Sen's limit of a Weierstrass
model is that the motion of the D7 branes is obstructed by the
structure of the singularities of the D7 brane locus
\cite{Braun:2008ua,Collinucci:2008pf}. Such obstructions are less
sever in $E_6$ and $E_7$ elliptic fibrations since one can have smooth
brane configurations.

\subsubsection*{ \bf String dualities and the  structure of elliptic fibrations}
The F-theory-type-IIB D3 tadpole matching is inspired by string theory
duality. However, the corresponding relation between Euler
characteristics is true in a more general context than its origin
might indicate. We have shown that for each family we have considered
there is a generalization of the F-theory-type-IIB D3 tadpole matching
condition valid at the level of the total Chern classes and
independent from the dimension of the base and the Calabi-Yau
condition.  In fact, in all the examples we have examined, when the
tadpole relation holds it does so as a formal identity in the Chern
classes of the base and of the chosen line bundle, independent of
finer geometric features of the varieties involved in the limit.

\subsubsection*{ \bf Open questions  and future directions}
There are several aspects of the constructions presented here that beg
for additional analysis.  For example, one would like to know if all
the smooth configurations presented in \S\ref{Class.smooth}, which
satisfy the F-theory -type-IIB D3 tadpole matching condition, can be
realized by a weak coupling limit.  It would also be interesting to
enlarge the possible F-theory lift of a type IIB configuration by
studying other families of elliptic fibrations than the $E_8$, $E_7$
and $E_6$ families studied here.  For example, the familly of elliptic
fibration of type $D_5$ in the notation of \cite{Klemm:1996ts} would
be a natural starting point.  It would also be important to determine
the relevance of the limits which lead to brane configurations which
do not satisfy the double intersection property that requires that a
D7 brane always have a double intersection with an O$7^-$ orientifold
plane.  We note that a careful classification of all types of
orientifold planes has been announced in \cite{Distler:2009ri}.  It is
clear that the D3 tadpole is not preserved by birational
transformations since the Euler characteristic can change if one of
the elliptic fibrations is singular or is not a Calabi-Yau. It would
be interesting to study the possible gauge groups associated with each
family. In particular, a discussion of the presence of exceptional
groups for families that are not in Weierstrass form would be
interesting in view of the central role of exceptional groups in
F-theory local model building of Grand Unified Theories (GUTs).

\vspace{1 cm}

{\em Acknowledgments.}  P.A.~thanks the Max-Planck-Institut f\"ur
Mathematik in Bonn for the support in the month of July 2009, when
most of this paper was written. He also thanks S.-T.~Yau, F.~Denef,
and the Jefferson Laboratory at Harvard for the hospitality during a
visit in May 2009, when this work was begun.  M.E. would like to
thanks D.~Bakary, F.~Denef, A.~Strominger, A.~Tomasiello and S.-T.~Yau
for their support.  He also thanks A.~P.~Braun, F.~Denef, A.~Klemm,
H.~Omer, A.~Tomasiello, M.~Wijnholt and H.~Yavartanoo for interesting
discussions.  M.E. is partially funded by a DOE grant  DE-FG02-91ER40654.


\begin{thebibliography}{10}

\bibitem{Vafa:1996xn}
C.~Vafa, ``{Evidence for F-Theory},''
  \href{http://dx.doi.org/10.1016/0550-3213(96)00172-1}{{\em Nucl. Phys.} {\bf
  B469} (1996)  403--418},
\href{http://arxiv.org/abs/hep-th/9602022}{ arXiv:hep-th/9602022}.

\bibitem{Morrison:1996na}
D.~R. Morrison and C.~Vafa, ``{Compactifications of F-Theory on Calabi--Yau
  Threefolds -- I},''
  \href{http://dx.doi.org/10.1016/0550-3213(96)00242-8}{{\em Nucl. Phys.} {\bf
  B473} (1996)  74--92},
\href{http://arxiv.org/abs/hep-th/9602114}{ arXiv:hep-th/9602114}.

\bibitem{Morrison:1996pp}
D.~R. Morrison and C.~Vafa, ``{Compactifications of F-Theory on Calabi--Yau
  Threefolds -- II},''
  \href{http://dx.doi.org/10.1016/0550-3213(96)00369-0}{{\em Nucl. Phys.} {\bf
  B476} (1996)  437--469},
\href{http://arxiv.org/abs/hep-th/9603161}{  arXiv:hep-th/960316}.


\bibitem{Schwarz:1995dk}
J.~H. Schwarz, ``{An SL(2,Z) multiplet of type IIB superstrings},''
  \href{http://dx.doi.org/10.1016/0370-2693(95)01138-G}{{\em Phys. Lett.} {\bf
  B360} (1995)  13--18},
\href{http://arxiv.org/abs/hep-th/9508143}{ arXiv:hep-th/9508143 }.

\bibitem{Schwarz:1995jq}
J.~H. Schwarz, ``{The power of M theory},''
  \href{http://dx.doi.org/10.1016/0370-2693(95)01429-2}{{\em Phys. Lett.} {\bf
  B367} (1996)  97--103},
\href{http://arxiv.org/abs/hep-th/9510086}{ arXiv:hep-th/9510086 }.

\bibitem{Douglas:1996du}
  M.~R.~Douglas and M.~Li,
  ``D-Brane Realization of N=2 Super Yang-Mills Theory in Four Dimensions,''
\href{http://arxiv.org/abs/hep-th/9604041} { arXiv:hep-th/9604041}.


\bibitem{Johansen:1996am}
A.~Johansen, ``{A comment on BPS states in F-theory in 8 dimensions},''
  \href{http://dx.doi.org/10.1016/S0370-2693(97)00053-1}{{\em Phys. Lett.} {\bf
  B395} (1997)  36--41},
\href{http://arxiv.org/abs/hep-th/9608186}{   arXiv:hep-th/9608186}.

\bibitem{Gaberdiel:1997ud}
M.~R. Gaberdiel and B.~Zwiebach, ``{Exceptional groups from open strings},''
  \href{http://dx.doi.org/10.1016/S0550-3213(97)00841-9}{{\em Nucl. Phys.} {\bf
  B518} (1998)  151--172},
\href{http://arxiv.org/abs/hep-th/9709013}{   arXiv:hep-th/9709013}.

\bibitem{Sethi:1996es}
S.~Sethi, C.~Vafa, and E.~Witten, ``{Constraints on low-dimensional string
  compactifications},''
  \href{http://dx.doi.org/10.1016/S0550-3213(96)00483-X}{{\em Nucl. Phys.} {\bf
  B480} (1996)  213--224},
\href{http://arxiv.org/abs/hep-th/9606122}{  arXiv:hep-th/9606122}.

\bibitem{Denef:2008wq}
F.~Denef, ``{Les Houches Lectures on Constructing String Vacua},''
\href{http://arxiv.org/abs/0803.1194}{  arXiv:0803.1194 [hep-th]}.

\bibitem{Blumenhagen:2008zz}
R.~Blumenhagen, V.~Braun, T.~W. Grimm, and T.~Weigand, ``{GUTs in Type IIB
  Orientifold Compactifications},''
  \href{http://dx.doi.org/10.1016/j.nuclphysb.2009.02.011}{{\em Nucl. Phys.}
  {\bf B815} (2009)  1--94},
\href{http://arxiv.org/abs/0811.2936}{arXiv:0811.2936 [hep-th]}.

\bibitem{Beasley:2008dc}
C.~Beasley, J.~J. Heckman, and C.~Vafa, ``{GUTs and Exceptional Branes in
  F-theory - I},'' \href{http://dx.doi.org/10.1088/1126-6708/2009/01/058}{{\em
  JHEP} {\bf 01} (2009)  058},
\href{http://arxiv.org/abs/0802.3391}{arXiv:0802.3391 [hep-th]}.

\bibitem{Beasley:2008kw}
C.~Beasley, J.~J. Heckman, and C.~Vafa, ``{GUTs and Exceptional Branes in
  F-theory - II: Experimental Predictions},''
  \href{http://dx.doi.org/10.1088/1126-6708/2009/01/059}{{\em JHEP} {\bf 01}
  (2009)  059},
\href{http://arxiv.org/abs/0806.0102}{ arXiv:0806.0102 [hep-th]}.

\bibitem{Heckman:2009mn}
J.~J. Heckman, A.~Tavanfar, and C.~Vafa, ``{The Point of E8 in F-theory
  GUTs},''
\href{http://arxiv.org/abs/0906.0581}{arXiv:0906.0581 [hep-th]}.

\bibitem{Heckman:2008rb}
J.~J. Heckman and C.~Vafa, ``{From F-theory GUTs to the LHC},''
\href{http://arxiv.org/abs/0809.3452}{arXiv:0809.3452 [hep-ph]}.

\bibitem{Donagi:2008ca}
R.~Donagi and M.~Wijnholt, ``{Model Building with F-Theory},''
\href{http://arxiv.org/abs/0802.2969}{arXiv:0802.2969 [hep-th]}.

\bibitem{Bourjaily:2009vf}
J.~L. Bourjaily, ``{Local Models in F-Theory and M-Theory with Three
  Generations},''
\href{http://arxiv.org/abs/0901.3785}{ arXiv:0901.3785 [hep-th]}.

\bibitem{Andreas:2009uf}
B.~Andreas and G.~Curio, ``{From Local to Global in F-Theory Model Building},''
\href{http://arxiv.org/abs/0902.4143}{arXiv:0902.4143 [hep-th]}.

\bibitem{Blumenhagen:2009up}
R.~Blumenhagen, T.~W. Grimm, B.~Jurke, and T.~Weigand, ``{F-theory uplifts and
  GUTs},''
\href{http://arxiv.org/abs/0906.0013}{arXiv:0906.0013 [hep-th]}.

\bibitem{Donagi:2009ra}
R.~Donagi and M.~Wijnholt, ``{Higgs Bundles and UV Completion in F-Theory},''
\href{http://arxiv.org/abs/0904.1218}{ arXiv:0904.1218 [hep-th]}.

\bibitem{Marsano:2009ym}
J.~Marsano, N.~Saulina, and S.~Schafer-Nameki, ``{F-theory Compactifications
  for Supersymmetric GUTs},''
\href{http://arxiv.org/abs/0904.3932}{arXiv:0904.3932 [hep-th]}.

\bibitem{Marsano:2009gv}
J.~Marsano, N.~Saulina, and S.~Schafer-Nameki, ``{Monodromies, Fluxes, and
  Compact Three-Generation F-theory GUTs},''
\href{http://arxiv.org/abs/0906.4672}{ arXiv:0906.4672 [hep-th]}.

\bibitem{Sen:1997gv}
A.~Sen, ``Orientifold limit of {$F$}-theory vacua,'' {\em Phys. Rev. D (3)}
  {\bf 55} (1997) no.~12, R7345--R7349.
 \href{http://arxiv.org/abs/hep-th/9702165}
{arXiv:hep-th/9702165}.

\bibitem{AlEs}
P.~Aluffi and M.~Esole, ``{Chern class identities from tadpole matching in type
  IIB and F-theory},''
  \href{http://dx.doi.org/10.1088/1126-6708/2009/03/032}{{\em JHEP} {\bf 03}
  (2009)  032},
\href{http://arxiv.org/abs/0710.2544}{ arXiv:0710.2544 [hep-th]}.

\bibitem{Collinucci:2008pf}
A.~Collinucci, F.~Denef, and M.~Esole, ``{D-brane Deconstructions in IIB
  Orientifolds},'' \href{http://dx.doi.org/10.1088/1126-6708/2009/02/005}{{\em
  JHEP} {\bf 02} (2009)  005},
\href{http://arxiv.org/abs/0805.1573}{ arXiv:0805.1573 [hep-th]}.

\bibitem{Braun:2008ua}
A.~P. Braun, A.~Hebecker, and H.~Triendl, ``{D7-Brane Motion from M-Theory
  Cycles and Obstructions in the Weak Coupling Limit},''
  \href{http://dx.doi.org/10.1016/j.nuclphysb.2008.03.021}{{\em Nucl. Phys.}
  {\bf B800} (2008)  298--329},
\href{http://arxiv.org/abs/0801.2163}{ arXiv:0801.2163 [hep-th]}.

\bibitem{Brunner:2008bi}
I.~Brunner and M.~Herbst, ``{Orientifolds and D-branes in N=2 gauged linear
  sigma models},''
\href{http://arxiv.org/abs/0812.2880}{ arXiv:0812.2880 [hep-th]}.

\bibitem{Collinucci:2008zs}
A.~Collinucci, ``{New F-theory lifts},''
\href{http://arxiv.org/abs/0812.0175}{ arXiv:0812.0175 [hep-th]}.

\bibitem{Collinucci:2009uh}
A.~Collinucci, ``{New F-theory lifts II: Permutation orientifolds and enhanced
  singularities},''
\href{http://arxiv.org/abs/0906.0003}{ arXiv:0906.0003 [hep-th]}.

\bibitem{Denef:2009ja}
F.~Denef, M.~Esole, and M.~Padi, ``{Orientiholes},''
\href{http://arxiv.org/abs/0901.2540}{ arXiv:0901.2540 [hep-th]}.

\bibitem{MR1714818}
V.~V. Batyrev, ``Birational {C}alabi-{Y}au {$n$}-folds have equal {B}etti
  numbers,'' in {\em New trends in algebraic geometry ({W}arwick, 1996)},
  vol.~264 of {\em London Math. Soc. Lecture Note Ser.}, pp.~1--11.
\newblock Cambridge Univ. Press, Cambridge, 1999.

\bibitem{math.AG/0401167}
P.~Aluffi, ``Chern classes of birational varieties,'' {\em Int. Math. Res.
  Not.} (2004) no.~63, 3367--3377.

\bibitem{MR2280127}
P.~Aluffi, ``Celestial integration, stringy invariants, and
  {C}hern-{S}chwartz-{M}ac{P}herson classes,'' in {\em Real and complex
  singularities}, Trends Math., pp.~1--13.
\newblock Birkh\"auser, Basel, 2007.

\bibitem{Bershadsky:1996nh}
  M.~Bershadsky, K.~A.~Intriligator, S.~Kachru, D.~R.~Morrison, V.~Sadov and C.~Vafa,
  ``Geometric singularities and enhanced gauge symmetries,''
  Nucl.\ Phys.\  B {\bf 481}, 215 (1996) 
  \href{http://arxiv.org/abs/hep-th/9605200}
  {arXiv:hep-th/9605200}.



\bibitem{Klemm:1996ts}
A.~Klemm, B.~Lian, S.-S. Roan, and S.-T. Yau, ``Calabi-{Y}au four-folds for
  {M}- and {F}-theory compactifications,'' {\em Nuclear Phys. B} {\bf 518}
  (1998) no.~3, 515--574.
  \href{http://arxiv.org/abs/hep-th/9701023}
{arXiv:hep-th/9701023}.

\bibitem{Klemm:1995tj}
A.~Klemm, W.~Lerche, and P.~Mayr, ``{K3 Fibrations and heterotic type II string
  duality},'' \href{http://dx.doi.org/10.1016/0370-2693(95)00937-G}{{\em Phys.
  Lett.} {\bf B357} (1995)  313--322},
\href{http://arxiv.org/abs/hep-th/9506112}{ arXiv:hep-th/9506112}.

\bibitem{Andreas:1999ty}
B.~Andreas, G.~Curio, and A.~Klemm, ``{Towards the standard model spectrum from
  elliptic Calabi- Yau},''
  \href{http://dx.doi.org/10.1142/S0217751X04018087}{{\em Int. J. Mod. Phys.}
  {\bf A19} (2004)  1987},
\href{http://arxiv.org/abs/hep-th/9903052}{ arXiv:hep-th/9903052}.

\bibitem{Berglund:1998va}
P.~Berglund, A.~Klemm, P.~Mayr, and S.~Theisen, ``{On type IIB vacua with
  varying coupling constant},''
  \href{http://dx.doi.org/10.1016/S0550-3213(99)00420-4}{{\em Nucl. Phys.} {\bf
  B558} (1999)  178--204},
\href{http://arxiv.org/abs/hep-th/9805189}
{ arXiv:hep-th/9805189}.

\bibitem{Klemm:1996hh}
A.~Klemm, P.~Mayr, and C.~Vafa, ``{BPS states of exceptional non-critical
  strings},''
\href{http://arxiv.org/abs/hep-th/9607139}{ arXiv:hep-th/9607139}.

\bibitem{MR2024529}
D.~Husem{\"o}ller, {\em Elliptic curves}, vol.~111 of {\em Graduate Texts in
  Mathematics}.
\newblock Springer-Verlag, New York, second~ed., 2004.
\newblock With appendices by Otto Forster, Ruth Lawrence and Stefan Theisen.

\bibitem{MR85k:14004}
W.~Fulton, {\em Intersection theory}.
\newblock Springer-Verlag, Berlin, 1984.

\bibitem{Andreas:1999ng}
B.~Andreas and G.~Curio, ``{On discrete twist and four-flux in N = 1
  heterotic/F- theory compactifications},'' {\em Adv. Theor. Math. Phys.} {\bf
  3} (1999)  1325--1413,
\href{http://arxiv.org/abs/hep-th/9908193}{ arXiv:hep-th/9908193}.

\bibitem{Jockers:2004yj}
H.~Jockers and J.~Louis, ``{The effective action of D7-branes in N = 1
  Calabi-Yau orientifolds},''
  \href{http://dx.doi.org/10.1016/j.nuclphysb.2004.11.009}{{\em Nucl. Phys.}
  {\bf B705} (2005)  167--211},
\href{http://arxiv.org/abs/hep-th/0409098}{ arXiv:hep-th/0409098}.

\bibitem{Distler:2009ri}
J.~Distler, D.~S. Freed, and G.~W. Moore, ``{Orientifold Precis},''
\href{http://arxiv.org/abs/0906.0795}{ arXiv:0906.0795 [hep-th]}.

\end{thebibliography}
\end{document}